\documentclass[11pt, letterpaper]{article}
\usepackage{fullpage,indentfirst,rotating}
\usepackage{adjustbox}
\usepackage{color}

\usepackage{epsfig}
\usepackage{rotating}
\usepackage{subfigure}
\usepackage{amsfonts}
\usepackage{latexsym}
\usepackage{amsmath}
\usepackage{amssymb}
\usepackage{amsthm}
\usepackage{amstext}
\usepackage{float}
\usepackage[utf8]{inputenc}
\usepackage[english]{babel}
\usepackage{mathbbol}
\usepackage{bigints}
\usepackage{longtable,array}
\usepackage{longtable}
\usepackage{mathtools}
\usepackage{bbm}
\usepackage{booktabs}
\usepackage{newpxmath,newpxtext}
\usepackage{setspace}
\usepackage{xcolor}
\usepackage{blindtext}
\usepackage{sectsty}
\usepackage{todonotes}
\usepackage{bigints}
\usepackage{multicol}
\usepackage{comment}
\usepackage{enumitem}
\usepackage{subfigure}
\sectionfont{\centering}
\subsectionfont{\centering}
\usepackage{tikz}
\usetikzlibrary{tikzmark,calc}
\usetikzlibrary{arrows.meta, patterns}
\usetikzlibrary{shapes.arrows}
\usetikzlibrary{positioning,tikzmark}

\newcounter{alphasect}
\def\alphainsection{0}

\let\oldsection=\section
\def\section{%
  \ifnum\alphainsection=1%
    \addtocounter{alphasect}{1}
  \fi%
\oldsection}%

\renewcommand\thesection{%
  \ifnum\alphainsection=1%
    \Alph{alphasect}
  \else%
    \arabic{section}
  \fi%
}%

\newenvironment{alphasection}{%
  \ifnum\alphainsection=1%
    \errhelp={Let other blocks end at the beginning of the next block.}
    \errmessage{Nested Alpha section not allowed}
  \fi%
  \setcounter{alphasect}{0}
  \def\alphainsection{1}
}{%
  \setcounter{alphasect}{0}
  \def\alphainsection{0}
}%

\usepackage{booktabs}
\usepackage{multirow}
\usepackage{siunitx}

\usepackage[colorlinks=true, linkcolor=blue, citecolor=blue]{hyperref}
\usepackage{soul}
\usepackage[authoryear]{natbib}
\usepackage[top=1 in, bottom=1 in, left=1 in , right=1 in]{geometry}
\allowdisplaybreaks

\newcommand{\ben}{\begin{enumerate}}
\newcommand{\een}{\end{enumerate}}
\newcommand{\bit}{\begin{itemize}}
\newcommand{\eit}{\end{itemize}}

\newcommand{\indep}[0]{\ensuremath{\perp \! \! \! \perp}}

\newcommand{\bft}{\mathbf{t}}

\providecommand{\keywords}[1]{\textit{Keywords:} #1}
\providecommand{\codes}[1]{\textit{JEL Classification:} #1}

\newtheorem{assum}{Assumption}

\newtheorem{prop}{Proposition}
\newtheorem{theorem}{Theorem}
\newtheorem{lemma}{Lemma}
\newtheorem{fact}{Fact}
\theoremstyle{definition}
\newtheorem{remark}{Remark}
\newtheorem{exmp}{Example}
\newtheorem{Algorithm}{Algorithm}

\title{A Nonparametric Test of Heterogeneous Treatment Effects under Interference 
}

\author{
  Julius Owusu \thanks{School of Economics,  University of Bristol,
  \texttt{Julius.owusu@bristol.ac.uk}}\\
}

\vspace{2cm}

\begin{document}
\doublespacing

\maketitle
\thispagestyle{empty}
\begin{abstract}
Statistical inference of heterogeneous treatment effects (HTEs) across predefined subgroups is challenging when units interact because treatment effects may vary by pre-treatment variables, post-treatment exposure variables (that measure the exposure to other units’ treatment statuses), or both. Thus, the conventional HTEs testing procedures may be invalid under interference. In this paper, I develop statistical methods to infer HTEs and disentangle the drivers of treatment effects heterogeneity in populations where units interact. Specifically, I incorporate clustered interference into the potential outcomes model and propose kernel-based test statistics for the null hypotheses of (i) no HTEs by treatment assignment (or post-treatment exposure variables) for all pre-treatment variables values and (ii) no HTEs by pre-treatment variables for all treatment assignment vectors. I recommend a multiple-testing algorithm to disentangle the source of heterogeneity in treatment effects. I prove the asymptotic properties of the proposed test statistics. Finally, I illustrate the application of the test procedures in an empirical setting using an experimental data set from a Chinese weather insurance program.
\newline
\newline
 \codes{C01, C12, C14.}
\end{abstract}

\clearpage
\setcounter{page}{1}
\section{Introduction}\label{Introduction}
The literature on causal inference focuses primarily on identifying and estimating aggregate treatment effects. However, it is increasingly recognized that these aggregate effect metrics, while useful for measuring social welfare, often overlook the variations in treatment effects. These variations are crucial for designing welfare-maximizing treatment assignment rules. For instance, a job search assistance program may show positive aggregate effects on income and welfare. However, it could also lead to significant income and welfare disparities in society, as the impact of the program may vary across individuals.
This observation has sparked a significant and growing research interest in estimating and inferring heterogeneous treatment effects (HTEs). This shift in focus from aggregate treatment effects to HTEs is a significant development in the field of causal inference. The classical approach to infer HTEs involves estimating and comparing conditional average treatment effects (CATEs) of predefined subgroups in a population.
The formal comparison of CATEs requires some hypothesis testing procedure.
\cite{crump2006nonparametric}, \cite{ding2016randomization}, \cite{wager2018estimation}, and  \cite{sant2021nonparametric} are  among the few  existing HTEs testing papers.

This paper introduces two nonparametric tests for inferring HTEs across subgroups while accounting for interference among economic units. According to \cite{cox1958planning}, interference arises when the treatment response of one unit depends on the treatment of other units in the population. In other words, a unit's treatment response is a function of the treatment assignment --- the vector of treatments of all units. Interference may result from physical, virtual, or social connections between members of a population. While the mechanisms behind interference may vary, its presence complicates the inference of HTEs across subgroups. In the presence of interference, treatment effects may depend on pre-treatment variables, post-treatment exposure variables (which capture exposure to other units' treatment statuses), or both.  For example, suppose we find that the effectiveness of a Covid-19 vaccine varies across cities. The vaccine's efficacy for an individual may be influenced by the vaccination rate among their physical contacts. It is, therefore, essential to investigate whether the observed variation in vaccine effectiveness is driven by inherent population differences across cities (e.g., genetic factors) or by differences in vaccination rates, a post-treatment exposure variable.  In this example, variations in vaccination rates may mask the heterogeneity due to natural differences across cities in a standard test that does not control for the vaccination rates.  This highlights that traditional tests for HTEs --- which assume no interference between units --- are likely to produce inaccurate results when interference is present.

To address the challenges of testing HTEs in the presence of interactions between units, I consider a clustered interference setting, where interference is limited to within clusters. 
 In this setting, I introduce kernel-based test statistics for the null hypotheses of (i) constant treatment effects (CTEs) by treatment assignment for all values of the pre-treatment variables and (ii) CTEs by the pre-treatment variables for all treatment assignment vectors. Then, I recommend the \cite{holm1979simple} multiple testing procedure to jointly test the null hypotheses and disentangle the source of heterogeneity in the treatment effects. The proposed test statistics are sums of the weighted $L_1$ norm of the differences in the consistent kernel estimators of CATE that characterize the null hypotheses. 

The proposed testing procedure in this paper provides several practical advantages to policymakers.  First, testing the null hypothesis of CTEs by treatment assignment for all covariates in isolation informs policymakers on whether to scale a program. Suppose we fail to reject the null hypothesis of CTEs across treatment assignments for all values of the pre-treatment variables. In that case, it implies that treatment spillover effects are absent, and a program can be scaled without any negative or positive externalities. Second, the null hypothesis of CTEs by the pre-treatment variables for all treatment assignment vectors in isolation helps to detect HTEs across subgroups defined by pre-treatment variables. Knowledge of a program's treatment effect variation across subgroups can guide its extension to other populations (external validation). Third, jointly testing both null hypotheses helps to disentangle the source of variation in treatment effects. Finding the drivers of the variations in treatment effects in an interconnected human society is crucial in designing welfare-maximizing treatment assignment rules.

The main theoretical results  show that the proposed test statistics have an asymptotically standard normal null distribution. I prove these properties via a poissonization technique formulated by \cite{rosenblatt1975quadratic} and developed by \cite{gine2003bm}. The  technique introduces additional randomness by assuming that the sample size is a Poisson random variable. It enables techniques that exploit Poisson processes' independent increments and infinite divisibility properties. Also, I show that the test statistics have asymptotically valid sizes and power against fixed alternatives and  sequences of local alternatives. Moreover, I propose bootstrap procedures to benchmark the proposed asymptotic-based methods in finite samples. All the asymptotic results are based on a regime where the number of sampled clusters goes to infinity. As I discuss further in Section \ref{Main Asymptotic Results}, an essential feature of this regime is that it permits the safe disregard of within-cluster dependencies for local estimators in large samples.

The main contribution of this paper is twofold. First, I highlight the issue of HTEs testing in the presence of interference. As a solution, I propose testing procedures for HTEs that (i) can accommodate several forms of clustered interference and (ii) can disentangle the sources of heterogeneity. Second, I extend the CATE identification and estimation in the literature to allow for clustered interference. From the technical perspective, I prove the asymptotic properties of the proposed testing procedures in this paper using the modern Poissonization technique of \cite{gine2003bm}.

 The  rest of the paper is organized as follows. I review the existing literature in the following subsection.
 Section \ref{framework} describes the setup, the testing problem, and the test statistics. In Section \ref{Main Asymptotic Results}, I present the main asymptotic properties of the test statistics.  Monte Carlo simulation design and results are in Section \ref{Monte Carlo Simulation}. To illustrate the usage of the proposed tests, in Section \ref{application},  I revisit the Chinese weather insurance policy data set in \cite{cai2015social}  and test for HTEs by a post-treatment variable (treatment ratio) and a pre-treatment variable (the fraction of household income from rice production). My concluding remarks are in Section \ref{Conclusion}. All proofs, useful theorems, lemmas, and other simulation experiments are in the  Appendix.

\subsection{Related Literature}
The nascent literature on the estimation and inference of HTEs continues to grow and spans multiple fields. The present paper falls under the arm of the  literature that studies the inference
of HTEs using tests based on average treatment effects (ATEs) of subgroups defined by pre-treatment variables. \cite*{bitler2006mean} provide an in-depth critique of this approach to testing for HTEs. They argue that heterogeneity of CATEs across subgroups often does not imply individual treatment effect variation unless one assumes constant subgroup treatment effect (CSTE). Nonetheless,  regardless of the CSTE assumption, variations in ATEs of subgroups are crucial in the design of treatment assignment rules where pre-treatment variables are used to set eligibility conditions. In econometrics, tests to detect variation in ATEs across subgroups have been studied by \cite{crump2006nonparametric} and \cite{lee2014multiple}. Both studies abstract from interference and propose nonparametric tests to infer HTEs across predefined subgroups. In contrast, I allow clustered interference and use kernel-based estimators to construct the test statistics. To the best of my knowledge, this is the first paper to test for HTEs under any form of interference.

The problem I consider, testing for HTEs in the presence of clustered interference, is distinct from and complementary to other CATE estimation problems in the literature. \cite{abrevaya2015estimating} and \cite{fan2022estimation} study the identification and kernel-based estimation of CATE for traditional and high dimensional data sets, respectively. By contrast, this paper focuses on testing for HTEs or varying CATEs. Also, the identification strategy and estimation in the current paper account for within-cluster interferences, whereas the above references impose a no-interference assumption. Similarly, \cite{bargagli2020heterogeneous} proposes a Network Causal Tree algorithm that seeks to find and estimate the treatment effect of subpopulations where treatment and spillover effects differ across individual, neighborhood, and network characteristics in clustered network populations. This is again distinct from my goal, which takes the subpopulations of interest as given and studies testing procedures to infer HTEs across these covariates-defined subpopulations.
On the theoretical side, \cite{lee2013testing} and \cite{chang2015nonparametric} establish asymptotic null distributions for the  $L_p$-type functions of kernel-based CATE estimators using the  modern poissonization technique of \cite{gine2003bm}. Allowing for clustered interference in the potential outcomes model requires a  modification of  the estimator of CATE to fit the current framework. In addition, the proposed test statistics in this paper are different, and the theoretical results are  extensions of those in these \cite{chang2015nonparametric}.

Finally, the nonparametric bootstrap algorithms proposed in this paper are similar to existing algorithms in the literature. For instance, \cite{li2009nonparametric} uses a bootstrap algorithm akin to that described in Section \ref{balgorithm 1}  to test for the equality of two density or conditional density functions. Also, \cite{racine1997consistent} employs a residual bootstrap algorithm similar to the wild bootstrap-t procedure proposed in Section \ref{balgorithm 2} to test for the significance of pre-treatment variables in regression models. Despite the similarities, the bootstrap procedures in this paper must account for the within cluster dependencies and, as such, are novelties that fit the current setting.

\section{Framework }\label{framework}
\subsection{Setup} \label{setup}
Consider a setting where units are partitioned into clusters, such as classrooms, villages,  or states. Let $N_c $ denote the number of \textit{finite} units in cluster $c \in \mathcal{I}.$ Suppose that $C$ clusters are randomly drawn from the population of clusters, and let $N$ denote the number of units that make up the observed sample,  i.e., $\sum_{i=1}^{C}N_c= N.$ Furthermore, assume that units in the same cluster interact arbitrarily, leading to clustered or partial interference  \citep{sobel2006randomized}.
Let $T_{ci} \in \{0,1\}$ be the binary treatment condition of unit $i\in \{1,\dots, N\}$ in cluster $c \in \{1, \dots, C\}.$ Treatment  is selected or assigned at the unit level. Thus, different clusters may have different treatment assignment vectors. The treatment vector of the units in cluster $c$ is denoted as $\mathbf{T}_c.$ 
Moreover, we observe the outcome variables $Y_{ci} \in \mathcal{Y}\subset \mathbbm{R},$  and a vector $X_{ci} \in \mathcal{X}\subset \mathbbm{R}^{^{d}}$ of pre-treatment variables. 

Adopting the potential outcomes model of \cite{neyman1923applications} and \cite{rubin1974estimating}, let $Y_{ci}(\mathbf{t}_c)$ represent the potential outcome of unit $i$ in cluster $c$ when the cluster treatment assignment vector $\mathbf{T}_c$ equals  $\mathbf{t}_c.$ Therefore, the potential outcomes of a unit are indexed by her cluster treatment assignment vector.
Effectively, the number of potential outcomes of a unit in cluster $c$ is $2^{^{N_c}}$ (i.e., the number of all possible cluster treatment assignment vectors). Note that in the classical case of no interference, only two potential outcomes exist for each unit. Hence, allowing for (clustered) interference exacerbates the missing data problem of causal inference. As such, a salient element of causal inference in the presence of (clustered) interference is a mapping $\pi(\cdot)$, which summarizes how the (cluster) treatment vectors affect the treatment response outcome. This is commonly referred to as the \textit{exposure mapping} (see \cite{manski2013identification} and  \cite{aronow2017estimating}). Formally, I define  exposure mapping as 
\begin{equation}\label{exposure}
    \pi_{i}: \{0,1 \}^{N_c}  \mapsto \mathbf{\Pi}
\end{equation}
that maps the cluster treatment vector $\mathbf{T}_{c}$  into an exposure variable $\Pi_{ci}:=\pi_{i}(\mathbf{T}_{c}) \in \mathbf{\Pi} \subset \mathbbm{R}$. The definition of $\pi_i$ is application-specific, as it must account for the within-cluster interference structures of  the given problem.

 Next, for a given exposure mapping, I assume that $Y_{ci}(\mathbf{t}_c)=Y_{ci}(t,\pi)$ is the potential outcome for unit $i$ if the cluster treatment vector $\mathbf{t}_c$ is such that $T_{ci} =t$ and  $\pi_{i}(\mathbf{t}_c)=\pi.$ Thus, this assumption asserts that \textit{heterogeneity in treatment effects across treatment assignment vectors is analogous to heterogeneity in treatment effects across exposure variable's values (henceforth exposure values)}. 

 \begin{remark}
It is worth emphasizing that the exposure mapping in \eqref{exposure} depends only on the cluster treatment assignment vectors and  requires no knowledge of the connections between units in the population. Thus, the exposure variable is typically defined at the cluster level --- such as the number of treated units within a given cluster. However, the testing procedures in this paper also accommodate individual-specific exposure variables, including leave-one-out exposure variables --- such as the number of treated units in a cluster excluding the $i^{th}$ unit. As a result, the exposure variable of unit $i$ is a reduced measure of the contagion or direct interference  within  $i$'s cluster.
This is convenient because, in many applied settings, it is difficult to obtain information on links between economic units due to privacy concerns. For instance,  \cite{colpitts2002targeting} reveals that  Canada abandoned a targeting program designed to use individual-level and network information to assign unemployed workers to different job activation programs due to data security concerns. 
 \end{remark}
 The following example provides a plausible specification of exposure mapping defined in \eqref{exposure}.    
 
\begin{exmp}[Distributional clustered interaction]
\cite{manski2013identification} explains that distributional clustered interaction occurs if the outcome of unit $i$ does not depend on the sample size, and it is invariant to permutations of the treatments received by other units in the same cluster.
In essence,  distributional clustered interaction implies that $\pi_i(\mathbf{t}_{c}) \in \{0, 1/N_c, 2/N_c,\dots, 1\}$ represents the  treatment ratio in cluster $c$. 
See \cite{manski2013identification} for other specifications.
 \end{exmp}

I formalize the assumptions that describe the setting as follows.

\begin{assum}[Treatment-invariant clusters]\label{Treatment-invariant neighborhoods}  Let $G$ denote the cluster assignment variable, then for each unit $i,$ $\Pr(G_i|\mathbf{T}_{c})= \Pr(G_i),$ i.e.,
a unit's cluster and  treatment  assignment variables are independent.
\end{assum}
\noindent In other words, Assumption \ref{Treatment-invariant neighborhoods} suggests that the network is a fixed population characteristic, and units do not strategically select into clusters after treatment assignment. 

\begin{assum}[ Clustered Interference]\label{Discrete Network Exposure}
Let $\pi_i(\cdot)$ be an  exposure mapping function, i.e., $\pi_i: \{0,1 \}^{N_c} \mapsto \mathbf{\Pi},$ with $\mathbf{\Pi}=\{\pi_1, \dots,\pi_K:K< C\}$ being a discrete set of finite elements, $\forall c = 1,\dots,C,$ $\forall i = 1,\dots,N_c,$ and $\forall \bft,  \bft' \in \{0,1\}^N$ such that  $t_{_{i}}=t'_{i}$ and  $\pi_i(\bft_{_{c}})=\pi_i(\bft'_{c})$ then $Y_{ci}(\bft) =Y_{ci}(\bft')$
\end{assum}
 \noindent Assumption \ref{Discrete Network Exposure} is embedded with a couple of information and merits a lengthy discussion. First, it restricts the nature of interference, i.e., it imposes intra-cluster interference but no inter-cluster interference. Second, this assumption limits the range of the exposure variable to a discrete finite set of $K<C$ elements; thus, the number of potential outcomes reduces significantly to $2\cdot K$. This part of the assumption also implies that\textit{ units in different clusters can have the same exposure value}, as a result guaranteeing that there are "enough" units for each treatment and exposure value for precise estimation of the CATEs. Finally, Assumption \ref{Discrete Network Exposure} asserts that a unit's exposure to the treatment of other units in her cluster is in a "reduced form," and channels through which interference occurs in clusters are not distinguishable.
Exploring general forms of interference with network information is a possible avenue for future research.


For each $\pi \in \mathbf{\Pi}$  and $x \in \mathcal{X},$ I define  the CATE parameter as\,\, $\tau(x;\pi )\coloneq\mathbbm{E} [Y(1, \pi )|X=x] -\mathbbm{E} [Y(0, \pi )|X=x].$ The  following assumptions are crucial for the identification of the CATEs.
\begin{assum}[Consistency and Unconfoundedness]\label{Unconfoundedness}\,\,\,\,\,\,\,\,\,\,\,\,\,\,\,\,\,\,\,\,\,\,\,\,\,\,\,\,\,\,\,\,\,\,\,\,\,\,\,\,\,\,\,\,\,\,\,\,\,\,\,\,\,\,\,\,\,\,\,\,\,\,\,\,\,\,\,\,\,\,\,\,\,\,\,\,\,\,\,\,\,\,\,\,\,\,\,\,\,\,\,\,\,\,\,\,\,\,\,\,\,\,\,\,\,\,\,\,\,\,\,\,\,\,\,\,\,\,\,\,
(i) There is only a single version of each treatment.\\
(ii) For all $\pi \in \mathbf{\Pi}$, 
\begin{align}
    (T, \Pi)\indep (Y(0,\pi ), Y(1,\pi ))|X
\end{align}
\end{assum}
\begin{assum}[Overlap]\label{Overlap}
For all $\pi \in \mathbf{\Pi}$, $x \in \mathcal{X}$  and for some $\xi>\eta>0$ 
\begin{align}
  \xi<\mathrm{P}(T=1 | X=x )< 1-\xi \,\,\,\text{and}\,\,\,  \eta<\mathrm{P}(T=1, \Pi=\pi | X=x )< \mathrm{P}(T=1 | X=x )-\eta
\end{align}
\end{assum}
Assumptions \ref{Unconfoundedness}(ii) and  \ref{Overlap}  are the extensions of the usual ignorability assumptions imposed on the treatment assignment mechanism under no interference, see \cite{imbens2015causal}.
Assumption \ref{Unconfoundedness}(i) is the standard consistency condition that rules out multiple versions of the same treatment selected or assigned, and \ref{Unconfoundedness}(ii) asserts that conditional on pre-treatment variables, self-selection of \textit{effective treatment}, $(T, \Pi)$ is ruled out. This Assumption  holds when treatment is randomly assigned at the unit level. Assumption \ref{Overlap} is a modified version of the usual overlap or probabilistic assignment condition. It ensures a balance between treated and control units in each subgroup. It is particularly crucial because of the H\'ajek-type estimators of CATE I employ in this paper. 
\begin{prop}\label{identification}
    Suppose Assumptions \ref{Treatment-invariant neighborhoods}-\ref{Overlap} hold, then $\forall \pi \in \mathbf{\Pi} \,\,\,\text{and}\,\, \,x \in \mathcal{X}$
\begin{align}\label{identified CATE}
\tau(x, \pi)
  = \mathbbm{E} [Y|T=1, \Pi=\pi, X=x] -\mathbbm{E} [Y|T=0, \Pi=\pi,X=x].
\end{align}
\end{prop}

\subsection{The Testing Problem}\label{The Testing Problem}
In this section, I formally describe the testing problem. As mentioned above, in the presence of interference, treatment effects may vary by either pre-treatment variables, post-treatment exposure variables, or both. Figure \ref{fig:illustrations} shows some plausible cases in a simple setting where $\mathbf{\Pi}=\{\pi_1,\pi_2\}$ and CATE is linear in a continuous $X.$ Panels (a) shows CTEs by both classes of variables.
In contrast, Panels (b) and (c) show scenarios where one class of variable has HTEs, and the other class has CTEs. Panel (d) depicts the case where both classes of variables have HTEs.
These facts highlight that testing for HTEs in the presence of interference requires testing for heterogeneity across the two classes of variables and the need for methods to disentangle the source of the effect heterogeneity. 

\begin{figure}[h]
\centering
\caption{Treatment Effects Variation by a Continuous Pre-treatment Variable and a Binary Post-treatment Exposure Variable}
\subfigure[CTEs by a pre-treatment variable and an exposure variable.]{
	\begin{tikzpicture}[scale=4.3]
			\draw[thick] (0,0) -- (1.3,0);
			\draw[thick] (0,0) -- (0,1);
			\draw[thick] (0,1/3) -- (1.3,1/3);
            \node[right] at (0.5,0.4) {$\tau(X;\pi_1)=\tau(X;\pi_2)$};
			\draw[-{Latex[width=2mm]}] (1.2,0) -- (1.32,0);
			\node[right] at (1.2,-0.06) {$X$};
			\draw[-{Latex[width=2mm]}] (0,1) -- (0,1.02);
			\node[right] at (-0.34,1.) {$\tau(X;\Pi)$};
		\end{tikzpicture}
  \label{fig:subfig1}
}
\subfigure[CTEs by a pre-treatment variable and HTEs by  an exposure variable. ]{
    	\begin{tikzpicture}[scale=4.3]
			\draw[thick] (0,0) -- (1.3,0);
			\draw[thick] (0,0) -- (0,1);
   			\draw[thick] (0,0.5) -- (1.3,0.5);
            \node[right] at (0.9,0.56) {$\tau(X;\pi_1)$};
			\draw[thick] (0,1/3) -- (1.3,1/3);
            \node[right] at (0.9,0.4) {$\tau(X;\pi_2)$};
			\draw[-{Latex[width=2mm]}] (1.2,0) -- (1.32,0);
			\node[right] at (1.2,-0.06) {$X$};
			\draw[-{Latex[width=2mm]}] (0,1) -- (0,1.02);
			\node[right] at (-0.34,1.) {$\tau(X;\Pi)$};
		\end{tikzpicture}
    \label{fig:subfig2}
}
\subfigure[HTEs by a pre-treatment variable and CTEs by an exposure variable.]{
    	\begin{tikzpicture}[scale=4.3]
			\draw[thick] (0,0) -- (1.3,0);
			\draw[thick] (0,0) -- (0,1);
			\draw[thick] (0,1/3) -- (1.3,0.6);
            \node[rotate=12.8] at (0.9,0.6) {$\tau(X;\pi_1)=\tau(X;\pi_2)$};
			\draw[-{Latex[width=2mm]}] (1.2,0) -- (1.32,0);
			\node[right] at (1.2,-0.06) {$X$};
			\draw[-{Latex[width=2mm]}] (0,1) -- (0,1.02);
			\node[right] at (-0.34,1.) {$\tau(X;\Pi)$};
		\end{tikzpicture}
    \label{fig:subfig3}
}
\subfigure[HTEs by a pre-treatment variable and an exposure variable.]{
    	\begin{tikzpicture}[scale=4.3]
			\draw[thick] (0,0) -- (1.3,0);
			\draw[thick] (0,0) -- (0,1);
   			\draw[thick] (0,0.5) -- (1.3,0.77);
            \node[rotate=13] at (1.1,0.79) {$\tau(X;\pi_1)$};
			\draw[thick] (0,1/3) -- (1.3,0.6);
            \node[rotate=13] at (1.1,0.62) {$\tau(X;\pi_2)$};
            
			\draw[-{Latex[width=2mm]}] (1.2,0) -- (1.32,0);
			\node[right] at (1.2,-0.06) {$X$};
			\draw[-{Latex[width=2mm]}] (0,1) -- (0,1.02);
			\node[right] at (-0.34,1.) {$\tau(X;\Pi)$};
		\end{tikzpicture}
    \label{fig:subfig4}
}
\\
\label{fig:illustrations}
\end{figure}
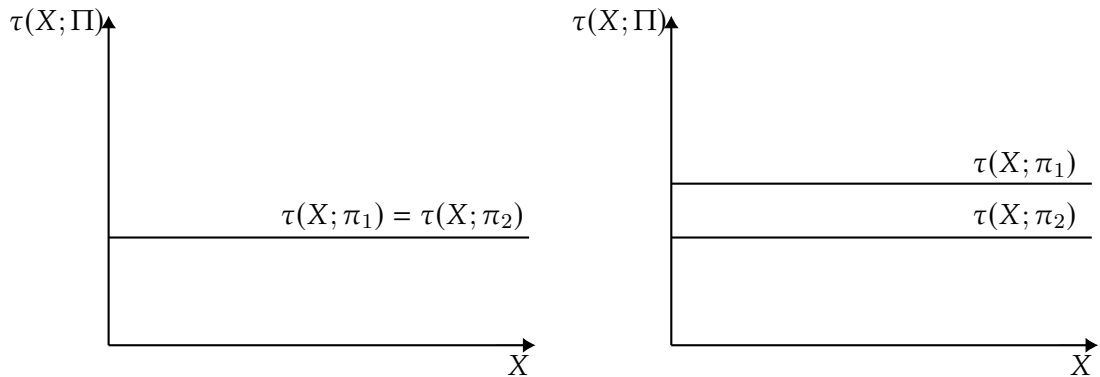
To test for heterogeneity across one class of variables requires controlling for the other class of variables. Formally, I consider the following null hypotheses. The first one is the null hypothesis of constant treatment effects (CTEs) by an exposure variable while controlling for pre-treatment variables:
\begin{align}\label{first  null hypothesis}
\begin{split}
&H^{\Pi}_{0}:
      \,\, \forall x \in \mathcal{X},\,\, \forall \pi, \pi' \in \mathbf{\Pi},\,\,     \tau(x;\pi)=\tau(x;\pi'), \,\, \\
\end{split}
\end{align}
against the alternative hypothesis of HTEs by the exposure variable: 
\begin{align}\label{first alternative hypothesis}
\begin{split}
&H^{\Pi}_{1}:
      \,\, \exists x \in \mathcal{X},\,\,  \exists\, \pi, \pi' \in \mathbf{\Pi},  \,\, \tau(x;\pi)\neq\tau(x;\pi'). \,\, \\
\end{split}
\end{align}
The null hypothesis \eqref{first  null hypothesis} is vital in answering several essential questions in program evaluation. For example,  testing \eqref{first  null hypothesis} when exposure is treatment ratios helps to determine whether a program should be scaled or not; rejecting the null hypothesis implies treatment spillover effects exist and program effects may vary by program scale. Although  \eqref{first  null hypothesis} resembles Hypothesis 2 in \cite{athey2018exact}, they are not the same. The null hypothesis \eqref{first  null hypothesis} is a restriction on the treatment effect that there are no indirect or spillover treatment effects. In contrast, Hypothesis 2 in \cite{athey2018exact} restricts the outcome of no spillovers or interference. Failure to reject Hypothesis 2 in \cite{athey2018exact} implies that we fail to reject \eqref{first  null hypothesis}; however, the converse is not valid.

\begin{remark}
Generally, the null $H_0^{\Pi}$ resembles the null of no or constant treatment effects for all covariates $X$ in settings with multiple treatment effects. Given $K$ possible treatments, to test for no treatment effects at all covariates points, one may have to test for no effects for all pairwise treatments across covariates. The main difference between the setting in this paper and the multiple treatment environment is the possibility of within-cluster dependencies. However, as I thoroughly discuss in Section \ref{Main Asymptotic Results}, such dependencies can be safely ignored in large samples because the CATE function is estimated using the local constant estimators. Hence, with a slight modification of the estimands of interest, the testing procedure developed in this paper for $H_0^{\Pi},$ can be adapted to test for the null of no treatment effects for all covariates $X$ in settings with multiple treatment effects.
\end{remark}

The second hypothesis of interest concerns the null hypothesis of CTEs by pre-treatment variables while controlling for the post-treatment exposure variable:
\begin{align}\label{second  null hypothesis}
\begin{split}
&H^{X}_{0}:
      \,\, \forall \pi \in \mathbf{\Pi}, \forall x, x' \in \mathcal{X},\,\,    \tau(x;\pi)=\tau(x';\pi), \,\, 
\end{split}
\end{align}
against the alternative hypothesis of HTEs by pre-treatment variables:
\begin{align}\label{second alternative hypothesis}
\begin{split}
&H^{X}_{1}:
      \,\, \exists \pi \in \mathbf{\Pi},  \exists \,x, x' \in \mathcal{X}, \,\,   \tau(x;\pi)\neq \tau(x';\pi). \,\, 
\end{split}
\end{align}
Tests for the null hypothesis similar to \eqref{second  null hypothesis} exist in the literature; see \cite{crump2006nonparametric}. However, this paper is the first to allow some form of interference.  
Testing \eqref{second  null hypothesis} is critical in extending existing programs to new populations; rejecting this hypothesis implies that treatment effects vary across subgroups defined by the pre-treatment variables. Therefore, a policymaker should expect different aggregate effects if the program is extended to a new population.

It is worth discussing the importance of distinguishing between the exposure variable (i.e., $\Pi$) and the pre-treatment variables in this framework. First, note that $\Pi$ is a causal variable and, as such, is used to index the potential outcomes. In contrast, the pre-treatment variables are invariant to treatment assignment, i.e., they are a priori variables. Therefore, this distinction is necessary in the potential outcome framework for identification. Secondly, distinguishing between the pre-treatment and treatment exposure variables is imperative to determine the drivers of the heterogeneity in treatment effects. It has vital implications for policymakers. For instance, if the heterogeneity is solely due to the treatment exposure variable, policymakers may have complete control over the distribution of treatment effects in the population via the treatment assignment mechanism. See \cite{han2022statistical} for a statistical treatment assignment rule for homogeneous and heterogeneous populations in the presence of social interaction.

\begin{remark} \label{conclusions 1} An important feature of $H^{\Pi}_{0}$ and $H^{X}_{0}$  is that testing each hypothesis in isolation may not clarify the source of heterogeneity in treatment effects. For example, failure to reject the null hypothesis $H^{\Pi}_{0}$ implies either  \textit{CTEs by exposure} and \textit{HTEs by pre-treatment variables} or  \textit{CTEs by both exposure and pre-treatment variables}. Conversely, failure to reject $H_0^X$ implies either \textit{ CTEs  by the pre-treatment variables} and \textit{HTEs by the exposure} or  \textit{CTEs by both the exposure and pre-treatment variables}. This highlights the necessity of testing both hypotheses simultaneously. In the following subsection, I propose a multiple testing procedure designed to conduct such simultaneous testing.
\end{remark}

\subsection{Test Statistics}\label{Test Statistics}
The proposed test statistic for the null hypothesis $H_0^\Pi$ in \eqref{first  null hypothesis} is  
 \begin{align} 
     \hat{T}_1 \coloneqq  \int_{\mathcal{X}} \sum_{k =1}^K \sum_{j >k}^K \left\{\sqrt{C}|\hat{\tau}(x;\pi_k)-\hat{\tau}(x;\pi_j)|\right\} \hat{w}(x, \pi_k, \pi_j)dx,
 \end{align}
where $\hat{\tau}(x;\pi)$ is a uniform consistent estimator of $\tau(x; \pi)$ and  $\hat{w}(x, \pi, \pi')$ is the uniform consistent estimator of the inverse standard error of $\sqrt{Ch^{d}}(\hat{\tau}(x;\pi)-\hat{\tau}(x;\pi'))$ which is defined as $w(x, \pi, \pi'):= 1/\sqrt{\rho_2(x, \pi)+ \rho_2(x, \pi')-2Ch^{d}Cov(\hat{\tau}(x;\pi),\hat{\tau}(x;\pi')},$ \footnote{Since the network is clustered, $w(x, \pi, \pi'):= 1/\sqrt{\rho_2(x, \pi)+ \rho_2(x, \pi')}$ when the exposure variable is at the cluster level.} with, $\rho_2(x, \pi)$ representing the variance of  $\sqrt{Ch^{d}}\hat{\tau}(x;\pi).$ Also, note that  $d$ is the dimension of $X,$ $h$ is the bandwidth, and recall that  $C$ is the number of sampled clusters.

I propose a nonparametric kernel-based  estimator of $\tau(x; \pi)$ defined as
\begin{align}\label{estimator}
\begin{split}
   \hat{\tau}(x;\pi):=&\frac{1}{Ch^d}\sum_{c=1}^{C} \frac{1}{N_c}\sum_{i=1}^{N_c}  Y_{ci} \cdot \mathbbm{1}(\Pi_{ci}=\pi)\hat{\phi}(T_{ci},x,\pi) K\left(\frac{x-X_{ci}}{h}\right), 
\end{split}   
\end{align}
where
$$\hat{\phi}(T_{ci},x,\pi):=\frac{T_{ci}}{\hat{P}_1(x;\pi)}-\frac{(1-T_{ci})}{\hat{P}_0(x;\pi)},$$
with $$\hat{P}_t(x;\pi):= \frac{1}{Ch^d}\sum_{c=1}^{C}  \frac{1}{N_c}\sum_{i=1}^{N_c} \mathbbm{1}(\Pi_{ci}=\pi)\mathbbm{1}(T_{ci}=t)K\left(\frac{x-X_{ci}}{h}\right),\,\,\,t=0, 1, $$
and $K(\cdot)$ being a $d$-dimensional kernel function.

The test statistic for the null hypothesis $H_0^{X}$ in  \eqref{second  null hypothesis} is defined as
 \begin{align}
     \hat{T}_2 \coloneqq &  \int_{\{x: x \in \mathcal{X}\}}\int_{\{x' :x' \in \mathcal{X}: x\neq x'\}} \sum_{k =1}^K \left\{\sqrt{C}|\hat{\tau}(x;\pi_k)-\hat{\tau}(x';\pi_k)|\right\} \frac{\hat{w}(x, x', \pi_k)}{2}dxdx' \label{teststat 2}\\ 
     =& \int_{\{(x,x') \in \mathcal{X}^2:x<x' \}} \sum_{k =1}^K \left\{\sqrt{C}|\hat{\tau}(x;\pi_k)-\hat{\tau}(x';\pi_k)|\right\} \hat{w}(x, x', \pi_k)dxdx'  
 \end{align}
 where, $\hat{w}(x, x', \pi)$  is a consistent kernel estimator of ${w}(x, x', \pi),$ the inverse standard error of  $\sqrt{Ch^{d}}\cdot(\hat{\tau}(x;\pi)-\hat{\tau}(x';\pi)).$ Mathematically, $w(x, x', \pi):= 1/\sqrt{\rho_2(x, x', \pi)}$ with 
 $\rho_2(x, x', \pi):=\rho_2(x, \pi)+ \rho_2(x', \pi)-2Ch^{d}Cov(\hat{\tau}(x;\pi),\hat{\tau}(x';\pi)).$
 Throughout the paper, I use the first form of $\hat{T}_2$ in \eqref{teststat 2}; however, to simplify notation, I let $\int_{\{x: x \in \mathcal{X}\}}\int_{\{x' :x' \in \mathcal{X}: x\neq x'\}}$ be equal to $\int_{ \mathcal{X}}\int_{  \mathcal{X}}.$ This should not be interpreted as integrating over the product of the covariate space as the set $\{x, x':x\neq x',  x \in \mathcal{X}, x'\in \mathcal{X}\}$ may have non-zero measure.

If we test the null hypotheses $H_0^\Pi$ and $H_0^X$ simultaneously using a multiple testing procedure (MTP), we can disentangle the source of heterogeneity of treatment effects. If we fail to reject both null hypotheses, it implies CTEs by both variable classes. If we reject the null hypothesis $H_0^\Pi$ but fail to reject $H_0^X$, it suggests HTEs by exposure variable and CTEs by pre-treatment variables. In contrast, if we fail to reject the null hypothesis $H_0^\Pi$ but reject  $H_0^X$, this implies CTEs by exposure variable and HTEs by pre-treatment variables. Finally, if we reject both null hypotheses, then it suggests HTEs for both variable classes.   Thus, as mentioned in Remark \ref{conclusions 1}, implementing both tests in a multiple-testing framework is imperative to disentangle the source of treatment effects heterogeneity. 

To control the probability of rejecting at least one null hypothesis, given that they are both true --- referred to as the family-wise error rate (FWER) --- I recommend the step-wise multiple testing procedure of \cite{holm1979simple}. 
Let the $p_{d_1} \leq p_{d_2}$ be the ordered $p$-values, with corresponding null hypotheses $H_{0,{d_1}}$ and $H_{0,{d_2}}.$ Then the Holm step-down algorithm is as follows:
\begin{Algorithm}[Holm Procedure]\
\begin{enumerate}
    \item  If $p_{d_1}>\alpha/2$ fail to reject both $H_{0,{d_1}}$ and $H_{0,{d_2}}$ and stop. If $p_{d_1}\leq \alpha/2$ \,reject $H_{0,{d_1}}$ and test $H_{0,{d_2}}$ at level $\alpha.$
    \item If $p_{d_1}\leq \alpha/2$ but $p_{d_2}>\alpha,$ fail to reject $H_{0,{d_2}}$ and stop. If $p_{d_1}\leq\alpha/2$ and $p_{d_2}\leq\alpha,$ reject $H_{0,{d_2}}.$ 
\end{enumerate}

\end{Algorithm}
 Several similar adjustments for multiple testing exist, but I opt for Holm's method due to its computational simplicity and robustness to dependencies between the two test statistics. See \cite{lehmann2022multiple} for a comparison between Holm's method and other MTPs.

\section{ Asymptotic Results }\label{Main Asymptotic Results}
 In this section, I discuss the asymptotic properties of the proposed test statistics $\hat{T}_1$ and $\hat{T}_2$ when the null hypotheses are true and false. As mentioned in the introduction, I adopt an asymptotic regime where the number of sampled clusters $C$ goes to infinity while the units with each cluster are fixed.

I show that appropriate studentized versions of the test statistics $\hat{T}_1$ and $\hat{T}_2$ converge to the standard normal  distribution under the null hypotheses. Two forms of dependencies complicate the derivation of these asymptotic properties. The first is the within-cluster dependency among observations due to interference. Secondly, the terms in the combinatorial sum or integral may also be correlated for each test statistic.
 
\cite{lin2000nonparametric} show that when cluster sizes are finite, the most asymptotically efficient kernel estimator of the CATEs is obtained by completely ignoring the within-cluster correlation. This is commonly called the \textit{working independence approach}. The intuition is that, for clusters of fixed sizes, the probability that two or more observations from the same cluster having significant kernel weights approaches zero as the bandwidth shrinks to zero (See \cite{wang2003marginal} for more details). 

Thus, the working independence approach provides a useful asymptotic method for handling the within-cluster correlations when the number of  bounded clusters is large. In essence,  it allows one to view observations --- used to estimate the CATEs --- as mutually independent in the asymptotics. Based on the foregoing argument and assuming, without loss of generality, that the clusters are of equal size $N_0,$ I rewrite the CATE kernel estimator as 
\begin{align*}
\begin{split}
   \hat{\tau}(x;\pi):=&\frac{1}{Nh^d}\sum_{i=1}^{N} Y_i \cdot \mathbbm{1}(\Pi_i=\pi)\hat{\phi}(T_i,x,\pi) K\left(\frac{x-X_i}{h}\right), 
\end{split}   
\end{align*}
where
$$\hat{\phi}(T_i,x,\pi):=\frac{T_i}{\hat{P}_1(x;\pi)}-\frac{(1-T_i)}{\hat{P}_0(x;\pi)}, \,\,\,\,\,\,\, N=CN_0,$$
with $$\hat{P}_t(x;\pi):= \frac{1}{Nh^d}\sum_{i=1}^{N} \mathbbm{1}(\Pi_i=\pi)\mathbbm{1}(T_i=t)K\left(\frac{x-X_i}{h}\right),\,\,\,t=0, 1. $$
Furthermore, the  kernel estimator of $w(x, \pi, \pi')$ is defined  as 
 $$\hat{w}(x, \pi, \pi'):= \frac{1}{\sqrt{\hat{\rho}_2(x, \pi)+ \hat{\rho}_2(x, \pi')}},$$
where
$\hat{\rho}_2(x, \pi)$ is the  kernel estimator of $\rho_2(x, \pi)$ (ignoring all covariance terms)  defined as 
$$\hat{\rho}_2(x, \pi):=(\hat{\mu}_1(x, \pi)-\hat{\mu}_2(x,\pi))\cdot\int K(\xi)^2 d\xi,$$
and
$$\hat{\mu}_1(x, \pi):=\frac{1}{Nh^d}\sum_{t \in \{0,1\}}\sum_{i=1}^N\frac{Y_i^2\mathbbm{1}(T_i=t)\mathbbm{1}(\Pi_i=\pi)K(\frac{x-X_i}{h})}{\hat{P}_t^2(x;\pi)},$$
and 
$$\hat{\mu}_2(x, \pi):=\frac{1}{N^2h^{2d}}\sum_{t \in \{0,1\}}\sum_{j=1}^N\sum_{i=1}^N\frac{Y_jY_i\mathbbm{1}(T_j=t)\mathbbm{1}(T_i=t)\mathbbm{1}(\Pi_i=\pi)\mathbbm{1}(\Pi_j=\pi)K(\frac{x-X_j}{h})K(\frac{x-X_i}{h})}{\hat{P}_t^3(x;\pi)}.$$
The asymptotic results for non-equal cluster sizes are a straightforward extension at the cost of additional notation. Before I state the asymptotic results, the following  regularity conditions are required. 
 \begin{assum} \label{regularity conditions}
 (a) The joint distribution of $(Y, X) \in \mathcal{Y}\times \mathcal{X}$   is absolutely continuous with respect to the Lebesgue measure; (b) the probability density function $f$ of $X$ is continuously differentiable almost everywhere;  (c) $\rho_2(\cdot, \pi)$ is strictly positive and continuous almost everywhere on $\mathcal{W}_X, \forall \pi \in \mathbf{\Pi}$, where $\mathcal{W}_X$ is a compact subset of $\mathcal{X}$; (d) $K$ is a product kernel function, i.e., $K(u) = \Pi_{j=1}^d K_j(u_j),u = (u_1,\dots, u_d)$, with each $K_j :\mathbbm{R} \mapsto \mathbbm{R}, j = 1,\dots, d$,
satisfying that $K_j$ is an $s$-order kernel function with support $\{u \in \mathbbm{R} : |u|\leq 0.5\}$, symmetric
around zero, bounded, and is of bounded variation, and integrates to 1, where $s$ is an integer that
satisfies $s > 1.5d$;  (e) as functions of $x$, $\mathbbm{E}[Y|X = x, T=t, \Pi= \pi]$, $f(x)$, $p_t(x, \pi)$ for $t = 0, 1$
are $s$-times continuously differentiable almost everywhere for each $\pi \in \mathbf{\Pi}$ with uniformly bounded derivatives; (f) $\sup_{x \in \mathcal{W}_X} \mathbbm{E}[|Y|^3|X = x,T=t, \Pi= \pi] < \infty$ for $t = 0, 1$ and $\pi \in \mathbf{\Pi}$; (g) the bandwidth satisfies $Ch^{2s} \to 0$, $Ch^{3d} \to \infty$ and $(Ch^{2d})^{1/2}/ log C   \to \infty$, where $s > 1.5d$; (h) For each $\pi, \pi' \in \boldsymbol{\Pi},$ $\sup_{x \in \mathcal{W}_X}  |\hat{w} (x, \pi,\pi') - {w} (x, \pi, \pi')| = o_p(h^{d/2})$\,\, and\,\, 
$\sup_{(x, x') \in \mathcal{W}^2_X } |\hat{w} (x, x'\pi) - {w} (x, x'\pi)| =o_p(h^{d/2}).$  
 \end{assum}
These conditions are standard in the kernel estimation literature (see  \cite{lee2013testing}, and \citet[p. 315]{chang2015nonparametric}).
 Assumptions \ref{regularity conditions}(a) and (b) are unnecessary for the asymptotic results. They are convenient assumptions imposing continuity on $X$ and $Y$  that help to present my main results.   Assumption \ref{regularity conditions}(c) ensures that the inverse standard error weight function is continuous and well-defined within a compact subset of $\mathcal{X}.$ Assumption \ref{regularity conditions}(d) imposes conditions on the kernel function.  Assumption \ref{regularity conditions}(e) and (f) imposes restrictions on the underlying true data-generating process to ensure smooth and finite moments. Assumption \ref{regularity conditions}(g) imposes standard restrictions on the choice of bandwidth. Finally, Assumption \ref{regularity conditions}(h) ensures that the estimated weight functions are uniformly consistent. 

It is worth emphasizing that for Assumption \ref{regularity conditions}(g), the bandwidth depends on the number of clusters $C.$ For instance, if $d=1$ and a second order kernel is used (i.e., $s=2$) then the bandwidth has to be of the form $h=\kappa_0C^{-\lambda},$ where $\kappa_0$ is a positive constant and $1/4<\lambda<1/3.$  However, since $N=CN_0$ and $N_0$ is fixed in asymptotics, the restrictions in  Assumption \ref{regularity conditions}(g) can be written in terms of $N.$

\subsection{Asymptotic Null Distribution  and Properties of the Test Statistics}
Let's focus on the test statistic $\hat{T}_1.$ Rewriting it in the form 
     $\hat{T}_1 = \sum_{k =1}^K \sum_{j =1}^K 2^{-1}\int_{\mathcal{X}}  \{\sqrt{N}|\hat{\tau}(x;\pi_k)-\hat{\tau}(x;\pi_j)|\} \cdot \allowbreak \hat{w}(x, \pi_k, \pi_j)dx,$
it resembles the test statistic $\hat{D}:=\int_{\mathcal{X}}\sqrt{N}|\hat{\tau}(x)|\hat{w}(x)dx$\footnote{ $\hat{\tau}(x)$ and $\hat{w}(x)$ are the CATE and inverse scaled standard error kernel estimators defined similar to $\hat{\tau}(x;\pi_k)$ and $\hat{w}(x, \pi_k, \pi_j)$  respectively.} proposed by \cite{chang2015nonparametric} under no-interference. They show that a studentized version of $\hat{D}$ converges to the standard normal distribution. Under the working independence approach, $\hat{T}_1$ is the sum of dependent random variables $\Big\{2^{-1}\int_{\mathcal{X}} \left(\sqrt{N}|\hat{\tau}(x;\pi_k)-\hat{\tau}(x;\pi_j)|\right) \hat{w}(x, \pi_k, \pi_j): \pi_j,\pi_k \in \mathbf{\Pi}\Big\}.$  Using the poissonization technique in  \cite{gine2003bm}, I show that each of these random variables is asymptotically normal. See the details in Appendix \ref{Proofs of Lemmas and Theorems}. Thus, as $C$ tends to infinity, $\hat{T}_1$ is the sum of $K(K-1)/2$  dependent normal random variables.
It suffices to find the pairwise asymptotic covariances between terms in the sum to obtain the asymptotic normality result of  $\hat{T}_1.$

To derive the formal asymptotic results, I studentized the test statistics $\hat{T}_1$ with its mean (bias) under the null and standard error.  
 The asymptotic bias\footnote{See Appendix \ref{bias and var} for the derivation.} of $\hat{T}_1$ is 
 \begin{align}\label{asym bias}
     a_{_1}\coloneqq &h^{\frac{-d}{2}}\cdot\mathbbm{E}|\mathbbm{Z}_1|\cdot \frac{K(K-1)}{2}\cdot \int_{\mathcal{X}} dx, 
 \end{align}
 where $\mathbbm{Z}_1$ is a standard normal random variable.
 This diverging asymptotic bias, solely determined by the data-generating process (DGP) through the bounds of $X,$ can be exactly computed.

Define  $\hat{\Gamma}(x;\pi,\pi'):= \hat{\tau}(x;\pi)-\hat{\tau}(x;\pi'),$  the asymptotic variance\footnote{See Appendix \ref{bias and var} for the derivation.} of  $\hat{T}_1$ is
\begin{equation} \label{asym variance}
    \sigma^2_{1}:=\int_1 \mathrm{Cov}\left(\left|\sqrt{1-\rho(x,t, \pi_i, \pi_j, \pi_k, \pi_l)^2}\mathbbm{Z}_1 +\rho(x,t, \pi_i, \pi_j, \pi_k, \pi_l)\mathbbm{Z}_2\right|,\left|\mathbbm{Z}_2\right|\right) dxdt,
\end{equation}
where $\int_1= \int_{\mathbbm{R}^{d}}\int_{T_0}\sum_{i=1}^K  \sum_{j=1}^K \sum_{k=1}^K \sum_{l=1}^K,$ $T_0=[-1,1]^d$  and  $\rho(x,t, \pi_i, \pi_j, \pi_k, \pi_l)$ is the \textit{unknown} correlation between  $\hat{\Gamma}(x;\pi_i,\pi_j)/\sqrt{\hat{\rho}_2(x, \pi_i, \pi_j)}$ and $\hat{\Gamma}(x;\pi_k,\pi_l)/\sqrt{\hat{\rho}_2(x, \pi_k, \pi_l)}$, (i.e., $\rho(x,t, \pi_i, \pi_j, \pi_k, \pi_l) =
\mathrm{Corr}(\hat{\Gamma}(x;\pi_i,\pi_j)/\sqrt{\hat{\rho}_2(x, \pi_i, \pi_j)},\hat{\Gamma}(x;\pi_k,\pi_l)/\sqrt{\hat{\rho}_2(x, \pi_k, \pi_l)})$).  $\mathbbm{Z}_1$ and $\mathbbm{Z}_2$ are mutually independent standard normal random variables.
A consistent kernel estimator of $\rho(x,t, \pi_i, \pi_j, \pi_k, \pi_l)$ is defined as
\begin{align}
\begin{split}
   \hat{\rho}(x, t,\pi_i,\pi_j,\pi_k,\pi_l)=&  
   \begin{cases} 
\frac{\int K(\xi)K(\xi+t)d\xi}{\int K(\xi)^2 d\xi} &\,\, \text{if}\,\, i=k \,\& \,j=l \\
\frac{-\int K(\xi)K(\xi+t)d\xi}{\int K(\xi)^2 d\xi} &\,\, \text{if}\,\, i=l \,\& \,j=k \\
 \frac{-\hat{\rho}_2(x, \pi_j)}{\sqrt{\hat{\rho}_2(x;\pi_i,\pi_j)}\sqrt{\hat{\rho}_2(x;\pi_k,\pi_l)}}\cdot \frac{\int K(\xi)K(\xi+t)d\xi}{\int K(\xi)^2 d\xi}& \,\, \text{if}\,\, j=k \,\& \,i\neq l \\
 \frac{-\hat{\rho}_2(x, \pi_i)}{\sqrt{\hat{\rho}_2(x;\pi_i,\pi_j)}\sqrt{\hat{\rho}_2(x;\pi_k,\pi_l)}}\cdot \frac{\int K(\xi)K(\xi+t)d\xi}{\int K(\xi)^2 d\xi}& \,\, \text{if}\,\, j\neq k \,\& \, i=l\\
 \frac{\hat{\rho}_2(x, \pi_j)}{\sqrt{\hat{\rho}_2(x;\pi_i,\pi_j)}\sqrt{\hat{\rho}_2(x;\pi_k,\pi_l)}}\cdot \frac{\int K(\xi)K(\xi+t)d\xi}{\int K(\xi)^2 d\xi}& \,\, \text{if}\,\, j=l  \,\& \, i\neq k\\
 \frac{\hat{\rho}_2(x, \pi_i)}{\sqrt{\hat{\rho}_2(x;\pi_i,\pi_j)}\sqrt{\hat{\rho}_2(x;\pi_k,\pi_l)}}\cdot \frac{\int K(\xi)K(\xi+t)d\xi}{\int K(\xi)^2 d\xi}&  \,\, \text{if}\,\,j\neq l  \,\& \, i= k\\
 0 & \,\, \text{otherwise},\,\,
    \end{cases} 
\end{split}
\end{align}
where $\hat{\rho}_2(x, \pi, \pi')\coloneqq \hat{\rho}_2(x, \pi) +\hat{\rho}_2(x, \pi').$
Plug  $\hat{\rho}(x, t,\pi_i,\pi_j,\pi_k,\pi_l)$ into the right hand side of (\ref{asym variance}), and  obtain the asymptotic variance estimator
\begin{equation} \label{estimated asym var}
\hat{\sigma}^2_{1}\coloneqq \int_1 \mathrm{Cov}\left(\left|\sqrt{1-\hat{\rho}(x,t, \pi_i, \pi_j, \pi_k, \pi_l)^2}\mathbbm{Z}_1 +\hat{\rho}(x,t, \pi_i, \pi_j, \pi_k, \pi_l)\mathbbm{Z}_2\right|,\left|\mathbbm{Z}_2\right|\right)dxdt. 
\end{equation}
The studentized version of the test statistic $\hat{T}_1$ is defined as 
$$\hat{S}_{1}\coloneqq\frac{\hat{T}_{1}-a_{_1}}{\hat{\sigma}_{1}}.$$

Similarly,  the studentized version of the test statistic $\hat{T}_2$ is defined as 
$$\hat{S}_{2}\coloneqq\frac{\hat{T}_{2}-a_{_2}}{\hat{\sigma}_{2}},$$
where 
\begin{align*}
   a_{_{2}}\coloneqq&h^{\frac{-d}{2}}\cdot\mathbbm{E}|\mathbbm{Z}_1|\cdot \frac{K}{2}\cdot \int_{\mathcal{X}} \int_{\mathcal{X}}dx dx',
\end{align*} and 
\begin{equation}\label{variance 2}
 \hat{\sigma}^2_{2}\coloneqq \int_{2'} \frac{\hat{\rho}_2(x,\pi_k)\cdot\mathrm{Cov}\left(\left|\sqrt{1-\rho(x, t, \pi_k)^2}\mathbbm{Z}_1 +\rho(x, t, \pi_k)\mathbbm{Z}_2\right|,\left|\mathbbm{Z}_2\right|\right) }{\sqrt{(\hat{\rho}_2(x,x',\pi_k)) (\hat{\rho}_2(x,x'',\pi_k))}}  dxdx'dx''dt,
\end{equation}
where $\rho(x, t, \pi_k):= (\int K(\xi)K(\xi+t)d\xi)/(\int K(\xi)^2 d\xi),$ $\hat{\rho}_2(x,x',\pi_k)$ is the plug-in estimator of ${\rho}_2(x,x',\pi_k)$ defined as  $\hat{\rho}_2(x,x',\pi_k):= \hat\rho_2(x, \pi)+ \hat\rho_2(x', \pi) -2 \hat \rho_2(x, \pi)\mathbbm{1}(t \in [-1,1]^d)\cdot\mathrm{Cov}((1-\rho(x, t, \pi_k)^2)^{1/2}\mathbbm{Z}_1 +\rho(x, t, \pi_k)\mathbbm{Z}_2,\mathbbm{Z}_2),$ and 
$\int_{2'}\coloneqq\int_{[-1,1]^{d}}\int_{\mathbbm{R}^{d}} \int_{\mathbbm{R}^{d}}\int_{\mathbbm{R}^{d}}\sum_{k=1}^K.$ 
 \footnote{In the case where errors are homoskedastic, $\hat{\rho}_2(x,\pi_k) \cdot \left((\hat{\rho}_2(x,\pi_k)+\hat{\rho}_2(x',\pi_k)) (\hat{\rho}_2(x,\pi_k)+\hat{\rho}_2(x'',\pi_k))\right)^{-1/2}$ in \eqref{variance 2} simplifies as $\hat f(x,\pi_k)^{-1}\cdot \{(\hat f(x,\pi_k)^{-1}+\hat f(x',\pi_k)^{-1})(\hat f(x,\pi_k)^{-1}+\hat f(x'',\pi_k)^{-1})\}^{1/2}$  with $\hat f(x,\pi_k)$ the estimated joint density.}
It is worth to note that the  asymptotic inverse scaled standard error to compute $\hat S_2$ also becomes $$w(x, x', \pi)=\frac{1}{\sqrt{\rho_2(x, \pi)+ \rho_2(x', \pi) -2 \rho_2(x, \pi)\mathbbm{1}(t \in [-1,1]^d)\cdot\mathrm{Cov}((1-\rho(x, t, \pi_k)^2)^{1/2}\mathbbm{Z}_1 +\rho(x, t, \pi_k)\mathbbm{Z}_2,\mathbbm{Z}_2) }},$$ where $t=(x-x')/h$.\footnote{$w(x, x', \pi)$ is well defined  since $Ch^{d}Cov(\hat{\tau}(x;\pi),\hat{\tau}(x';\pi))\to \rho_2(x, \pi)\mathbbm{1}(t \in [-1,1]^d)\cdot\mathrm{Cov}((1-\rho(x, t, \pi_k)^2)^{1/2}\mathbbm{Z}_1 +\rho(x, t, \pi_k)\mathbbm{Z}_2,\mathbbm{Z}_2)$ after applying a change of variable. Thus, the argument of the square root in the denominator is positive.}   
As such, I obtain the sample analog of $w(x, x', \pi)$  by plugging in the estimators of $\rho_2(x, \pi) $ for all $x \in \mathcal{X}.$ See Appendix \ref{bias and var} for the derivation of $a_2$ and $\hat{\sigma}^2_{2}.$

\begin{prop} \label{normality}
 Let  Assumptions \ref{Treatment-invariant neighborhoods} - 
 \ref{regularity conditions} hold, then under the \\  
 (i)  null hypothesis \eqref{first  null hypothesis} $\hat{S}_1$ converges to the standard normal distribution, i.e., $\hat{S}_{1}\to N(0,1)$ as $C\to \infty$;\\
(ii) null hypothesis \eqref{second  null hypothesis} $\hat{S}_2$ converges to the standard normal distribution, i.e., $\hat{S}_{2}\to N(0,1)$ as $C\to \infty.$
\end{prop}
\noindent Proposition \ref{normality} suggests that critical values of the tests can be obtained  from the standard normal distribution. This proposition forms an integral part  of the asymptotic properties that follow. 

\begin{theorem} \label{valid size}
Let  Assumptions \ref{Treatment-invariant neighborhoods} - 
 \ref{regularity conditions} hold, then under the \\
 (i) null hypothesis \eqref{first  null hypothesis} 
$$\lim_{C\to \infty} \Pr(\hat{S}_1>z_{1-\alpha}) = \alpha;$$\\
 (ii) null hypothesis \eqref{second  null hypothesis} 
$$\lim_{C\to \infty} \Pr(\hat{S}_2>z_{1-\alpha}) = \alpha$$
\end{theorem}
Theorem \ref{valid size} shows that the test statistics $\hat{S}_1$ and $\hat{S}_2$ have correct sizes asymptotically under $H^{\Pi}_{0}$ and $H^{X}_{0}$ respectively. Hence, the following one-sided decision rule suffices: For $j=1,2$  reject the null hypothesis if $\hat{S}_j>z_{1-\alpha},$ where $z_{1-\alpha},  \alpha \in [0,1]$ is the $(1-\alpha)^{th}$ quantile (critical value) obtained from the standard normal distribution.

\subsection{Power properties of the Test statistics}
In this subsection, I investigate the power of the test statistics against a fixed and  a sequence of local alternatives drifting to the null.
First, I establish that $\hat{S}_1$ and $\hat{S}_2$ are consistent against the following  fixed alternatives 
\begin{align} \label{fixed alternative 1}
    H^{\Pi}_{1}: \int_{\mathcal{X}} \sum_{k =1}^K \sum_{j =1}^K \left\{|\tau(x;\pi_k)-\tau(x;\pi_j)|\right\} w(x, \pi_k, \pi_j)dx>0, 
\end{align}

\begin{align}\label{fixed alternative 2}
    H^{X}_{1}: \int_{\mathcal{X}} \int_{\mathcal{X}} \sum_{k =1}^K \left\{|\tau(x;\pi_k)-\tau(x';\pi_k)|\right\} \frac{w(x, x', \pi_k)}{2}dxdx'>0
\end{align}
respectively. These alternative hypotheses \eqref{fixed alternative 1} and \ref{fixed alternative 2} are analogous to those in  \eqref{first alternative hypothesis} and \eqref{second alternative hypothesis} respectively.

\begin{theorem}\label{power theorem}
Let  Assumptions \ref{Treatment-invariant neighborhoods} - 
 \ref{regularity conditions} hold, then\\ 
 (i) under the fixed alternative hypothesis \eqref{fixed alternative 1}
$$\lim_{C\to \infty} \Pr(\hat{S}_1>z_{1-\alpha})=1, \,\,\,\, \text{and}$$\\
(ii) under the fixed alternative hypothesis \eqref{fixed alternative 2},
$$\lim_{C\to \infty} \Pr(\hat{S}_2>z_{1-\alpha})=1. \,\,\,\,$$
\end{theorem}

Theorem \ref{power theorem} shows that the proposed test statistics have  power against fixed alternatives.
Next, I show that the proposed tests $\hat{S}_1$ and $\hat{S}_2$ can detect a sequence of local alternatives converging to the null hypotheses.
Specifically, consider the following sequences of local alternatives  
\begin{equation} \label{local alt 1}
    H^{\Pi}_{a}: \tau(x,\pi)- \tau(x,\pi')= C^{-1/2}h^{-d/4}\cdot N_0^{-1/2}\delta_1(x, \pi, \pi')\,\, \forall x \in \mathcal{X}, \pi, \pi' \in \Pi
\end{equation}
\begin{equation} \label{local alt 2}
    H^{X}_{a}: \tau(x,\pi)- \tau(x',\pi)= C^{-1/2}h^{-d/4}\cdot N_0^{-1/2}\delta_2(x, x', \pi)\,\, \forall x, x' \in \mathcal{X}, \pi\in \Pi,
\end{equation}
converging to the null hypotheses $H^{\Pi}_{a}$ and $H^{X}_{a}$ respectively where for $j=1,2$,\, $\delta_j(\cdot,\cdot,\cdot)$ is a real bounded  function satisfying:\\ $\int_{\mathcal{X}}\sum_{\pi \in \Pi}\sum_{\pi' \in \Pi}|\delta_1(x, \pi, \pi')|w(x, \pi, \pi')dx>0,$ and $\int_{\mathcal{X}}\int_{\mathcal{X}}\sum_{\pi \in \Pi}|\delta_2(x, x', \pi)|w(x,x', \pi)dxdx' \allowbreak >0.$ 

\begin{theorem}\label{local power theorem}
Let  Assumptions \ref{Treatment-invariant neighborhoods} - 
 \ref{regularity conditions} hold, then\\ 
 (i) under the sequences of  alternative hypotheses (\ref{local alt 1}),
$$\lim_{C\to \infty} \Pr(\hat{S}_1>z_{1-\alpha})= 1- {\Phi}\Bigg(z_{1-\alpha}- \frac{1}{\sqrt{2\pi}\sigma_{1}}\int_{\mathcal{X}}\sum_{k=1}^K\sum_{j=1}^K\delta^2(x, \pi_k, \pi_j)dx\Bigg),$$
and;
\newline
(ii) under the sequences of  alternative hypotheses (\ref{local alt 2}),
$$\lim_{C\to \infty} \Pr(\hat{S}_2>z_{1-\alpha})= 1- {\Phi}\Bigg(z_{1-\alpha}- \frac{1}{2\sqrt{2\pi}\sigma_{2}}\int_{\mathcal{X}}\int_{\mathcal{X}}\sum_{j=1}^K\delta^2(x, x', \pi_k)dxdx'\Bigg),$$ 
where $\Phi$ denotes the cumulative distribution function (CDF) of the standard normal distribution.
\end{theorem}
Theorem \ref{local power theorem} demonstrates that the test statistics have statistical power greater than zero against local alternatives approaching the null at a rate slower than $C^{-1/2}h^{-d/4}.$  This rate agrees with that found in nonparametric kernel-based tests of parametric restrictions.

\begin{remark}\label{COD}
 While the asymptotic results hold for $d \geq 1,$ it is important to recognize that kernel-based estimators may be biased in higher dimensions due to the curse of dimensionality. However, this does not diminish the practical value of the testing procedures developed in this paper, as they remain highly effective for covariates of low dimension --- as long as the data size is moderate --- which is a common setting in many empirical applications. 
    
In general, testing HTEs in the presence of high-dimensional covariates remains an open area of research. Although \cite{fan2022estimation}, \cite{huang2022robust}, and some other papers have studied estimation and inference of CATE under the SUTVA assumption with high dimensional covariates $X,$ no existing papers have studied HTEs testing with high dimensional covariates. A straightforward approach will involve an application of the two-stage method on split samples described in \cite{fan2022estimation}. Specifically, a machine learning algorithm will reduce the dimension of $X$ in the first stage. Then, the testing procedure proposed in this paper can be applied to the $X$ selected from the first stage on an independent set of clusters.
\end{remark}

\subsection{Bootstrap-Based Inference}\label{bootstrap}
 I introduce bootstrap resampling procedures in this subsection to obtain the null distributions of $\hat{S}_1$ and $\hat{S}_2$. These bootstrap procedures provide a benchmark for the aforementioned asymptotic-based methods in finite samples. Specifically, I numerically compare the power and size of the  asymptotic-based and bootstrap-based procedures for finite samples in Appendix \ref{appendsim}.

Following the recommendations of several authors, particularly \cite{davidson1999size},  I propose the following bootstrap resampling procedures that impose the null hypotheses. 
\subsubsection{The Bootstrap resampling algorithm for $\hat{S}_1$}\label{balgorithm 1}
Let $\mathbf{W}_c=(\mathbf X_c, \mathbf T_c,\mathbf Y_c)$ denote the vector of variables for the $c^{th}$ cluster. Therefore, the pooled sample across all clusters can be written as $\{\mathbf W_{c}\}_{c=1}^{C}.$ I propose the following pairs cluster bootstrap-t   procedure with the null $H_0^{\Pi}$ imposed to generate the null distribution of $\hat{S}_1$:
\begin{enumerate}[noitemsep,topsep=0pt]
    \item For $k=1,\dots, K,$ randomly draw $C_k$ clusters from the pooled clusters with replacement and denote the resulting bootstrapped pseudo-sample combined with a treatment exposure  $\pi_k$ as  $W^{*}(\pi_k):=\{\mathbf W^*_c, \pi_k\}_{c=1}^{C_k}.$  That is, for each $k,$ all the $C_k$ clusters in that pseudo-sample have a treatment exposure of $\pi_k$ which may differ from their true  exposure.
    \item Compute the test statistic $\hat{S}^*_1=\hat{T}^*_{1}-a^*_{_1}$ using the pooled bootstrapped data $\{W^{*}(\pi_k)\}_{k=1}^K$, where the  definition of $\hat{T}^*_{1}$  and $a^*_{_1}$ are the same as $\hat{T}_{1}$ and $a_{_1}$ respectively. Comparing $\hat{S}^*_1$ to its asymptotic counterpart, note that I omit the standard error term in the denominator to reduce computation time.
    \item  Repeat  1 and 2 a large number of times (say $B_1$ times) and use the empirical distribution of the  $B_1$ bootstrapped test statistics $\{\hat{S}^*_{1, b}\}_{b=1}^{B_1}$ to approximate the null distribution of $\hat{T}_{1}-a_{_1}.$
    \item Compute the empirical $p$-value as $\hat{p}*=B_1^{-1}\sum_{b=1}^{B_1}\mathbbm{1}(\hat{S}^*_{1, b} >\hat{S}^o_{1})$ where $\hat{S}^o_{1}$ is the test statistic computed with the original data.
\end{enumerate}

\subsubsection{The Bootstrap resampling algorithm for $\hat{S}_2$}\label{balgorithm 2}
 I propose the following wild cluster bootstrap-t procedure with the null $H_0^{X}$ imposed to generate the null distribution of $\hat{S}_2$:
\begin{enumerate}[noitemsep,topsep=0pt]
    \item Estimate a "restricted" nonparametric conditional mean function   $\mathbbm{E}(Y_{ci}|X=\bar{x},\Pi_{ci},T_{ci}).$ Let the resulting fitted values be $\hat{M}(\bar{x},\pi_c,t_{ci}),c=1,\dots, C, i=1,\dots,N.$ 
    The restricted conditional mean does not vary by the pre-treatment variables $X,$ since they are held constant at their average value $\bar{x}.$ 
    \item Let $\hat{\mathbf{M}}_c$ denote the vector of fitted values in cluster $c.$ Then, obtain the cluster residuals $\hat{\boldsymbol{\varepsilon}}_c= \boldsymbol{Y}_c-\hat{\mathbf{M}}_c,$ $c=1,\dots, C.$  Since we compute these residuals using the restricted conditional means, they are residuals obtained under the null hypothesis $H^{X}_{0}.$
    \item For each cluster $c=1,\dots, C,$  form  cluster dependent vectors as $$\mathbf{Y}^*_c=\hat{\mathbf{M}}_c +\hat{\boldsymbol{\varepsilon}}^*_c,$$ where $\hat{\boldsymbol{\varepsilon}}^*_c= \hat{\boldsymbol{\varepsilon}}_c \cdot V^*_c$ with $V^*_c$ Rademacher distributed.   These dependent vectors  are used to create a bootstrap pseudo-sample $(\mathbf X_c, \mathbf T_c, \boldsymbol{\Pi}_c, \mathbf Y^*_c), c=1\dots C,$ where $(\mathbf X_c, \mathbf T_c, \boldsymbol{\Pi}_c)$ are from the original sample.
    \item Compute the test statistic $\hat{S}^*_2=\hat{T}^*_{2}-a^*_{_2}$ using the null bootstrap sample where I define $\hat{T}^*_{2}$ and $a^*_{_2}$ the same way as $\hat{T}_{2}$ and $a_{_2}$ respectively. 
    Here also, I omit the scaling factor based on the same argument as in  $\hat{S}^*_1$
    \item Repeat 3--4 many times (say $B_2$ times ) and use the empirical distribution of the  $B_2$ bootstrapped test statistics $\{\hat{S}^*_{2, b}\}_{b=1}^{B_2}$ to approximate the null distribution of $\hat{T}_{1}-a_{_2}.$
    \item Compute the empirical $p$-value as $\hat{p}*=B_2^{-1}\sum_{b=1}^{B_2}\mathbbm{1}(\hat{S}^*_{2, b} >\hat{S}^o_{2}),$ where  $\hat{S}^o_{2}$ is the test statistic computed via the original data. 
\end{enumerate}

\section{Monte Carlo Simulation}\label{Monte Carlo Simulation}
In this section,  I investigate the finite sample performance of the proposed test statistics via Monte Carlo experiments. I first check the size and power of the test statistics. Then, I compare the performance with their parametric counterparts. I also report additional simulation results in Appendix \ref{appendsim}.
\subsection{ Empirical size and statistical power}\label{size and power sim}
All the rejection probabilities in the simulations are based on 1000 replications.
In the experiments in this section, I focus on one pre-treatment variable $X_{ci}$\footnote{I extend the experiment to the case with two pre-treatment variables in Appendix \ref{multivarX}.}  drawn from the uniform $[0, 1]$ distribution allowing for within-cluster correlations. Each cluster is assigned one of four treatment vectors $\textbf{T}(k), k=1\dots,4.$ The average of $\textbf{T}(k)$ is $\pi_k$ where $\{\pi_1,\pi_2,\pi_3,\pi_4 \}= \{0.3, 0.4, 0.5, 0.6\}=\mathbf{\Pi}.$ The realized outcome $Y_{ci}$ is  constructed  as follows:
$$Y_{ci} = (\tau(X_{ci}, \Pi_{ci}) + U_{1ci}) \times T_{ci} +  U_{0ci} \times (1 - T_{ci}),$$
where  $U_{1ci}$ and $U_{0ci}$ are independent normals with a mean of zero and variance of 0.1, and $T_{ci} \in \textbf{T}(k).$
The general specification of the CATE  is 
$$\tau(x;\pi) = \beta_0 x + \beta_1 \pi.$$
In the experiments, I compute CATEs on a uniform grid in the interval between the $10^{th}$ and $90^{th}$ percentiles of $X$. This way, I avoid the \textit{boundary bias} issue associated with the kernel estimators. Integrals are computed with the composite trapezoid technique.
 
 I adopt the following kernel function that satisfies Assumption \ref{regularity conditions}(d):
 \begin{equation} \label{kernel}
   K(u) =1.5(1 - (2u)^2)\cdot\mathbbm{1}\{|u|\leq 0.5\},  
 \end{equation}
and a   bandwidth 
\begin{equation}\label{bandwidth}
 h = \kappa_h\hat{s}_XC^{-2/7},   
\end{equation}
where $\hat{s}_X$ is the sample standard deviation of  $X$ and $\kappa_h$ is a constant. A similar kernel and bandwidth are employed in \cite{chang2015nonparametric}.
Moreover, I set $C=150$ with 10 units $(N_c=10)$ in each cluster (i.e., N=1500 units).

To compute the empirical size of $\hat{S}_1$, I fix $\beta_0=1$ and $\beta_1=0.$ Thus,   $\tau(x;\pi) =  x,$ and   the  null hypothesis $H_0^\Pi$ of CTEs by $\Pi$ is true.
Note that as  $\beta_1$ deviates further from  $0$ in both directions, the  null hypothesis  deviates further away from the truth. Again, I report the empirical rejection probabilities for values of $\beta_1$ in the range $[-0.5, 0.5]$, with  0.05 increments.
 
In contrast, to obtain the empirical size of $\hat{S}_2$, I set $\beta_0=0,$ and $\beta_1=1.$ Thus, $\tau(x;\pi) =  \pi,$ and  the  null hypothesis $H_0^X$ of CTEs by $X$ is true.
I report the empirical rejection probabilities for each $\beta_0$ in the range $[-0.5,0.5 ]$, with  0.05 increments.

Focusing on the asymptotic-based procedure, a summary of the empirical rejection probabilities is in  Figure \ref{powercurves-asym} (and in Tables \ref{power table:1}--\ref{power table:2b}  in  Appendix \ref{appendsim}). In each panel, the plots represent the rejection probabilities at the 1\%, 5\%, and 10\% nominal levels. The left panel reports the empirical rejection probabilities of $\hat{S}_1,$ and the right panel reports those of $\hat{S}_2.$ When $\beta_0=\beta_1=0,$ the empirical rejection probabilities are close to the nominal probabilities, corroborating the theoretical results in Theorem \ref{valid size}. On the contrary, as $\beta_0$ and $\beta_1$ deviate towards $\pm 0.5,$ the rejection probabilities approach 1, which aligns with  Theorem \ref{power theorem}.  

Similarly, using the bootstrap procedures, I summarize the empirical rejection probabilities in  Figure \ref{powercurves-bootstrap} (and in Tables \ref{power table boot:1}--\ref{power table boot:2}  in  Appendix \ref{appendsim}). The results are based on 399 bootstrap resamples. Compared to the asymptotic-based empirical rejection probabilities, the differences in the probabilities are negligible. This suggests that the test statistics under the null hypotheses "converge in bootstrap distribution" to the standard normal distribution for large sample sizes.  

\begin{figure}[ht]
\setcounter{subfigure}{0}
\centering
\caption{Empirical rejection probabilities using asymptotic method.}
\subfigure[Power curve for $\hat{S}_1$ when $\beta_1$ lies between -0.5 and 0.5.]{
    \includegraphics[width=0.46\columnwidth, keepaspectratio]{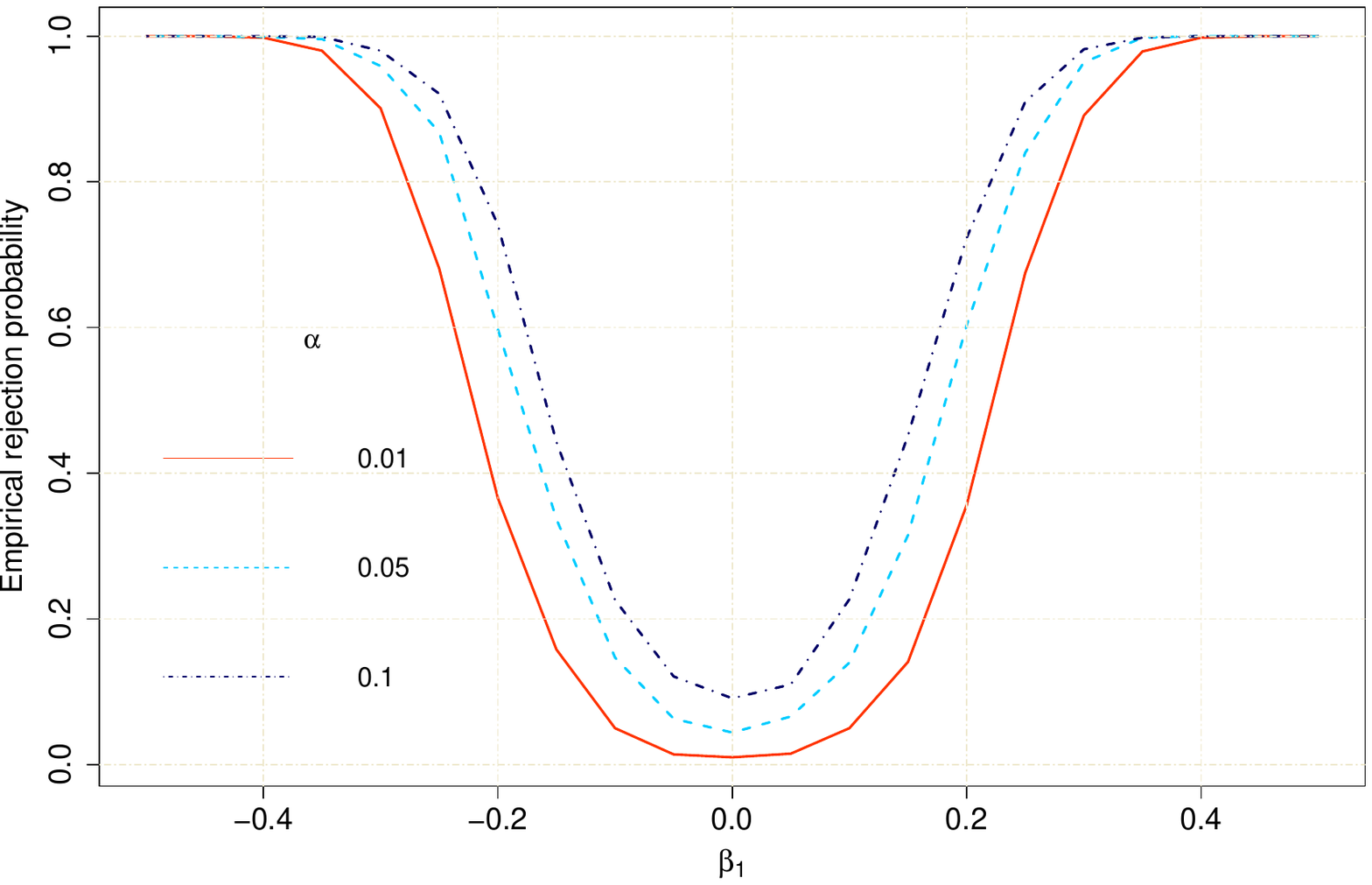}
}
\subfigure[Power curve for $\hat{S}_2$ when $\beta_0$ lies between -0.5 and 0.5.]{
    \includegraphics[width=0.46\columnwidth, keepaspectratio]{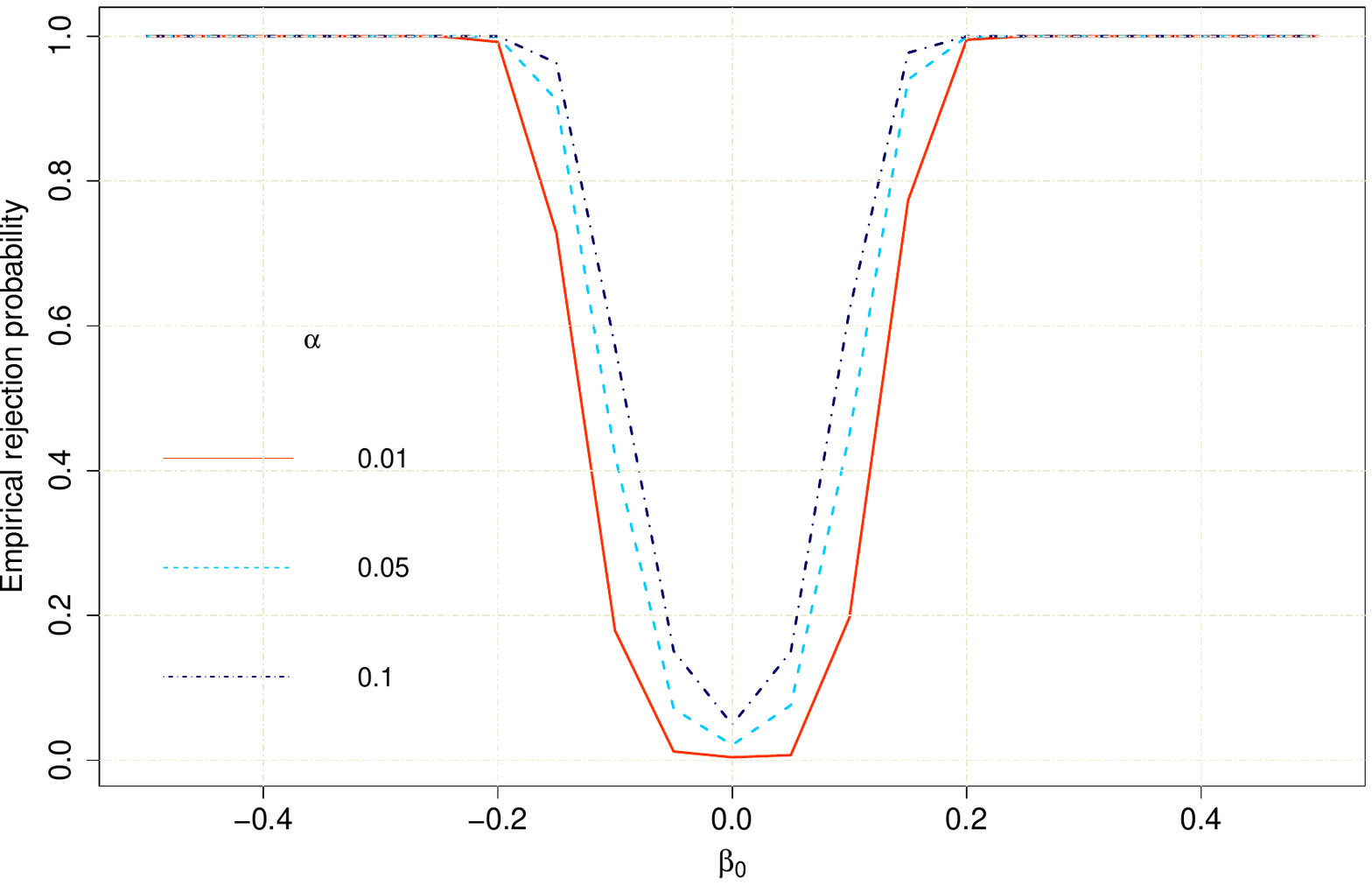}       
}\\
\label{powercurves-asym}
\end{figure}

\begin{figure}[ht]
\setcounter{subfigure}{0}
\centering
\caption{Empirical rejection probabilities using the bootstrap method.}
\subfigure[Power curve for $\hat{S}^*_1$ when $\beta_1$ lies between -0.5 and 0.5.]{
    \includegraphics[width=0.46\columnwidth, keepaspectratio]{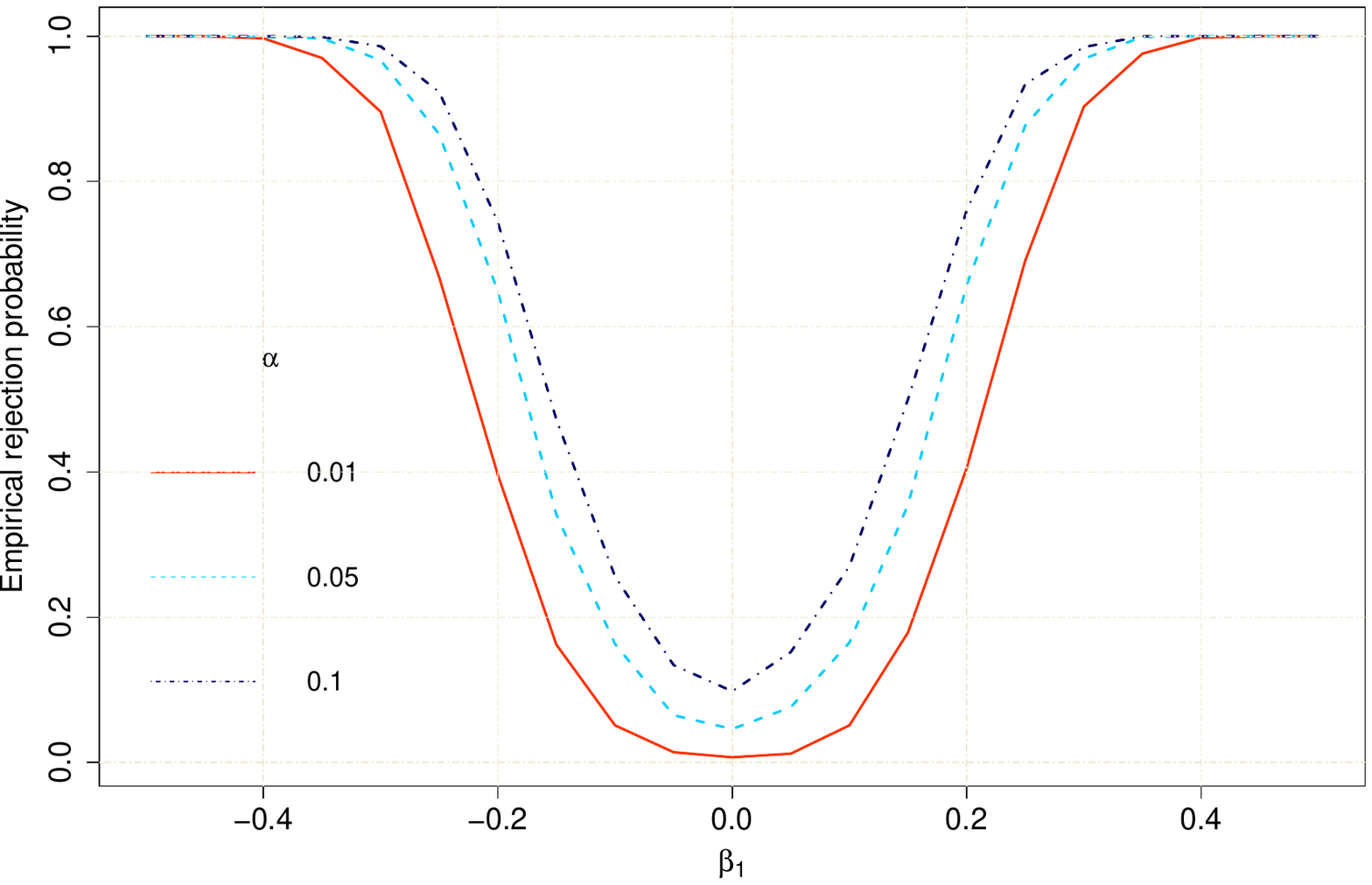}
}
\subfigure[Power curve for $\hat{S}^*_2$ when $\beta_0$ lies between -0.5 and 0.5.]{
    \includegraphics[width=0.46\columnwidth, keepaspectratio]{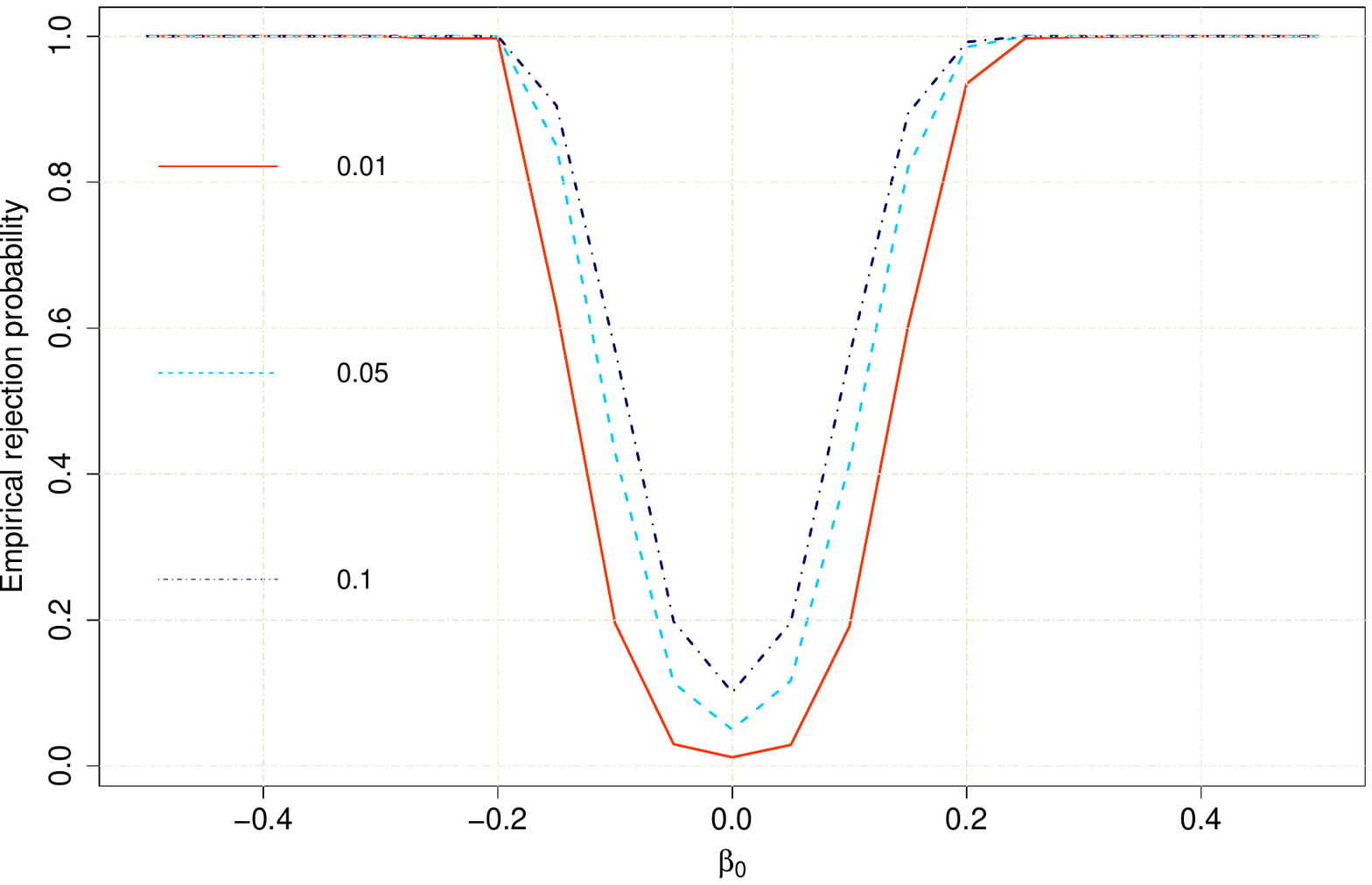}       
}
\\
\label{powercurves-bootstrap}
\end{figure}

\subsection{Parametric Testing and Misspecification}
The Monte Carlo experiment in this subsection seeks to show that parametric tests of $H_0^\Pi$ and $H_0^X$ may be misleading because parametric models are always misspecified to a certain degree. I only report the results of the asymptotic-based procedure to save space.
I generate sample data  $\{Y_i,X_i,T_i, \Pi_i\}$ of size N=600 ($C=60$ and $N_c=10$), and $\Pi_i \in \{0.3,0.4,0.5,0.6\}$. 
The outcome is of the form:
$$Y_{ci} = (\tau(X_{ci}, \Pi_{ci}) + U_{1ci}) \times T_{ci} +  U_{0ci} \times (1 - T_{ci}),$$
where  $$\tau(x;\pi) =30\cdot \cos(2\cdot3.142\cdot x)\cdot (\pi^2-\pi).$$ I keep the remaining design as in Section \ref{size and power sim} above. This new  CATE specification is highly non-linear in  $X$ and $\Pi.$

 I estimate the following linear regression model using ordinary least squares:
\begin{equation*}
    Y_{ci}=\beta_0 + \beta_1 T_{ci} +\beta_2 X_{ci} + \beta_3 \Pi_{ci} + \beta_4 T_{ci}\cdot X_{ci} + \beta_5 T_{ci}\cdot \Pi_{ci} + \varepsilon_{ci}.
\end{equation*}
The parameters $\beta_4$ and $\beta_5$ measures HTEs across $X$ and $\Pi$ respectively. The estimation results in Table \ref{Para}  show that $\beta_4$ and $\beta_5$  are insignificant when we use clustered standard errors. It leads to the wrong conclusion that treatment effects do not vary by the exposure variable $\Pi,$ and the pre-treatment variable $X$.

\begin{table}
\centering
 \caption{Summary of Test Results for Simulated DGP based on Parametric Tests using clustered standard errors} 
  \label{Para} 
\begin{tabular}{llllll}
\toprule
& Estimate  & Std. Error & t value  &$p$-value  \\
\toprule
Intercept & 0.027   &0.018   & 1.549 & 0.122 \\
\\
$T$ &  4.473 & 3.018 & 1.483 & 0.139 \\
\\
$X$ & -0.012& 0.016  &-0.760  & 0.448\\
\\
$\Pi$ & -0.022 & 0.042 & -0.527& 0.598 \\
\\
$T\cdot X$ &-3.101  &2.362  &-1.313  & 0.190 \\
\\
$T \cdot \Pi$ & -6.565 &6.702  &-0.980  & 0.328 \\
 \bottomrule
 \hline \\[-1.8ex] 
Number of observations: & 600 \\ 
R$^{2}$ & 0.002 \\ 
Adjusted R$^{2}$ & $-$0.007 \\ 
Residual Std. Error & 3.517 (df = 594) \\ 
F Statistic & 0.200 (df = 5,594) \\ 
\hline 
\end{tabular}
\end{table}

Next, I test the null hypotheses using this paper's proposed nonparametric test statistics. I use the kernel function in \eqref{kernel} and the bandwidth formula in \eqref{bandwidth}.
Table \ref{Para:1} summarizes the results of the two tests at different bandwidth choices (different $\kappa_h$ in the bandwidth formula in \eqref{bandwidth}). The results unequivocally reject the null hypotheses of CTEs by the exposure variable $\Pi,$  and the pre-treatment variable $X$. This serves as a stark reminder of the potential for misspecification of the functional form of the conditional mean in parametric models to lead to an erroneous inference of HTEs, underscoring the need for the nonparametric testing procedures in this paper.

\begin{table}[H]
\caption{Summary of Test results for Simulated DGP based on proposed nonparametric test}
\label{Para:1}
\centering
  \begin{tabular}{lSSSSSS}
    \toprule
    \multirow{2}{*}{Bandwidth (h)} &
 \multicolumn{2}{c}{\underline{$H^{\Pi}_{0}$: CTEs across $\Pi$} } &
      \multicolumn{2}{c}{\underline{$H^{X}_{0}$: CTEs across $X$}} \\\\
      & {$\hat{S}_1$} & {$p$-value} & {$\hat{S}_2$} & {$p$-value}  \\
      \midrule
    0.195 & 6.705 & <0.01 & 67.266 & 0.000 \\
    0.232 & 5.205 & <0.01 & 51.360& 0.000  \\
    0.296 & 3.963 & <0.01 & 34.470 & 0.000 \\
    0.371 & 3.345 & <0.01 & 23.206 & 0.000 \\
    \bottomrule
     \hline \\[-1.8ex] 
Number of observations: & 600 \\
 \hline
  \end{tabular}
\end{table}


\section{Empirical Application} \label{application}
 In this section, I use the experimental data from \cite{cai2015social} to demonstrate the usage of the proposed test statistics.
This experiment was implemented to help determine whether farmers' understanding of a weather insurance policy affects purchasing decisions.  
The authors examine the impact of two types of information sessions on insurance adoption among 5335 households in 185 small rice-producing villages (47 administrative villages) in 3 regions in the Jiangxi province in China. They show that the type of information session directly affects participants' adoption and significantly affects the adoption decision of participants' friends.

The data includes each household's network information (each household has at most five friends) and additional pre-treatment information such as age, gender, rice production area, risk aversion score, the fraction of household income from rice production, and others. The outcome of interest is binary: whether or not a household buys the insurance policy. For each village in the experiment, there were \textit{two rounds} of information sessions offered to introduce the insurance product. Households are randomly assigned to rounds.\footnote{Each household can only participate in one of the two rounds.} In each round, \textit{two sessions} were held \textit{simultaneously}: a simple session  (with less information) and an intensive session. Households are randomly assigned to sessions. Participants have to make a purchase decision on the spot in each session. Since the second round of information sessions were held three days after the first, \cite{cai2015social} argue that information may spill to households participating in the second round of information sessions from their friends who attended the first round. That is, the insurance adoption of those in the first round depends only on their treatment (participation) status, \textit{but the insurance adoption of participants in the second round depends on their treatment and the treatment statuses of their friends in the first round.}
Table \ref{logit regression} shows the result of a nonlinear probability (logit) model of insurance adoption on "Second-round" --- a binary variable which is one if households participated in the second round session and zero otherwise --- among households in the \textit{simple sessions}. The coefficients of "Second-round" indicate that the probability of adoption among the second-round participants in the simple sessions is higher than their first-round counterparts. This is evidence of information spillover.

\begin{table}[!htbp] \centering 
\caption{Logit Regression Evidence of Information spillover} 
 \label{logit regression}  
\begin{tabular}{@{\extracolsep{5pt}}lcc} 
\\[-1.8ex]\hline 
\hline \\[-1.8ex] 
 & \multicolumn{2}{c}{\textit{Dependent variable:}} \\ 
\cline{2-3} 
\\[-1.8ex] & \multicolumn{2}{c}{Insurance takeup} \\ 
\\[-1.8ex] & (Without village fixed effects)  & (With village fixed effects)\\ 
\hline \\[-1.8ex] 
 Constant & $-$0.609$^{***}$ & 0.693 \\ 
  & (0.064) & (0.875) \\ 
  & & \\ 
 Second-round & 0.423$^{***}$ & 0.474$^{***}$ \\ 
  & (0.084) & (0.089) \\ 
  & & \\ 
\hline \\[-1.8ex] 
Observations & 2,453 & 2,453 \\ 
Log Likelihood & $-$1,646.452 & $-$1,546.740 \\ 
Akaike Inf. Crit. & 3,296.904 & 3,189.480 \\ 
\hline 
\hline \\[-1.8ex] 
\textit{Note:}  & \multicolumn{2}{r}{$^{*}$p$<$0.1; $^{**}$p$<$0.05; $^{***}$p$<$0.01} \\ 
 \hline
\end{tabular} 
\end{table}

Furthermore, the analyses in \cite{cai2015social} rule out any form of between-administrative village diffusion of information, which makes the data fit into the clustered network setting of this paper.
Figure \ref{Network plot} from \cite{bargagli2020heterogeneous} shows a network plot of the data for the 47 administrative villages.  
\begin{figure}[H]
    \centering
    \includegraphics[scale=0.6]{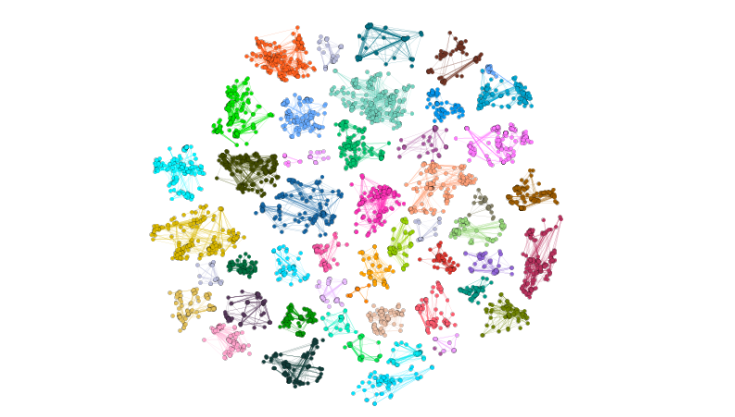}
    \caption{Links between households in the  47 administrative villages. Different Villages have different colors. }
    \label{Network plot}
\end{figure}
In addition, \cite{cai2015social} show that while farmers are inﬂuenced by their friends who attended the first round of intensive sessions, they are unaffected by friends who attended the first round of simple sessions. They also show that people are less inﬂuenced by their friends when they have the same education about the insurance products (e.g., information from first-round households in the intensive session does not affect the purchasing decision of second-round households in the intensive session). These observations are crucial for defining the post-treatment exposure variable in this application.

To illustrate the use of the proposed test statistics, I focus on the second-round participants in each administrative village. I define the post-treatment exposure variable as the fraction of households who attended the first-round intensive session in an administrative village. I also focus on one pre-treatment variable: the fraction of household income from rice production.\footnote{The fraction of household income from rice production should significantly influence a household's decision to purchase insurance. Households with a small fraction of income from rice production are unlikely to buy insurance, those with a moderate fraction are more likely to purchase it, and households with a large fraction are also unlikely to buy --- because they may already have an insurance policy --- \textit{ceteris paribus}.} A household is considered "treated" if it attended an intensive session, while those that participated in the simple session are labeled "untreated." The outcome of interest is the second-round participants' adoption of the insurance policy. 

There are 2653 households in my restricted sample (second-round participants with no missing values for the four variables). To ensure balance in the post-treatment exposure variable (treatment ratio), I binarize the treatment ratio: households from villages with the fraction of first-round intensive session participants less than 0.22 constitute the first group, and those from villages with the fraction of first-round intensive session participants greater than 0.22 make up the second group. There are 1488 and 1239 households in the two groups, respectively.\footnote{Note that the threshold for the binarization is purely based on sample size considerations. I chose the threshold so the two groups will have sufficient units.} Thus, technically, the exposure variable is a threshold function of the treatment ratio.

Using this sample, I test the two null hypotheses. Figure \ref{CATE plot} shows the estimated average treatment effects against the fraction of income from rice production for the two exposure groups. Visually, the estimated ATEs seem to vary with the exposure variable and the fraction of household income from rice production. However, are these variations statistically significant? 
Using the kernel function defined in \eqref{kernel} and the bandwidth formula in \eqref{bandwidth}, the results in Table \ref{CATE:1} show that the proposed asymptotic-based tests fail to reject the null hypotheses of CTEs by the treatment ratio variable and the fraction of household income from rice production at different bandwidths and conventional significance levels. Moreover, based on Holm's MTP, I jointly fail to reject the two null hypotheses at the different bandwidths and conventional significance levels.
\begin{figure}[H]
\caption{Estimated conditional average treatment effects using the full sample}
    \centering
    \includegraphics[scale=0.47]{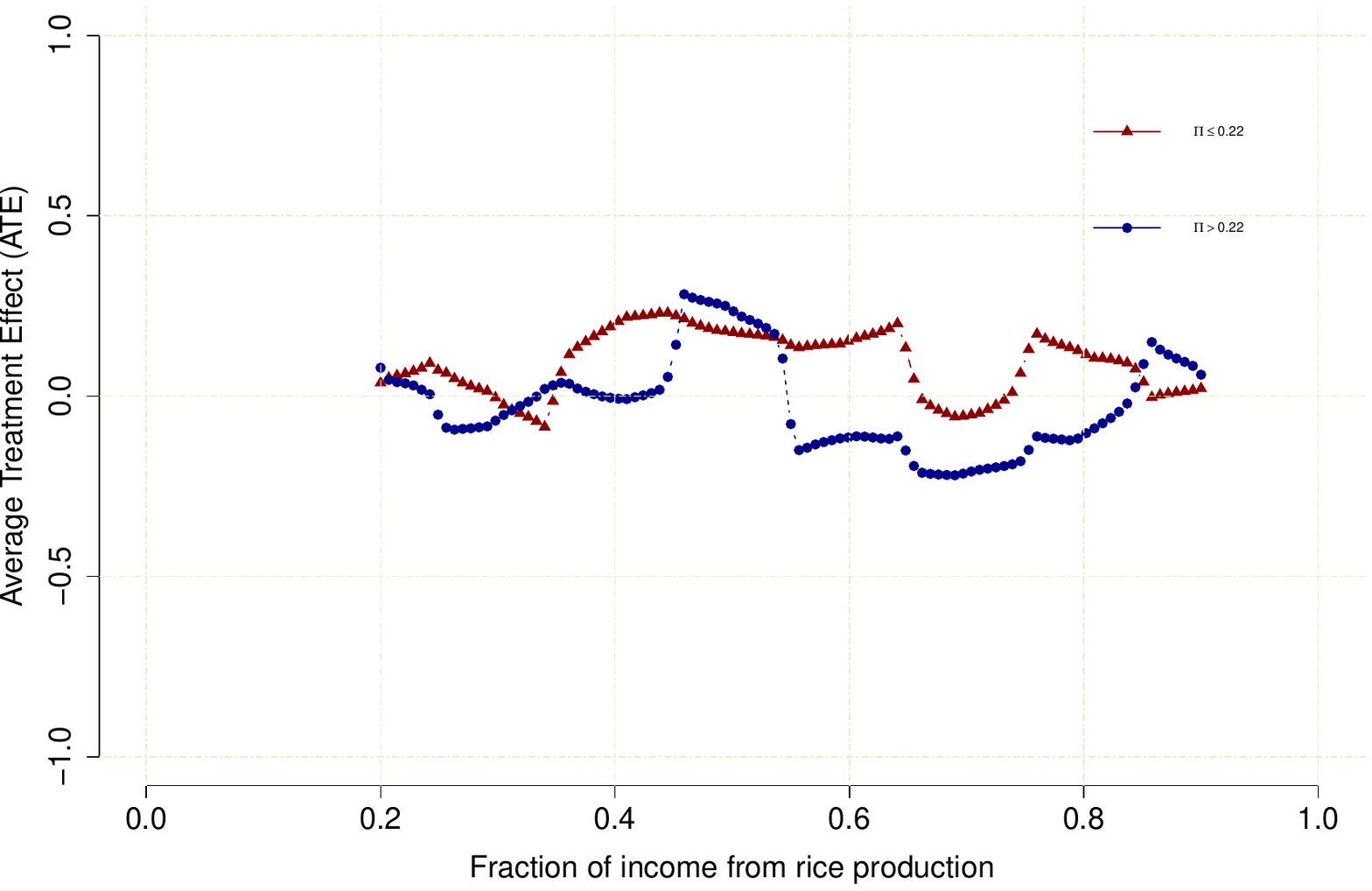}
\label{CATE plot}
\end{figure}

\begin{table}[H]
\caption{Nonparametric test based on full sample}
\label{CATE:1}
\centering
 \begin{adjustbox}{width=\textwidth}
  \begin{tabular}{lSSSSSS}
    \toprule
    \multirow{2}{*}{Bandwidth (h)} &
      \multicolumn{2}{c}{\underline{$H^{\Pi}_{0}$: CTEs across treatment ratios} } &
      \multicolumn{2}{c}{\underline{$H^{X}_{0}$: CTEs across fraction of rice income}} \\
      & {$\hat{S}_1$} & {$p$-value} & {$\hat{S}_2$} & {$p$-value}  \\
      \midrule
     0.096 & 1.076 & 0.141 & 0.143 & 0.443 \\
    0.106 & 1.117 & 0.132 &  0.144 & 0.443  \\
    0.116 & 1.192 & 0.117 &  0.363 & 0.358 \\
    \bottomrule
  \end{tabular}
    \end{adjustbox}
\end{table}

Based on Figure \ref{CATE plot}, there is a wider gap between the ATEs for the two categories of treatment ratio when the fraction of rice production lies between 0.5 and 0.8. Thus, to check the robustness of the proposed testing procedures, I restrict the sample to these households and re-implement the tests. Figure  \ref{CATE plot2} shows a plot of the estimated ATEs against the fraction of income from rice production using the new sample. The test results in Table \ref{CATE:2} show that ATE varies with the treatment ratio (at the conventional significance levels)  but is constant across the fraction of income from rice production as expected.
\begin{figure}[H]
\caption{Estimated conditional average treatment effects  using the subsample}
    \centering
    \includegraphics[scale=0.5]{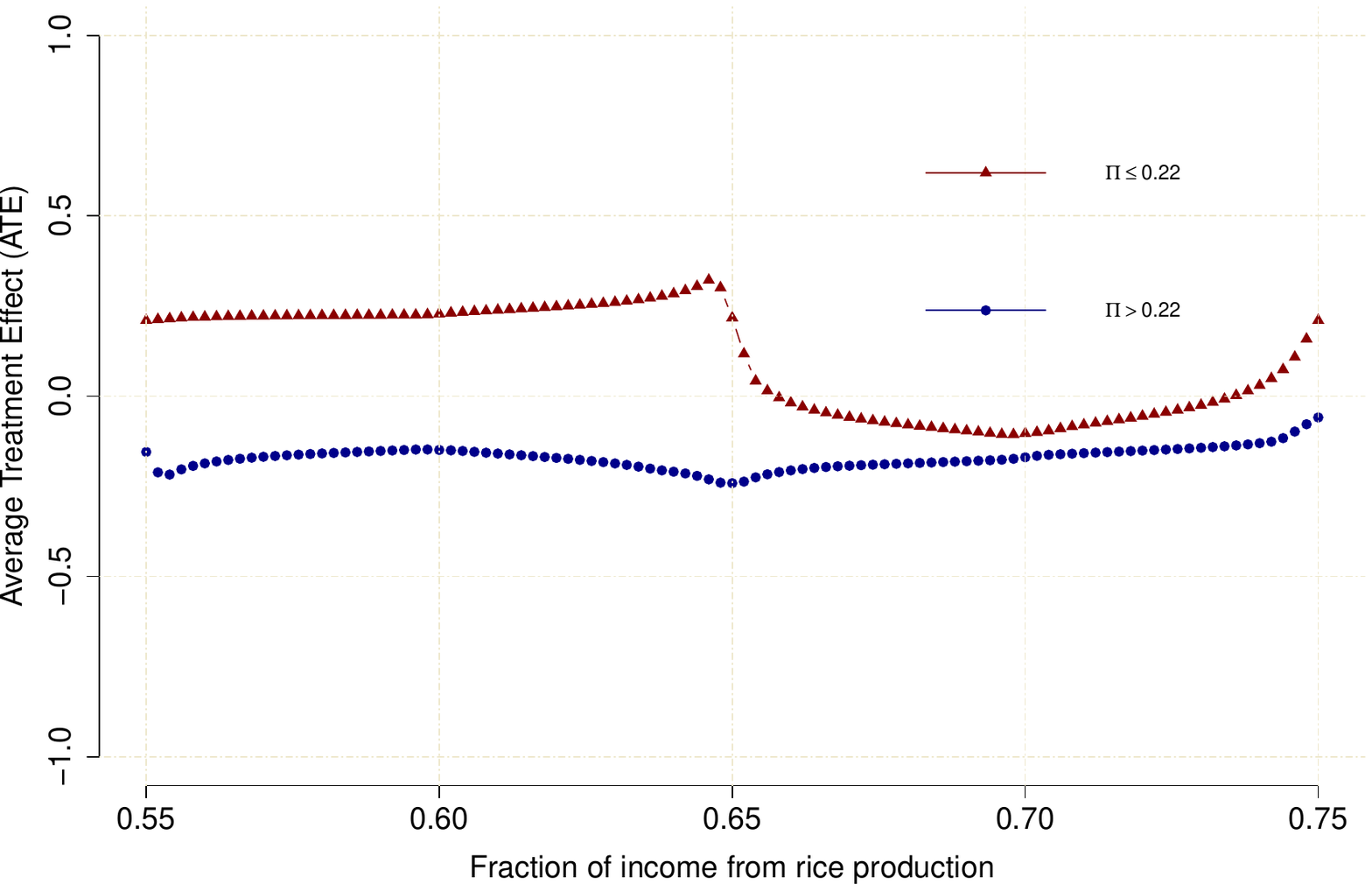}
\label{CATE plot2}
\end{figure}
\begin{table}[H]
\caption{nonparametric test of subsample}
\label{CATE:2}
\centering
\begin{adjustbox}{width=\textwidth}
  \begin{tabular}{lSSSSSS}
    \toprule
    \multirow{2}{*}{Bandwidth (h)} &
     \multicolumn{2}{c}{\underline{$H^{\Pi}_{0}$: CTEs across treatment ratios} } &
      \multicolumn{2}{c}{\underline{$H^{X}_{0}$: CTEs across fraction of rice income}} \\
      & {$\hat{S}_1$} & {$p$-value} & {$\hat{S}_2$} & {$p$-value}  \\
      \midrule
    0.099 & 3.288 & <0.01 & -0.086 & 0.534 \\
    0.102 & 3.390 & <0.01 & -0.081& 0.532  \\
    0.106 & 3.425 & <0.01 & -0.079 & 0.531 \\
    \bottomrule
  \end{tabular}
  \end{adjustbox}
\end{table}

\section{Conclusion}\label{Conclusion}
The nonparametric tests I develop in this paper not only allow for valid  inference for heterogeneous treatment effects in the presence of clustered interference but also play a crucial role in disentangling the source of variation in the treatment effects. This unique feature sets the proposed tests apart from existing procedures.
The test statistics are sums of weighted $L_1$-norm pairwise differences in consistent nonparametric kernel estimators of conditional average treatment effects.
Applying the poissonization technique in \cite{gine2003bm}, I show that the test statistics are asymptotically normally distributed and have correct sizes in large samples. Moreover, they are consistent under fixed alternatives and have nonzero power against local alternatives, drifting to the null. I provide Monte Carlo results that corroborate the theoretical findings. On the applied side, I illustrate the usage of the proposed tests with the data from \cite{cai2015social}. I find no evidence of heterogeneous treatment effects across the values of both a post-treatment exposure variable and a pre-treatment variable: the fraction of household income from rice production.

There is room for several extensions. It will be interesting to provide theoretical guarantees for the proposed bootstrap methods in the paper. For instance, Table \ref{comparison-bootasym} in Appendix \ref{appendsim}  provides simulation evidence that suggests that the bootstrap algorithm of $\hat{S}_1$ achieves asymptotic refinement. It will be insightful to theoretically investigate if the proposed bootstrap algorithms achieve asymptotic refinement over asymptotic-based tests. Since it is never trivial to show bootstrap refinements even in more standard tests, I defer such a study to future research. Secondly, it is valuable to provide HTEs testing procedures in networks that are not clustered (single un-clustered networks). \cite{owusu2023randomization} studies randomization tests for HTEs in single un-clustered networks. The procedures in this companion paper are, however, only applicable to experimental data.

\section{Acknowledgments}
This is a revised version of the first Chapter of my Ph.D. Dissertation at McMaster University, \cite{owusu2023treatment}.
It has benefited from feedback of Youngki Shin, Jeffrey S. Racine, Micheal Veall, Saraswata Chaudhuri, Sukjin Han, Emmanuel S. Tsyawo, Monika A. M\'arquez, Antoine A. Djogbenou, and participants at the 32nd Annual Meeting of Midwest Econometrics Group  and the 56th Annual Conference of the Canadian Economics Association. This work was made possible by the advanced research computing platform provided by the Digital Research Alliance of Canada (formerly Compute Canada). All errors are my own.

\bibliographystyle{chicago}
\bibliography{reference}

\begin{thebibliography}{}

\bibitem[\protect\citeauthoryear{Abrevaya, Hsu, and Lieli}{Abrevaya
  et~al.}{2015}]{abrevaya2015estimating}
Abrevaya, J., Y.-C. Hsu, and R.~P. Lieli (2015).
\newblock Estimating conditional average treatment effects.
\newblock {\em Journal of Business \& Economic Statistics\/}~{\em 33\/}(4),
  485--505.

\bibitem[\protect\citeauthoryear{Aronow, Samii, et~al.}{Aronow
  et~al.}{2017}]{aronow2017estimating}
Aronow, P.~M., C.~Samii, et~al. (2017).
\newblock Estimating average causal effects under general interference, with
  application to a social network experiment.
\newblock {\em The Annals of Applied Statistics\/}~{\em 11\/}(4), 1912--1947.

\bibitem[\protect\citeauthoryear{Athey, Eckles, and Imbens}{Athey
  et~al.}{2018}]{athey2018exact}
Athey, S., D.~Eckles, and G.~W. Imbens (2018).
\newblock Exact p-values for network interference.
\newblock {\em Journal of the American Statistical Association\/}~{\em
  113\/}(521), 230--240.

\bibitem[\protect\citeauthoryear{Bargagli~Stoffi, Tort{\'u}, and
  Forastiere}{Bargagli~Stoffi et~al.}{2020}]{bargagli2020heterogeneous}
Bargagli~Stoffi, F., C.~Tort{\'u}, and L.~Forastiere (2020).
\newblock Heterogeneous treatment and spillover effects under clustered network
  interference.
\newblock {\em Costanza and Forastiere, Laura, Heterogeneous Treatment and
  Spillover Effects Under Clustered Network Interference (August 3, 2020)\/}.

\bibitem[\protect\citeauthoryear{Billingsley}{Billingsley}{1968}]{billingsley1968convergence}
Billingsley, P. (1968).
\newblock Convergence of probability measures.

\bibitem[\protect\citeauthoryear{Bitler, Gelbach, and Hoynes}{Bitler
  et~al.}{2006}]{bitler2006mean}
Bitler, M.~P., J.~B. Gelbach, and H.~W. Hoynes (2006).
\newblock What mean impacts miss: Distributional effects of welfare reform
  experiments.
\newblock {\em American Economic Review\/}~{\em 96\/}(4), 988--1012.

\bibitem[\protect\citeauthoryear{Cai, De~Janvry, and Sadoulet}{Cai
  et~al.}{2015}]{cai2015social}
Cai, J., A.~De~Janvry, and E.~Sadoulet (2015).
\newblock Social networks and the decision to insure.
\newblock {\em American Economic Journal: Applied Economics\/}~{\em 7\/}(2),
  81--108.

\bibitem[\protect\citeauthoryear{Chang, Lee, and Whang}{Chang
  et~al.}{2015}]{chang2015nonparametric}
Chang, M., S.~Lee, and Y.-J. Whang (2015).
\newblock Nonparametric tests of conditional treatment effects with an
  application to single-sex schooling on academic achievements.
\newblock {\em The Econometrics Journal\/}~{\em 18\/}(3), 307--346.

\bibitem[\protect\citeauthoryear{Colpitts}{Colpitts}{2002}]{colpitts2002targeting}
Colpitts, T. (2002).
\newblock Targeting reemployment services in canada.
\newblock {\em Targeting Employment Services. Kalamazoo, MI: WE Upjohn
  Institute for Employment Research\/}, 283--302.

\bibitem[\protect\citeauthoryear{Cox}{Cox}{1958}]{cox1958planning}
Cox, D.~R. (1958).
\newblock Planning of experiments.

\bibitem[\protect\citeauthoryear{Crump, Hotz, Imbens, and Mitnik}{Crump
  et~al.}{2006}]{crump2006nonparametric}
Crump, R.~K., V.~J. Hotz, G.~Imbens, and O.~A. Mitnik (2006).
\newblock Nonparametric tests for treatment effect heterogeneity.

\bibitem[\protect\citeauthoryear{Davidson and MacKinnon}{Davidson and
  MacKinnon}{1999}]{davidson1999size}
Davidson, R. and J.~G. MacKinnon (1999).
\newblock The size distortion of bootstrap tests.
\newblock {\em Econometric theory\/}~{\em 15\/}(3), 361--376.

\bibitem[\protect\citeauthoryear{Ding, Feller, and Miratrix}{Ding
  et~al.}{2016}]{ding2016randomization}
Ding, P., A.~Feller, and L.~Miratrix (2016).
\newblock Randomization inference for treatment effect variation.
\newblock {\em Journal of the Royal Statistical Society: Series B (Statistical
  Methodology)\/}~{\em 78\/}(3), 655--671.

\bibitem[\protect\citeauthoryear{Fan, Hsu, Lieli, and Zhang}{Fan
  et~al.}{2022}]{fan2022estimation}
Fan, Q., Y.-C. Hsu, R.~P. Lieli, and Y.~Zhang (2022).
\newblock Estimation of conditional average treatment effects with
  high-dimensional data.
\newblock {\em Journal of Business \& Economic Statistics\/}~{\em 40\/}(1),
  313--327.

\bibitem[\protect\citeauthoryear{Gin{\'e}, Mason, and Zaitsev}{Gin{\'e}
  et~al.}{2003}]{gine2003bm}
Gin{\'e}, E., D.~M. Mason, and A.~Y. Zaitsev (2003).
\newblock The $\backslash$bm $\{$L$\}$ \_$\backslash$mathbf $\{$1$\}$-norm
  density estimator process.
\newblock {\em The Annals of Probability\/}~{\em 31\/}(2), 719--768.

\bibitem[\protect\citeauthoryear{Han, Owusu, and Shin}{Han
  et~al.}{2022}]{han2022statistical}
Han, S., J.~Owusu, and Y.~Shin (2022).
\newblock Statistical treatment rules under social interaction.
\newblock {\em arXiv preprint arXiv:2209.09077\/}.

\bibitem[\protect\citeauthoryear{Holm}{Holm}{1979}]{holm1979simple}
Holm, S. (1979).
\newblock A simple sequentially rejective multiple test procedure.
\newblock {\em Scandinavian journal of statistics\/}, 65--70.

\bibitem[\protect\citeauthoryear{Huang and Yang}{Huang and
  Yang}{2022}]{huang2022robust}
Huang, M.-Y. and S.~Yang (2022).
\newblock Robust inference of conditional average treatment effects using
  dimension reduction.
\newblock {\em Statistica Sinica\/}~{\em 32\/}(Suppl), 547.

\bibitem[\protect\citeauthoryear{Imbens and Rubin}{Imbens and
  Rubin}{2015}]{imbens2015causal}
Imbens, G.~W. and D.~B. Rubin (2015).
\newblock {\em Causal inference in statistics, social, and biomedical
  sciences}.
\newblock Cambridge University Press.

\bibitem[\protect\citeauthoryear{Lee and Shaikh}{Lee and
  Shaikh}{2014}]{lee2014multiple}
Lee, S. and A.~M. Shaikh (2014).
\newblock Multiple testing and heterogeneous treatment effects: re-evaluating
  the effect of progresa on school enrollment.
\newblock {\em Journal of Applied Econometrics\/}~{\em 29\/}(4), 612--626.

\bibitem[\protect\citeauthoryear{Lee, Song, and Whang}{Lee
  et~al.}{2013}]{lee2013testing}
Lee, S., K.~Song, and Y.-J. Whang (2013).
\newblock Testing functional inequalities.
\newblock {\em Journal of Econometrics\/}~{\em 172\/}(1), 14--32.

\bibitem[\protect\citeauthoryear{Lehmann and Romano}{Lehmann and
  Romano}{2022}]{lehmann2022multiple}
Lehmann, E. and J.~P. Romano (2022).
\newblock Multiple testing and simultaneous inference.
\newblock In {\em Testing Statistical Hypotheses}, pp.\  405--491. Springer.

\bibitem[\protect\citeauthoryear{Li, Maasoumi, and Racine}{Li
  et~al.}{2009}]{li2009nonparametric}
Li, Q., E.~Maasoumi, and J.~S. Racine (2009).
\newblock A nonparametric test for equality of distributions with mixed
  categorical and continuous data.
\newblock {\em Journal of Econometrics\/}~{\em 148\/}(2), 186--200.

\bibitem[\protect\citeauthoryear{Lin and Carroll}{Lin and
  Carroll}{2000}]{lin2000nonparametric}
Lin, X. and R.~J. Carroll (2000).
\newblock Nonparametric function estimation for clustered data when the
  predictor is measured without/with error.
\newblock {\em Journal of the American statistical Association\/}~{\em
  95\/}(450), 520--534.

\bibitem[\protect\citeauthoryear{Manski}{Manski}{2013}]{manski2013identification}
Manski, C.~F. (2013).
\newblock Identification of treatment response with social interactions.
\newblock {\em The Econometrics Journal\/}~{\em 16\/}(1), S1--S23.

\bibitem[\protect\citeauthoryear{Neyman}{Neyman}{1923}]{neyman1923applications}
Neyman, J. (1923).
\newblock Sur les applications de la th{\'e}orie des probabilit{\'e}s aux
  experiences agricoles: Essai des principes.
\newblock {\em Roczniki Nauk Rolniczych\/}~{\em 10}, 1--51.

\bibitem[\protect\citeauthoryear{Owusu}{Owusu}{2023a}]{owusu2023randomization}
Owusu, J. (2023a).
\newblock Randomization inference of heterogeneous treatment effects under
  network interference.
\newblock {\em arXiv preprint arXiv:2308.00202\/}.

\bibitem[\protect\citeauthoryear{Owusu}{Owusu}{2023b}]{owusu2023treatment}
Owusu, J. (2023b).
\newblock {\em Treatment Effect Heterogeneity and Statistical Decision-making
  in the Presence of Interference}.
\newblock Ph.\ D. thesis.

\bibitem[\protect\citeauthoryear{Racine}{Racine}{1997}]{racine1997consistent}
Racine, J. (1997).
\newblock Consistent significance testing for nonparametric regression.
\newblock {\em Journal of Business \& Economic Statistics\/}~{\em 15\/}(3),
  369--378.

\bibitem[\protect\citeauthoryear{Rosenblatt}{Rosenblatt}{1975}]{rosenblatt1975quadratic}
Rosenblatt, M. (1975).
\newblock A quadratic measure of deviation of two-dimensional density estimates
  and a test of independence.
\newblock {\em The Annals of Statistics\/}, 1--14.

\bibitem[\protect\citeauthoryear{Rubin}{Rubin}{1974}]{rubin1974estimating}
Rubin, D.~B. (1974).
\newblock Estimating causal effects of treatments in randomized and
  nonrandomized studies.
\newblock {\em Journal of educational Psychology\/}~{\em 66\/}(5), 688.

\bibitem[\protect\citeauthoryear{Sant’Anna}{Sant’Anna}{2021}]{sant2021nonparametric}
Sant’Anna, P.~H. (2021).
\newblock Nonparametric tests for treatment effect heterogeneity with duration
  outcomes.
\newblock {\em Journal of Business \& Economic Statistics\/}~{\em 39\/}(3),
  816--832.

\bibitem[\protect\citeauthoryear{Shergin}{Shergin}{1993}]{shergin1993central}
Shergin, V. (1993).
\newblock Central limit theorem for finitely-dependent random variables.
\newblock {\em Journal of Soviet Mathematics\/}~{\em 67\/}(4), 3244--3248.

\bibitem[\protect\citeauthoryear{Sobel}{Sobel}{2006}]{sobel2006randomized}
Sobel, M.~E. (2006).
\newblock What do randomized studies of housing mobility demonstrate? causal
  inference in the face of interference.
\newblock {\em Journal of the American Statistical Association\/}~{\em
  101\/}(476), 1398--1407.

\bibitem[\protect\citeauthoryear{Wager and Athey}{Wager and
  Athey}{2018}]{wager2018estimation}
Wager, S. and S.~Athey (2018).
\newblock Estimation and inference of heterogeneous treatment effects using
  random forests.
\newblock {\em Journal of the American Statistical Association\/}~{\em
  113\/}(523), 1228--1242.

\bibitem[\protect\citeauthoryear{Wang}{Wang}{2003}]{wang2003marginal}
Wang, N. (2003).
\newblock Marginal nonparametric kernel regression accounting for
  within-subject correlation.
\newblock {\em Biometrika\/}~{\em 90\/}(1), 43--52.

\end{thebibliography}

\section{Appendix}
\begin{alphasection}
\numberwithin{equation}{section}

\section{Simulation Results}\label{appendsim}
\subsection{Asymptotic-based Inference}
\begin{table}[H]
\caption{Empirical Rejection Probabilities with $n_c=10,$ $C=150,$ $h=\kappa_h\hat{s}_XC^{-2/7}$ and $\Pi =(0.3,0.4,0.5, 0.6).$}
\label{power table:1}
\begin{adjustbox}{width=\textwidth}
\begin{tabular}{ |p{3cm} p{3cm} p{3cm} p{3cm} p{3cm}|  }
 \hline
  \cline{3-5}
 &\multicolumn{4}{c|}{Nominal probabilities} \\
  \cline{3-5}

 \multicolumn{1}{|c}{Test statistic}&$\beta_1$ & 0.01 &0.05&0.10\\
 \hline
$\hat{S}_{1}$& 0.50 & 1.000 &1.000 &1.000\\
&0.45&1.000  &1.000   &1.000\\
&0.40& 0.998 & 1.000&1.000\\
&0.35&0.979& 0.998&0.998\\
&0.30& 0.891 & 0.964&  0.982\\
&0.25& 0.675 &0.840 &  0.910\\
&0.20&0.356 & 0.603&  0.723\\
&0.15&0.141 & 0.315&  0.452\\
&0.10&0.050& 0.140&  0.227\\
&0.05&0.015 & 0.066&  0.110\\
\hline
&0.00   &0.010   &0.044&   0.091\\
\hline
&-0.05 &0.014 & 0.063  &0.121\\
&-0.10 &0.050  &0.147   &0.226\\
&-0.15 &0.153 & 0.337  &0.443\\
&-0.20 &0.366 & 0.598 &0.741\\
&-0.25 &0.681 &0.868   &0.921  \\
&-0.30 &0.901 &0.959   &0.980  \\
&-0.35&0.980 &0.996   &0.999\\
&-0.40&0.998 &0.999  &1.000\\
&-0.45&1.000 &1.000  &1.000\\
&-0.50&1.000 &1.000  &1.000\\
 \hline
\end{tabular}
\end{adjustbox}
\end{table} 

\begin{table}[H]
\caption{Empirical Rejection Probabilities: $n_c=10, C=150,$ $h= C_h\hat{s}_XN^{-2/7}$ and $\Pi =(0.3,0.4,0.5, 0.6).$ All covariances accounted for in the variance estimators.}
\label{power table:2}
\begin{adjustbox}{width=\textwidth}
\begin{tabular}{ |p{3cm} p{3cm} p{3cm} p{3cm} p{3cm}|  }
 \hline
  \cline{3-5}
 &\multicolumn{4}{c|}{Nominal probabilities} \\
  \cline{3-5}
 \multicolumn{1}{|c}{Test statistic}&$\beta_0$ & 0.01 &0.05&0.10\\
 \hline
$\hat{S}_{2}$&0.50& 1.000  &1.000 &1.000\\
&0.45&1.000   &1.000   &1.000 \\
&0.40&1.000   &1.000   &1.000 \\
&0.35& 1.000 & 1.000 &1.000 \\
&0.30&1.000  &1.000   &1.000\\
&0.25&0.995& 0.998&0.998\\
&0.20&0.996  &1.000  &1.000\\
&0.15 & 0.818 & 0.950&  0.983\\
&0.10&0.244  &0.527   &0.695\\
&0.05&0.018  &0.119  &0.215\\
\hline
&0.00   &0.005  &0.036&   0.088\\
\hline
&-0.05&0.028  &0.124   &0.241\\
&-0.10 &0.279& 0.538  &0.713\\
&-0.15&0.828  &0.953   &0.984\\
&-0.20 &0.996  &1.000  &1.000\\
&-0.25&1.000  &1.000   &1.000\\
&-0.30 &1.000 & 1.000  &1.000\\
&-0.35&1.000  &1.000  &1.000\\
&-0.40 &1.000 & 1.000  &1.000\\
&-0.45&1.000  &1.000  &1.000\\
&-0.50 &1.000 &1.000   &1.000  \\
 \hline
\end{tabular}
\end{adjustbox}
\end{table}

\begin{table}[H]
\caption{Empirical Rejection Probabilities: $n_c=10, C=150,$ $h= C_h\hat{s}_XN^{-2/7}$ and $\Pi =(0.3,0.4,0.5, 0.6).$ No covariances accounted for in the variance estimators.}
\label{power table:2b}
\begin{adjustbox}{width=\textwidth}
\begin{tabular}{ |p{3cm} p{3cm} p{3cm} p{3cm} p{3cm}|  }
 \hline
  \cline{3-5}
 &\multicolumn{4}{c|}{Nominal probabilities} \\
  \cline{3-5}
 \multicolumn{1}{|c}{Test statistic}&$\beta_0$ & 0.01 &0.05&0.10\\
 \hline
$\hat{S}_{2}$&0.50& 1.000  &1.000 &1.000\\
&0.45&1.000   &1.000   &1.000 \\
&0.40&1.000   &1.000   &1.000 \\
&0.35& 1.000 & 1.000 &1.000 \\
&0.30&1.000  &1.000   &1.000\\
&0.25&0.995& 0.998&0.998\\
&0.20&0.998  &1.000  &1.000\\
&0.15 & 0.827 & 0.926&  0.999\\
&0.10&0.234  &0.434   &0.533\\
&0.05&0.015  &0.059  &0.102\\
\hline
&0.00   &0.004  &0.011&   0.024\\
\hline
&-0.05&0.021  &0.063   &0.111\\
&-0.10 &0.265& 0.448 &0.547\\
&-0.15&0.841  &0.928   &0.962\\
&-0.20 &0.998  &1.000  &1.000\\
&-0.25&1.000  &1.000   &1.000\\
&-0.30 &1.000 & 1.000  &1.000\\
&-0.35&1.000  &1.000  &1.000\\
&-0.40 &1.000 & 1.000  &1.000\\
&-0.45&1.000  &1.000  &1.000\\
&-0.50 &1.000 &1.000   &1.000  \\
 \hline
\end{tabular}
\end{adjustbox}
\end{table}

\subsection{Bootstrap-based Inference}
\begin{table}[H]
\caption{Empirical Rejection Probabilities with $n_c=10,$ $C=150,$ $h=\kappa_h\hat{s}_XC^{-2/7}$ and $\Pi =(0.3,0.4,0.5, 0.6).$}
\label{power table boot:1}
\begin{adjustbox}{width=\textwidth}
\begin{tabular}{ |p{3cm} p{3cm} p{3cm} p{3cm} p{3cm}|  }
 \hline
  \cline{3-5}
 &\multicolumn{4}{c|}{Nominal probabilities} \\
  \cline{3-5}

 \multicolumn{1}{|c}{Test statistic}&$\beta_1$ & 0.01 &0.05&0.10\\
 \hline
$\hat{S}_{1}$& 0.50 & 1.000 &1.000 &1.000\\
&0.45&1.000  &1.000   &1.000\\
&0.40& 0.997 & 0.999&1.000\\
&0.35&0.970& 0.997&0.999\\
&0.30& 0.896 & 0.967&  0.986\\
&0.25& 0.668 &0.865 &  0.923\\
&0.20&0.396 & 0.649&  0.744\\
&0.15&0.162 & 0.341&  0.472\\
&0.10&0.051& 0.163&  0.256\\
&0.05&0.014 & 0.065&  0.134\\
\hline
&0.00   &0.007   &0.046  & 0.098\\
\hline
&-0.05 &0.012 & 0.076  &0.152\\
&-0.10 &0.051  &0.165   &0.270\\
&-0.15 &0.179 & 0.355  &0.500\\
&-0.20 &0.405 & 0.657 &0.760\\
&-0.25 &0.691 &0.877   &0.934  \\
&-0.30 &0.903 &0.969   &0.985  \\
&-0.35&0.976 &0.999   &1.000\\
&-0.40&0.998 &1.000  &1.000\\
&-0.45&1.000 &1.000  &1.000\\
&-0.50&1.000 &1.000  &1.000\\
 \hline
\end{tabular}
\end{adjustbox}
\end{table} 

\begin{table}[H]
\caption{Empirical Rejection Probabilities with $n_c=10,$ $C=150,$ $h=\kappa_h\hat{s}_XC^{-2/7}$ and $\Pi =(0.3,0.4,0.5, 0.6).$}
\label{power table boot:2}
\begin{adjustbox}{width=\textwidth}
\begin{tabular}{ |p{3cm} p{3cm} p{3cm} p{3cm} p{3cm}|  }
 \hline
  \cline{3-5}
 &\multicolumn{4}{c|}{Nominal probabilities} \\
  \cline{3-5}

 \multicolumn{1}{|c}{Test statistic}&$\beta_0$ & 0.01 &0.05&0.10\\
 \hline
$\hat{S}_{2}$& 0.50 & 1.000 &1.000 &1.000\\
&0.45&1.000  &1.000   &1.000\\
&0.40&1.000  &1.000   &1.000\\
&0.35&1.000& 1.000&1.000\\
&0.30& 0.999 &1.000&  1.000\\
&0.25& 0.997 &1.000 & 1.000\\
&0.20&0.935 & 0.985&  0.992\\
&0.15&0.602 & 0.820&  0.894\\
&0.10&0.191& 0.414&  0.562\\
&0.05&0.029 & 0.117&  0.197\\
\hline
&0.00 &0.012   &0.050   &0.101 \\
\hline
&-0.05 &0.030 & 0.114  &0.198\\
&-0.10 &0.196  &0.430   &0.571\\
&-0.15 &0.628 & 0.850  &0.905\\
&-0.20 &0.997 & 1.000  &1.000\\
&-0.25 &0.997 & 1.000  &1.000  \\
&-0.30 &1.000 & 1.000  &1.000  \\
&-0.35&1.000 &1.000   &1.000\\
&-0.40&1.000 &1.000  &1.000\\
&-0.45&1.000 &1.000  &1.000\\
&-0.50&1.000 &1.000  &1.000\\
 \hline
\end{tabular}
\end{adjustbox}
\end{table} 

\subsection{ Bootstrap versus Asymptotic in Small Samples} 
 In Table \ref{comparison-bootasym}, I compare the empirical sizes of the bootstrap-based test statistics with their asymptotic counterparts when the sample size is small, precisely when $C=50$ and $N_c=10).$ The result shows that the  empirical sizes computed using the bootstrapping algorithms and their asymptotic counterparts are close in general.

\begin{table}[H]
\caption{
Comparison of empirical size for the bootstrap and asymptotic-based testing approach  when sample size is small: $n_c=10, C=50,$ $h=\kappa_h\cdot\hat{s}_XC^{-2/7},$ and $\Pi=(0.3,0.4,0.5,0.6).$}
\label{comparison-bootasym}
\centering
\begin{adjustbox}{width=\textwidth}
  \begin{tabular}{lSSSSSS}
    \toprule
    \multirow{2}{*}{Nominal probabilities} &
 \multicolumn{2}{c}{\underline{Test statistic for $H^{\Pi}_{0}$}} &
      \multicolumn{2}{c}{\underline{Test statistic for $H^{X}_{0}$}} \\\\
      & {Bootstrap-based} & {Asymptotic-based} & {Bootstrap-based} & {Asymptotic-based}  \\
      \midrule
    0.01 & 0.012 & 0.029 & 0.017 & 0.013 \\
    0.05 & 0.052 & 0.094 & 0.055& 0.050 \\
    0.10 & 0.106  & 0.158  & 0.106 & 0.128 \\
 
    \bottomrule
  \end{tabular}
  \end{adjustbox}
\end{table}

\subsection{Extension of the Monte Carlo Simulation Experiment to Multivariate Covariates} \label{multivarX}
 I extend the experiment in Section \ref{Monte Carlo Simulation} to multivariate pre-treatment variables. For both test statistics, each pre-treatment variable $X_d, $ $d>1$ is drawn independently from the standard uniform distribution that allows for within-cluster dependence. Each cluster is assigned one of two treatment vectors $\mathbf{T}(k).$ Here,the average of $\mathbf{T}(k)$ is either 0.3 or  0.4. I use the Monte Carlo integration technique to compute all integrals. The general  functional form of the CATE is
 \begin{equation} \label{gen-cate}
     \tau(\mathbf{x}, \pi) = \beta_0 \sum_{l=1}^d x_l + \beta_1 \pi. 
 \end{equation}
I restrict attention to the case where $d=2$ and defer cases with $d>2$ to a companion paper (in progress), which provides an \texttt{R} package of the testing procedures in this paper. This \texttt{R} package will be available for download and use. I keep the remaining design as in Section \ref{size and power sim}.

Focusing on $\hat{S}_1$, fix  $\beta_0 =1$ and $\beta_1 =0$  ( i.e., the null hypothesis of CTEs by $\Pi$ is true); and $\beta_1=0.5,$ (which implies that the null hypothesis of CTEs by $\Pi$ is false). Due to the curse of dimensionality associated
with kernel estimation, one should expect a poor performance of the tests when the dimension of continuous variables increases.  Table \ref{multix 1} reports the rejection probabilities of  $\hat{S}_1$ under the two  $\beta_2$ specifications, which give us the empirical size and power, respectively.
\begin{table}[H] 
\caption{Empirical size and power of $\hat{S}_1$ using multivariate $X.$ $C=200, N_c=10,$  bandwidth=$5\cdot\hat{s}_XN^{-2/7}$ and $\Pi=(0.3,0.4).$}
\label{multix 1}
\centering
  \begin{tabular}{lSSSS}
    \toprule
    \multirow{2}{*}{Nominal probabilities} &
 \multicolumn{2}{c}{\underline{$d=2$}} &\\
      & {Size } & {Power }\\ 
      \midrule
    0.01 & 0.013 & 0.292\\  
    0.05 & 0.053 & 0.444 \\  
    0.10 & 0.086& 0.533 \\    
    \bottomrule
  \end{tabular}
\end{table}

Next,  I turn  attention to $\hat{S}_2$. Using the CATE in \eqref{gen-cate}, fix  $\beta_0 =0$ and $\beta_1 =1$ ( i.e., the null hypothesis of CTEs by $X$ is true); and $\beta_0=0.5,$ (which implies that the null hypothesis of CTEs by $X$ is false). In Table \ref{multix 2}, I report the rejection probabilities of  $\hat{S}_2$ under the two specifications.

\begin{table}[H]
\caption{Empirical size and power of $\hat{S}_2$ using multivariate $X.$ $C=200,N_c=10,$ bandwidth=$5\cdot\hat{s}_XN^{-2/7}$ and $\Pi=(0.3,0.4).$}
\label{multix 2}
\centering
  \begin{tabular}{lSSSS}
    \toprule
    \multirow{2}{*}{Nominal probabilities} &
 \multicolumn{2}{c}{\underline{$d=2$}} \\
     & {Size } & {Power }  \\
      \midrule
    0.01 & 0.040 & 1.000\\
    0.05 & 0.097 & 1.000 \\
    0.10 & 0.018 & 1.000\\
    \bottomrule
  \end{tabular}
\end{table}
The empirical size and power calculations in Tables \ref{multix 1} and \ref{multix 2} numerically show that the proposed test statistics have non-zero power and are valid for a multivariate $X \in \mathcal{X}^2.$ The empirical size and power calculations in Tables \ref{multix 1} and \ref{multix 2} numerically show that the proposed test statistics have non-zero power and are valid for a multivariate $X \in \mathcal{X}^2.$ However, as highlighted in the Remark \ref{COD}, due to the curse of dimensionality, we observe some size distortion and loss in power.

\section{Proof of Main Results}
\subsection{Proof of Proposition \ref{identification}}
\noindent \textbf{Proposition 1.} 
\textit{ Suppose Assumptions \ref{Treatment-invariant neighborhoods}-\ref{Overlap} hold, then
\begin{align*}
  \tau(x;\pi_k )\coloneq& \mathbbm{E} [Y(1, \pi_k )|X=x] -\mathbbm{E} [Y(0, \pi_k )|X=x] \nonumber\\
  =& \mathbbm{E} [Y|T=1, \Pi=\pi_k, X=x] -\mathbbm{E} [Y|T=0, \Pi=\pi_k,X=x] \,\, \forall k=1,\dots,K \,\,\text{and}\,\, x \in \mathcal{X}.
  \end{align*}}
\begin{proof}
    Under Assumptions \ref{Treatment-invariant neighborhoods}-\ref{Discrete Network Exposure} and \ref{Unconfoundedness}(i) realized outcomes can be written in terms of potential outcomes as:
 \begin{equation}
     Y=\sum_{k=1}^K \Bigg( Y(0,\pi_k )  +  [ Y(1,\pi_k )-Y(0,\pi_k )] \cdot T\Bigg)\cdot \mathbbm{1}(\Pi=\pi_k ) .
 \end{equation}
 Hence, by Assumption \ref{Overlap}, $\mathbbm{E} [Y|T=t, \Pi=\pi_k, X=x]$ exists for all $t=0,1$ $k=1,\dots, K$ and $x \in \mathcal{X},$ with
 \begin{align}
   \mathbbm{E} [Y|T=1, \Pi=\pi_k, X=x]=&\mathbbm{E}\left[\sum_{k=1}^K \Bigg( Y(0,\pi_k )  +  [ Y(1,\pi_k )-Y(0,\pi_k )] \cdot T\Bigg)\cdot \mathbbm{1}(\Pi=\pi_k )\Big|T=1, \Pi=\pi_k, X=x\right] \nonumber\\
   =& \mathbbm{E}\left[  Y(1,\pi_k )|T=1, \Pi=\pi_k, X=x\right] \nonumber\\
 =& \mathbbm{E}\left[  Y(1,\pi_k )| X=x\right] \,\,\,\,\,\,\,\,\,\,\,\,\,\,\,\,\,\,\,\,\,\,\,\,\,\,\,[\text{by Assumption \ref{Unconfoundedness}(ii)}] \label{mean 1}
 \end{align}
and 
 \begin{align}
   \mathbbm{E} [Y|T=0, \Pi=\pi_k, X=x]=&\mathbbm{E}\left[\sum_{k=1}^K \Bigg( Y(0,\pi_k )  +  [ Y(1,\pi_k )-Y(0,\pi_k )] \cdot T\Bigg)\cdot \mathbbm{1}(\Pi=\pi_k )\Big|T=0, \Pi=\pi_k, X=x\right] \nonumber\\
   =& \mathbbm{E}\left[  Y(0,\pi_k )|T=0, \Pi=\pi_k, X=x\right] \nonumber\\
 =& \mathbbm{E}\left[  Y(0,\pi_k )| X=x\right] \,\,\,\,\,\,\,\,\,\,\,\,\,\,\,\,\,\,\,\,\,\,\,\,\,\,\,[\text{by Assumption \ref{Unconfoundedness}(ii)}] \label{mean 2}
 \end{align}
Hence, from \eqref{mean 1} and \eqref{mean 2}, we have 
$$ \mathbbm{E} [Y|T=1, \Pi=\pi_k, X=x] -\mathbbm{E} [Y|T=0, \Pi=\pi_k,X=x]=\mathbbm{E} [Y(1, \pi_k )|X=x] -\mathbbm{E} [Y(0, \pi_k )|X=x]$$ as required.

\end{proof}

\subsection{Asymptotic Variance and Bias Derivations} \label{bias and var}
 \begin{lemma}\label{Gine}
  Suppose $\mathcal{H}$ is a finite class of uniformly bounded real-valued functions $H$, equal to zero outside a known compact set. Further, let $g(x)f(x)$ be  continuously differentiable in $x$ with $\sup_{x\in B}\left|\frac{d(g( x)f(x))}{dx}\right|<\infty$ where $B \subset \mathbbm{R}^d$ is a compact set. Then uniformly in $H \in \mathcal{H}$
  \begin{align}
      \sup_{x\in B }\left|\frac{1}{h^d}\int_{\mathbbm{R}^d} g(z)f(z)H\left(\frac{x-z}{h}\right)dz-g(x)f(x)\int_{\mathbbm{R}^d}H\left(\frac{x-z}{h}\right)dz \right| \to 0 \,\,\text{as}\, h\to 0.
  \end{align}
 \end{lemma}
 \begin{proof}
     This lemma is similar to Lemma 6.1 in \cite{gine2003bm} and Lemma B.9 in \cite{chang2015nonparametric}. The proof given in \cite{gine2003bm}.
\end{proof}
\subsubsection{Test statistics $\hat{T}_1$}
Define
\begin{align*}
  \hat{\Gamma}(x;\pi_k,\pi_j) \coloneqq&\hat{\tau}(x;\pi_k)-\hat{\tau}(x;\pi_j)\\ =&\frac{1}{Nh^d}\sum_{i=1}^{N} Y_i\left[ \mathbbm{1}(\Pi_i=\pi_k)\hat{\phi}(T_i,x,\pi_k) - \mathbbm{1}(\Pi_i=\pi_j)\hat{\phi}(T_i,x,\pi_j)\right]K\left(\frac{x-X_i}{h}\right).
\end{align*}
Under the null hypothesis, the bias of $\hat{T}_1$ is define as
\begin{align*}
   \mathrm{Bias}(\hat{T}_1):=&\mathbbm{E}[\hat{T}_1]\\
   =&\mathbbm{E}\left[\int_{\mathcal{X}} \sum_{k =1}^K \sum_{j>k}^K \left\{\sqrt{N}| \hat{\Gamma}(x;\pi_k,\pi_j)|\right\}\hat w(x, \pi_k, \pi_j)dx\right]\\
   =&\int_{\mathcal{X}} \sum_{k =1}^K \sum_{j>k}^K\mathbbm{E}\left[ \sqrt{N}|\hat{\Gamma}(x;\pi_k,\pi_j)|\right] \hat w(x, \pi_k, \pi_j)dx\\
   =&\frac{1}{\sqrt{h^d}}\int_{\mathcal{X}}  \sum_{k =1}^K \sum_{j>k}^K \left\{\mathbbm{E}\left|(\hat{\Gamma}(x;\pi_k,\pi_j))\sqrt{Nh^d}\hat w(x, \pi_k, \pi_j) \right|\right\}dx\\
   \to&\frac{1}{\sqrt{h^d}}\int_{\mathcal{X}}  \sum_{k =1}^K \sum_{j>k}^K \left\{\mathbbm{E}\left|\frac{(\hat{\Gamma}(x;\pi_k,\pi_j))}{\sqrt{\mathrm{Var}(\hat{\Gamma}(x;\pi_k,\pi_j))}}\right|\right\} dx \,\,\,\,\,\,\,\,\text{[by Assumption \ref{regularity conditions}(h)]}\\
    =& \frac{1}{\sqrt{h^d}}\int_{\mathcal{X}}   \sum_{k =1}^K \sum_{j>k}^K \mathbbm{E}\left|\mathbbm{Z}_1\right|dx\\
   =& h^{\frac{-d}{2}}\cdot\mathbbm{E}|\mathbbm{Z}_1|\cdot \frac{K(K-1)}{2}\cdot \int_{\mathcal{X}} dx \\
   =& a_1.
\end{align*}
On the other hand, the variance of $\hat{T}_1$  under the null is               
\begin{align*}
Var(\hat T_1)\coloneqq& \int_1 \mathrm{Cov}(|\sqrt{N}\hat{\Gamma}(x;\pi_i,\pi_j)\hat w(x,\pi_i, \pi_j)|,|\sqrt{N}\hat{\Gamma}(x';\pi_k,\pi_l))\hat w(x',\pi_k, \pi_l)|)dxdx' \\
=& \frac{1}{h^d}\int_1 \mathbbm{1}\left(\frac{x-x'}{h} \in [-1,1]^d\right)\mathrm{Cov}(|\sqrt{Nh^d}\hat{\Gamma}(x;\pi_i,\pi_j)\hat w(x,\pi_i, \pi_j)|,|\sqrt{Nh^d}\hat{\Gamma}(x';\pi_k,\pi_l))\\
&\cdot \hat w(x',\pi_k, \pi_l)|)dxdx', 
\end{align*}
\normalsize
where $\int_1\coloneqq\int_{\mathbbm{R}^{d}}\int_{\mathbbm{R}^{d}}\sum_{i=1}^K  \sum_{j=1}^K \sum_{k=1}^K \sum_{l=1}^K.$

Now, let $(Z_{1n}(x, \pi_i, \pi_j), Z_{2n}(x', \pi_k,\pi_l)), x,x' \in \mathbbm{R}^d$ be mean zero, bivariate Gaussian process such that for each $x \in  \mathbbm{R}^d$ and $x' \in  \mathbbm{R}^d,$ $(Z_{1n}(x, \pi_i, \pi_j), Z_{2n}(x', \pi_k,\pi_l))$ and ($\sqrt{Nh^d}\hat{\Gamma}(x;\pi_i,\pi_j)\hat w(x,\pi_i, \pi_j),$ $\sqrt{Nh^d}\hat{\Gamma}(x';\pi_k,\pi_l))\hat w(x',\pi_k, \pi_l)$ ) have the same covariance structure. Thus,  by the "bivariate normal distribution generation formula",
\begin{equation}\label{bivariate normal distribution generation}
    (Z_{1n}(x, \pi_i, \pi_j), Z_{2n}(x', \pi_k,\pi_l))\overset{d}{=}\left(\sqrt{1-\rho^*(x,x', \pi_i, \pi_j, \pi_k, \pi_l)^2}\mathbbm{Z}_1 +\rho^*(x,x', \pi_i, \pi_j, \pi_k, \pi_l)\mathbbm{Z}_2,\mathbbm{Z}_2\right)
\end{equation}
where $\mathbbm{Z}_1$ and $\mathbbm{Z}_2$ are independent standard normal random variables and 
$$\rho^*(x,x', \pi_i, \pi_j, \pi_k, \pi_l):=\mathrm{Corr}[\sqrt{Nh^d}(\hat{\Gamma}(x;\pi_i,\pi_j))\hat w(x,\pi_i, \pi_j),\sqrt{Nh^d}(\hat{\Gamma}(x';\pi_k,\pi_l))\hat w(x',\pi_k, \pi_l)]$$

Let $$\bar{s}^2_{1}= \frac{1}{h^d}\int_1\mathbbm{1}\left(\frac{x-x'}{h} \in [-1,1]^d\right) \mathrm{Cov}(|Z_{1n}(x, \pi_i, \pi_j)|,|Z_{2n}(x', \pi_k, \pi_l)|)dxdx' $$
Applying the change of variable $x'=x+th,$  we have  
$$\bar{s}^2_{1}= \int_{\mathbbm{R}^{d}}\int_{[-1,1]^d}\sum_{i=1}^K  \sum_{j=1}^K \sum_{k=1}^K \sum_{l=1}^K \mathrm{Cov}(|Z_{1n}(x, \pi_i, \pi_j)|,|Z_{2n}(x+th, \pi_k, \pi_l)|) dxdt,  $$
Furthermore, from Assumption \ref{regularity conditions}(h) (i.e.,$\sqrt{Nh^d}\hat w(x,\pi, \pi')\to (Var(\hat{\Gamma}(x;\pi,\pi')))^{-1/2}$  ), Lemma \ref{Gine} (which holds under Assumption \ref{regularity conditions}(e))  and  the change of variable $x'=x+th,$ we have that for almost every $(x,t),$
$$\rho^*(x,x', \pi_i, \pi_j, \pi_k, \pi_l)\to \rho(x,t, \pi_i, \pi_j, \pi_k, \pi_l)$$
with 
\begin{align*}
   \rho(x, t,\pi_i,\pi_j,\pi_k,\pi_l)=&\frac{\mathrm{Cov}(\hat{\Gamma}(x,\pi_i,\pi_j), \hat{\Gamma}(x+th;\pi_k,\pi_l))}{\sqrt{Var(\hat{\Gamma}(x;\pi_i,\pi_j))}\sqrt{Var(\hat{\Gamma}(x+th;\pi_k,\pi_l))}}\\
 =&  \begin{cases} 
\frac{\int K(\xi)K(\xi+t)d\xi}{\int K(\xi)^2 d\xi} &\,\, \text{if}\,\, i=k \,\& \,j=l \\
\frac{-\int K(\xi)K(\xi+t)d\xi}{\int K(\xi)^2 d\xi} &\,\, \text{if}\,\, i=l \,\& \,j=k \\
 \frac{-Var(\hat{\tau}(x, \pi_j))}{\sqrt{Var(\hat{\Gamma}(x;\pi_i,\pi_j))}\sqrt{Var(\hat{\Gamma}(x;\pi_k,\pi_l))}}\cdot \frac{\int K(\xi)K(\xi+t)d\xi}{\int K(\xi)^2 d\xi}& \,\, \text{if}\,\, j=k \,\& \,i\neq l \\
 \frac{-Var(\hat{\tau}(x, \pi_i))}{\sqrt{Var(\hat{\Gamma}(x;\pi_i,\pi_j))}\sqrt{Var(\hat{\Gamma}(x;\pi_k,\pi_l))}}\cdot \frac{\int K(\xi)K(\xi+t)d\xi}{\int K(\xi)^2 d\xi}& \,\, \text{if}\,\, j\neq k \,\& \, i=l\\
 \frac{Var(\hat{\tau}(x, \pi_j))}{\sqrt{Var(\hat{\Gamma}(x;\pi_i,\pi_j))}\sqrt{Var(\hat{\Gamma}(x;\pi_k,\pi_l))}}\cdot \frac{\int K(\xi)K(\xi+t)d\xi}{\int K(\xi)^2 d\xi}& \,\, \text{if}\,\, j=l  \,\& \, i\neq k\\
 \frac{Var(\hat{\tau}(x, \pi_i))}{\sqrt{Var(\hat{\Gamma}(x;\pi_i,\pi_j))}\sqrt{Var(\hat{\Gamma}(x;\pi_k,\pi_l))}}\cdot \frac{\int K(\xi)K(\xi+t)d\xi}{\int K(\xi)^2 d\xi}&  \,\, \text{if}\,\,j\neq l  \,\& \, i= k\\
 0 & \,\, \text{otherwise}\,\,
    \end{cases} 
\end{align*}
Hence, based on \eqref{bivariate normal distribution generation}, $$\mathrm{Cov}(|Z_{1n}(x)|,|Z_{2n}(x+th)|)\to \mathrm{Cov}\left(\left|\sqrt{1-\rho(x,t, \pi_i, \pi_j, \pi_k, \pi_l)^2}\mathbbm{Z}_1 +\rho(x,t, \pi_i, \pi_j, \pi_k, \pi_l)\mathbbm{Z}_2\right|,\left|\mathbbm{Z}_2\right|\right)$$
Finally, as in the proof of (6.35) in \cite{gine2003bm}, using the bounded convergence theorem with respect to the Lebesgue measure, we have $\bar{s}^2_{1}- Var(\hat T_1)\to 0,$ as such we have 
\begin{align*}
&\lim_{N\to \infty} Var(\hat T_1) \\
      &= \int_{\mathbbm{R}^{d}}\int_{[-1,1]^d}\sum_{i=1}^K  \sum_{j>i}^K \sum_{k=1}^K \sum_{l>k}^K \mathrm{Cov}\left(\left|\sqrt{1-\rho(x,t, \pi_i, \pi_j, \pi_k, \pi_l)^2}\mathbbm{Z}_1 +\rho(x,t, \pi_i, \pi_j, \pi_k, \pi_l)\mathbbm{Z}_2\right|,\left|\mathbbm{Z}_2\right|\right) dxdt\\ 
     & =: \sigma_1^2.
\end{align*}
Recall that all the asymptotic results are based on equal cluster sizes. As a result, I safely interchange $C\to \infty$ and  $N=N_0C\to \infty$  unless there is a confusion. 
Plugging  the kernel-based estimates of the variances into the formula for $\rho(x, t,\pi_i,\pi_j,\pi_k,\pi_l),$ gives us the plug-in estimator $\hat{\rho}(x, t,\pi_i,\pi_j,\pi_k,\pi_l).$
Hence, a consistent estimator of the asymptotic variance  of $\hat T_1$is 
\begin{equation*}
\resizebox{1.1\hsize}{!}{
    $\hat\sigma^2_{1}= \int_{\mathbbm{R}^{d}}\int_{\mathbbm{R}^d}\sum_{i=1}^K  \sum_{j>i}^K \sum_{k=1}^K \sum_{l>k}^K \mathrm{Cov}\left(\left|\sqrt{1-\hat \rho(x,t, \pi_i, \pi_j, \pi_k, \pi_l)^2}\mathbbm{Z}_1 +\hat \rho(x,t, \pi_i, \pi_j, \pi_k, \pi_l)\mathbbm{Z}_2\right|,\left|\mathbbm{Z}_2\right|\right) dxdt.$ 
    }
\end{equation*}

\subsubsection{Test statistics $\hat{T}_2$}
Define
\begin{align*}
  \hat{\Gamma}(x;x',\pi) :=&\hat{\tau}(x;\pi_k)-\hat{\tau}(x';\pi_k) 
\end{align*}
Now under the null hypothesis, the bias of $\hat{T}_2$  as 
\begin{align*}
  \mathrm{Bias}(\hat{T}_2):=&\mathbbm{E}[\hat{T}_2]\\
   =&\mathbbm{E}\left[\int_{\mathcal{X}} \int_{\mathcal{X}} \sum_{k =1}^K \left\{\sqrt{N}| \hat{\Gamma}(x,x',\pi_k)|\right\} \frac{\hat w(x,x',\pi_k)}{2}dxdx'\right]\\
   =&\frac{1}{\sqrt{h^d}}\int_{\mathcal{X}} \int_{\mathcal{X}} \sum_{k =1}^K\mathbbm{E}\left[ \sqrt{Nh^d}|\hat{\Gamma}(x,x',\pi_k)|\right] \frac{\hat w(x,x',\pi_k)}{2}dxdx'\\
   \to&\frac{1}{2\sqrt{h^d}}\int_{\mathcal{X}} \int_{\mathcal{X}} \sum_{k =1}^K \left\{\mathbbm{E}\left|\frac{(\hat{\Gamma}(x,x',\pi_k))}{\sqrt{\mathrm{Var}(\hat{\Gamma}(x,x',\pi_k))}}\right|\right\} dxdx' \,\,\,\,\,\,\,\,\text{[by Assumption \ref{regularity conditions}(h)]}\\
    =&\frac{1}{2\sqrt{h^d}}\int_{\mathcal{X}} \int_{\mathcal{X}} \sum_{k =1}^K \mathbbm{E}\left|\mathbbm{Z}_1\right|dxdx'\\
   =& h^{\frac{-d}{2}}\cdot\mathbbm{E}|\mathbbm{Z}_1|\cdot \frac{K}{2}\cdot \int_{\mathcal{X}} \int_{\mathcal{X}}dx dx'=:  a_{2}.
\end{align*}
Also, the variance of  $\hat{T}_2$ is
\begin{align*}
    Var(\hat{T}_2):=& \frac{1}{4}\int_2 \mathrm{Cov}(\sqrt{N}|\hat{\Gamma}(x,x',\pi_k)|,|\sqrt{N}\hat{\Gamma}(x'',x''',\pi_k))|)\hat w(x,x',\pi_k) \hat w(x'',x''',\pi_k)dxdx'dx''dx''' \\
    =& \frac{1}{4h^d}\int_2 \mathrm{Cov}(\sqrt{Nh^d}|\hat{\Gamma}(x,x',\pi_k)|,|\sqrt{Nh^d}\hat{\Gamma}(x'',x''',\pi_k))|)\hat w(x,x',\pi_k) \hat w(x'',x''',\pi_k)dxdx'dx''dx''', 
\end{align*}
where 
$\int_2\coloneqq\int_{\mathbbm{R}^{d}}\int_{\mathbbm{R}^{d}} \int_{\mathbbm{R}^{d}}\int_{\mathbbm{R}^{d}}\sum_{k=1}^K.$ 
Now, note that the integral above is non-zero if at least one of the elements in the set  $$\textbf{H}:=\left\{\mathbbm{1}((x-x'')h^{-1} \in [-1,1]^d),\mathbbm{1}((x-x''')h^{-1} \notin [-1,1]^d),\mathbbm{1}((x'-x''){h} \in [-1,1]^d),\mathbbm{1}((x'-x''')h^{-1}\in [-1,1]^d)\right\}$$ equals 1. Thus, the variance becomes
\begin{align*}
      Var(\hat{T}_2) =& \frac{1}{4h^d}\int_2 \mathbbm{1}\left(\frac{x-x''}{h} \in [-1,1]^d\right)\\ 
    &\cdot \mathrm{Cov}(|\sqrt{Nh^d}\hat w(x,x',\pi_k)\hat{\Gamma}(x,x',\pi_k)|,|\sqrt{Nh^d}\hat w(x'',x''',\pi_k)\hat{\Gamma}(x'',x''',\pi_k))|) 
    dxdx'dx''dx'''\\
    &+  \frac{1}{4h^d}\int_2\mathbbm{1}\left(\frac{x-x'''}{h} \in [-1,1]^d\right)\\ 
    &\cdot \mathrm{Cov}(|\sqrt{Nh^d}\hat w(x,x',\pi_k)\hat{\Gamma}(x,x',\pi_k)|,|\sqrt{Nh^d}\hat w(x'',x''',\pi_k)\hat{\Gamma}(x'',x''',\pi_k))|) 
    dxdx'dx''dx'''\\
     &+  \frac{1}{4h^d}\int_2 \mathbbm{1}\left(\frac{x'-x''}{h} \in [-1,1]^d\right)\\ 
    &\cdot \mathrm{Cov}(|\sqrt{Nh^d}\hat w(x,x',\pi_k)\hat{\Gamma}(x,x',\pi_k)|,|\sqrt{Nh^d}\hat w(x'',x''',\pi_k)\hat{\Gamma}(x'',x''',\pi_k))|) 
    dxdx'dx''dx'''\\
    &+  \frac{1}{4h^d}\int_2 \mathbbm{1}\left(\frac{x'-x'''}{h} \in [-1,1]^d\right)\\ 
    &\cdot \mathrm{Cov}(|\sqrt{Nh^d}\hat w(x,x',\pi_k)\hat{\Gamma}(x,x',\pi_k)|,|\sqrt{Nh^d}\hat w(x'',x''',\pi_k)\hat{\Gamma}(x'',x''',\pi_k))|) 
    dxdx'dx''dx'''\\
    &+  \frac{1}{4h^d}\int_2 \mathbbm{1}\left(\frac{x'-x''}{h} \in [-1,1]^d\right)\mathbbm{1}\left(\frac{x'-x'''}{h} \in [-1,1]^d\right)\\ 
    &\cdot \mathrm{Cov}(|\sqrt{Nh^d}\hat w(x,x',\pi_k)\hat{\Gamma}(x,x',\pi_k)|,|\sqrt{Nh^d}\hat w(x'',x''',\pi_k)\hat{\Gamma}(x'',x''',\pi_k))|) 
    dxdx'dx''dx'''\\
    &+  \frac{1}{4h^d}\int_2 \mathbbm{1}\left(\frac{x-x'''}{h} \in [-1,1]^d\right)\mathbbm{1}\left(\frac{x'-x'''}{h} \in [-1,1]^d\right)\\ 
    &\cdot \mathrm{Cov}(|\sqrt{Nh^d}\hat w(x,x',\pi_k)\hat{\Gamma}(x,x',\pi_k)|,|\sqrt{Nh^d}\hat w(x'',x''',\pi_k)\hat{\Gamma}(x'',x''',\pi_k))|) 
    dxdx'dx''dx'''\\
    &+  \frac{1}{4h^d}\int_2 \mathbbm{1}\left(\frac{x-x'''}{h} \in [-1,1]^d\right)\mathbbm{1}\left(\frac{x'-x''}{h} \in [-1,1]^d\right)\\ 
    &\cdot \mathrm{Cov}(|\sqrt{Nh^d}\hat w(x,x',\pi_k)\hat{\Gamma}(x,x',\pi_k)|,|\sqrt{Nh^d}\hat w(x'',x''',\pi_k)\hat{\Gamma}(x'',x''',\pi_k))|) 
    dxdx'dx''dx'''\\
    &+  \frac{1}{4h^d}\int_2 \mathbbm{1}\left(\frac{x-x''}{h} \in [-1,1]^d\right)\mathbbm{1}\left(\frac{x'-x'''}{h} \in [-1,1]^d\right)\\ 
    &\cdot \mathrm{Cov}(|\sqrt{Nh^d}\hat w(x,x',\pi_k)\hat{\Gamma}(x,x',\pi_k)|,|\sqrt{Nh^d}\hat w(x'',x''',\pi_k)\hat{\Gamma}(x'',x''',\pi_k))|) 
    dxdx'dx''dx'''\\
    &+  \frac{1}{4h^d}\int_2 \mathbbm{1}\left(\frac{x-x''}{h} \in [-1,1]^d\right)\mathbbm{1}\left(\frac{x'-x''}{h} \in [-1,1]^d\right)\\ 
    &\cdot \mathrm{Cov}(|\sqrt{Nh^d}\hat w(x,x',\pi_k)\hat{\Gamma}(x,x',\pi_k)|,|\sqrt{Nh^d}\hat w(x'',x''',\pi_k)\hat{\Gamma}(x'',x''',\pi_k))|) 
    dxdx'dx''dx'''\\
    &+  \frac{1}{4h^d}\int_2 \mathbbm{1}\left(\frac{x-x''}{h} \in [-1,1]^d\right)\mathbbm{1}\left(\frac{x-x'''}{h} \in [-1,1]^d\right)\\ 
    &\cdot \mathrm{Cov}(|\sqrt{Nh^d}\hat w(x,x',\pi_k)\hat{\Gamma}(x,x',\pi_k)|,|\sqrt{Nh^d}\hat w(x'',x''',\pi_k)\hat{\Gamma}(x'',x''',\pi_k))|) 
    dxdx'dx''dx'''\\
    &+  \frac{1}{4h^d}\int_2 \mathbbm{1}\left(\frac{x-x'''}{h} \in [-1,1]^d\right)\mathbbm{1}\left(\frac{x'-x''}{h} \in [-1,1]^d\right)\mathbbm{1}\left(\frac{x'-x'''}{h} \in [-1,1]^d\right)\\ 
    &\cdot \mathrm{Cov}(|\sqrt{Nh^d}\hat w(x,x',\pi_k)\hat{\Gamma}(x,x',\pi_k)|,|\sqrt{Nh^d}\hat w(x'',x''',\pi_k)\hat{\Gamma}(x'',x''',\pi_k))|) 
    dxdx'dx''dx'''\\
    &+  \frac{1}{4h^d}\int_2 \mathbbm{1}\left(\frac{x-x''}{h} \in [-1,1]^d\right)\mathbbm{1}\left(\frac{x'-x''}{h} \in [-1,1]^d\right)\mathbbm{1}\left(\frac{x'-x'''}{h} \in [-1,1]^d\right)\\ 
    &\cdot \mathrm{Cov}(|\sqrt{Nh^d}\hat w(x,x',\pi_k)\hat{\Gamma}(x,x',\pi_k)|,|\sqrt{Nh^d}\hat w(x'',x''',\pi_k)\hat{\Gamma}(x'',x''',\pi_k))|) 
    dxdx'dx''dx'''\\
    &+  \frac{1}{4h^d}\int_2 \mathbbm{1}\left(\frac{x-x''}{h} \in [-1,1]^d\right)\mathbbm{1}\left(\frac{x-x'''}{h} \in [-1,1]^d\right)\mathbbm{1}\left(\frac{x'-x'''}{h} \in [-1,1]^d\right)\\ 
    &\cdot \mathrm{Cov}(|\sqrt{Nh^d}\hat w(x,x',\pi_k)\hat{\Gamma}(x,x',\pi_k)|,|\sqrt{Nh^d}\hat w(x'',x''',\pi_k)\hat{\Gamma}(x'',x''',\pi_k))|) 
    dxdx'dx''dx'''\\
    &+  \frac{1}{4h^d}\int_2 \mathbbm{1}\left(\frac{x-x''}{h} \in [-1,1]^d\right)\mathbbm{1}\left(\frac{x-x'''}{h} \in [-1,1]^d\right)\mathbbm{1}\left(\frac{x'-x''}{h} \in [-1,1]^d\right)\\ 
    &\cdot \mathrm{Cov}(|\sqrt{Nh^d}\hat w(x,x',\pi_k)\hat{\Gamma}(x,x',\pi_k)|,|\sqrt{Nh^d}\hat w(x'',x''',\pi_k)\hat{\Gamma}(x'',x''',\pi_k))|) 
    dxdx'dx''dx'''\\
    &+  \frac{1}{4h^d}\int_2 \mathbbm{1}\left(\frac{x-x''}{h} \in [-1,1]^d\right)\mathbbm{1}\left(\frac{x-x'''}{h} \in [-1,1]^d\right)\mathbbm{1}\left(\frac{x'-x''}{h} \in [-1,1]^d\right)\mathbbm{1}\left(\frac{x'-x'''}{h} \in [-1,1]^d\right)\\ 
    &\cdot \mathrm{Cov}(|\sqrt{Nh^d}\hat w(x,x',\pi_k)\hat{\Gamma}(x,x',\pi_k)|,|\sqrt{Nh^d}\hat w(x'',x''',\pi_k)\hat{\Gamma}(x'',x''',\pi_k))|) 
    dxdx'dx''dx'''.
\end{align*}
In large samples, any integrals where it is possible for more than one of the elements in $\textbf{H}$ to be equal to 1 will converge to zero as $C\to \infty$ ($h\to 0$). Theoretically, this is obvious after applying a change of variables. Intuitively, it is because the chance of having the scaled distance between two or more distinct points in the covariate space belonging to the hypercube goes to zero as $h\to 0.$  
Hence, 
\begin{align*}
    Var(\hat{T}_2)
    \to& \frac{1}{4h^d}\int_2 \mathbbm{1}\left(\frac{x-x''}{h} \in [-1,1]^d\right)\\ 
    &\cdot \mathrm{Cov}(|\sqrt{Nh^d}\hat w(x,x',\pi_k)\hat{\Gamma}(x,x',\pi_k)|,|\sqrt{Nh^d}\hat w(x'',x''',\pi_k)\hat{\Gamma}(x'',x''',\pi_k))|) 
    dxdx'dx''dx'''\\
    &+  \frac{1}{4h^d}\int_2\mathbbm{1}\left(\frac{x-x'''}{h} \in [-1,1]^d\right)\\ 
    &\cdot \mathrm{Cov}(|\sqrt{Nh^d}\hat w(x,x',\pi_k)\hat{\Gamma}(x,x',\pi_k)|,|\sqrt{Nh^d}\hat w(x'',x''',\pi_k)\hat{\Gamma}(x'',x''',\pi_k))|) 
    dxdx'dx''dx'''\\
     &+  \frac{1}{4h^d}\int_2 \mathbbm{1}\left(\frac{x'-x''}{h} \in [-1,1]^d\right)\\ 
    &\cdot \mathrm{Cov}(|\sqrt{Nh^d}\hat w(x,x',\pi_k)\hat{\Gamma}(x,x',\pi_k)|,|\sqrt{Nh^d}\hat w(x'',x''',\pi_k)\hat{\Gamma}(x'',x''',\pi_k))|) 
    dxdx'dx''dx'''\\
    &+  \frac{1}{4h^d}\int_2 \mathbbm{1}\left(\frac{x'-x'''}{h} \in [-1,1]^d\right)\\ 
    &\cdot \mathrm{Cov}(|\sqrt{Nh^d}\hat w(x,x',\pi_k)\hat{\Gamma}(x,x',\pi_k)|,|\sqrt{Nh^d}\hat w(x'',x''',\pi_k)\hat{\Gamma}(x'',x''',\pi_k))|) 
    dxdx'dx''dx'''\\
    &=:\sigma_2^2.
\end{align*}
Thus,
\begin{align*}
    \sigma^2_{2}
    =& \frac{1}{4h^d}\int_{2}\mathbbm{1}\left(\frac{x-x''}{h} \in [-1,1]^d\right) 
    \cdot \mathrm{Cov}\left(\left|\frac{\hat{\tau}(x,\pi_k)}{\sqrt{Var(\hat \Gamma(x,x',\pi_k))}}\right|,\left|\frac{\hat{\tau}(x'',\pi_k)}{\sqrt{Var(\hat \Gamma(x'',x''',\pi_k))}}\right|\right) 
    dxdx'dx''dx'''\\
    &+  \frac{1}{4h^d}\int_{2}  \mathbbm{1}\left(\frac{x-x'''}{h} \in [-1,1]^d\right)\cdot \mathrm{Cov}\left(\left|\frac{\hat{\tau}(x,\pi_k)}{\sqrt{Var(\hat \Gamma(x,x',\pi_k))}}\right|,\left|\frac{\hat{\tau}(x''',\pi_k)}{\sqrt{Var(\hat \Gamma(x'',x''',\pi_k))}}\right|\right)  dxdx'dx''dx''' \\ 
     &+  \frac{1}{4h^d}\int_{2} \mathbbm{1}\left(\frac{x'-x''}{h} \in [-1,1]^d\right) \mathrm{Cov}\left(\left|\frac{\hat{\tau}(x',\pi_k)}{\sqrt{Var(\hat \Gamma(x,x',\pi_k))}}\right|,\left|\frac{\hat{\tau}(x'',\pi_k)}{\sqrt{Var(\hat \Gamma(x'',x''',\pi_k))}}\right|\right)  dxdx'dx''dx''' \\ 
    &+  \frac{1}{4h^d}\int_{2} \mathbbm{1}\left(\frac{x'-x'''}{h} \in [-1,1]^d\right) \mathrm{Cov}\left(\left|\frac{\hat{\tau}(x',\pi_k)}{\sqrt{Var(\hat \Gamma(x,x',\pi_k))}}\right|,\left|\frac{\hat{\tau}(x''',\pi_k)}{\sqrt{Var(\hat \Gamma(x'',x''',\pi_k))}}\right|\right)  dxdx'dx''dx''' \\
     =& \frac{1}{4h^d}\int_{2}\mathbbm{1}\left(\frac{x-x''}{h} \in [-1,1]^d\right) \frac{\sqrt{Var(\hat{\tau}(x,\pi_k)) Var(\hat{\tau}(x'',\pi_k))}}{\sqrt{Var(\hat \Gamma(x,x',\pi_k)) Var(\hat \Gamma(x'',x''',\pi_k)}} \\
    & \cdot \mathrm{Cov}\left(\left|\frac{\hat{\tau}(x,\pi_k)}{\sqrt{Var(\hat \tau(x;\pi_k))}}\right|,\left|\frac{\hat{\tau}(x'',\pi_k)}{\sqrt{Var(\hat \tau(x'',\pi_k))}}\right|\right) 
    dxdx'dx''dx'''\\
    &+  \frac{1}{4h^d}\int_{2}  \mathbbm{1}\left(\frac{x-x'''}{h} \in [-1,1]^d\right)\frac{\sqrt{Var(\hat{\tau}(x,\pi_k)) Var(\hat{\tau}(x''',\pi_k))}}{\sqrt{Var(\hat \Gamma(x,x',\pi_k)) Var(\hat \Gamma(x'',x''',\pi_k)}} \\
    &\cdot \mathrm{Cov}\left(\left|\frac{\hat{\tau}(x,\pi_k)}{\sqrt{Var(\hat \tau(x;\pi_k))}}\right|,\left|\frac{\hat{\tau}(x''',\pi_k)}{\sqrt{Var(\hat \tau(x'',\pi_k))}}\right|\right)  dxdx'dx''dx''' \\ 
     &+  \frac{1}{4h^d}\int_{2} \mathbbm{1}\left(\frac{x'-x''}{h} \in [-1,1]^d\right)\frac{\sqrt{Var(\hat{\tau}(x',\pi_k)) Var(\hat{\tau}(x'',\pi_k))}}{\sqrt{Var(\hat \Gamma(x,x',\pi_k)) Var(\hat \Gamma(x'',x''',\pi_k)}} \\
     &\cdot \mathrm{Cov}\left(\left|\frac{\hat{\tau}(x',\pi_k)}{\sqrt{Var(\hat \tau(x',\pi_k))}}\right|,\left|\frac{\hat{\tau}(x'',\pi_k)}{\sqrt{Var(\hat \tau(x'',\pi_k))}}\right|\right)  dxdx'dx''dx''' \\ 
    &+  \frac{1}{4h^d}\int_{2} \mathbbm{1}\left(\frac{x'-x'''}{h} \in [-1,1]^d\right)\frac{\sqrt{Var(\hat{\tau}(x',\pi_k)) Var(\hat{\tau}(x''',\pi_k))}}{\sqrt{Var(\hat \Gamma(x,x',\pi_k)) Var(\hat \Gamma(x'',x''',\pi_k)}} \\
    &\cdot \mathrm{Cov}\left(\left|\frac{\hat{\tau}(x',\pi_k)}{\sqrt{Var(\hat \tau(x',\pi_k))}}\right|,\left|\frac{\hat{\tau}(x''',\pi_k)}{\sqrt{Var(\hat \tau(x''',\pi_k))}}\right|\right)  dxdx'dx''dx'''.
\end{align*}

Now, let $(Z_{1n}(x, \pi_k), Z_{2n}(x'',\pi_k)), x, x'' \in \mathbbm{R}^d$ be mean zero, bivariate Gaussian process such that for each $x \in  \mathbbm{R}^d$ and  $x'' \in  \mathbbm{R}^d,$ $(Z_{1n}(x,\pi_k), Z_{2n}(x'',\pi_k))$ and $(\sqrt{Nh^d}\hat{\tau}(x,\pi_k)w(x, \pi_k),$ $\sqrt{Nh^d}\hat{\tau}(x''\pi_k))\hat w(x'',\pi_k))$  have the same covariance structure. Thus,  by the "bivariate normal distribution generation formula,"
\begin{equation}\label{bivariate normal distribution generation 3}
    (Z_{1n}(x,\pi_k), Z_{2n}(x'',\pi_k))\overset{d}{=}\left(\sqrt{1-\rho^*(x,x'' \pi_k)^2}\mathbbm{Z}_1 +\rho^*(x, x'', \pi_k)\mathbbm{Z}_2,\mathbbm{Z}_2\right)
\end{equation}
where $\mathbbm{Z}_1$ and $\mathbbm{Z}_2$ are independent standard normal random variables and 
\begin{align*}
    \rho^*(x, x'', \pi_k):=&\mathrm{Corr}\left(\sqrt{Nh^d}(\hat{\tau}(x,\pi_k)-\mathbbm{E}[\hat{\tau}(x,\pi_k)])\hat w(x, \pi_k),\sqrt{Nh^d}(\hat{\tau}(x'',\pi_k)- \mathbbm{E}[\hat{\tau}(x'',\pi_k)])\hat w(x'',\pi_k)\right)\\
    =& \frac{\mathbbm{E}[(\hat{\tau}(x,\pi_k) -\mathbbm{E}[\hat{\tau}(x,\pi_k)]) (\hat{\tau}(x'',\pi_k)-\mathbbm{E}[\hat{\tau}(x'',\pi_k)])]}{\sqrt{Var(\hat{\tau}(x,\pi_k))}\sqrt{Var(\hat{\tau}(x'',\pi_k))}}
\end{align*}
Based on  Lemma \ref{Gine} (which holds under Assumption \ref{regularity conditions}(e))  and  the change of variable $x''=x+th,$ we have that for almost every $(x,t),$
\begin{align}\label{correlation}
  \rho^*(x, x+th, \pi_k)\to \rho(x, t, \pi_k):= \frac{\int K(\xi)K(\xi+t)d\xi}{\int K(\xi)^2 d\xi}  
\end{align}

Let 
\begin{align*}
   \bar{s}^2_{21}=& \frac{1}{4h^d} \int_{2}  \mathbbm{1}\left(\frac{x-x''}{h} \in [-1,1]^d\right) \frac{\sqrt{Var(\hat{\tau}(x,\pi_k)) Var(\hat{\tau}(x'',\pi_k))}}{\sqrt{Var(\hat \Gamma(x,x',\pi_k)) Var(\hat \Gamma(x'',x''',\pi_k)}}\\&\cdot\mathrm{Cov}(|Z_{1n}(x,\pi_k)|,|Z_{2n}(x'',\pi_k)|)dxdx'dx''dx''', 
\end{align*}
 Applying the change of variable $x''=x+th,$ 
we have,
\begin{align*}
   \bar{s}^2_{21}=& \frac{1}{4h^d} \int_{2'}h^d  \frac{\sqrt{Var(\hat{\tau}(x,\pi_k)) Var(\hat{\tau}(x+th,\pi_k))}}{\sqrt{Var(\hat \Gamma(x,x',\pi_k)) Var(\hat \Gamma(x+th, x''',\pi_k)}}\\&\cdot\mathrm{Cov}(|Z_{1n}(x,\pi_k)|,|Z_{2n}(x+th,\pi_k)|)dxdx'dx'''dt, 
\end{align*}
where 
$\int_{2'}\coloneqq\int_{[-1,1]^{d}}\int_{\mathbbm{R}^{d}} \int_{\mathbbm{R}^{d}}\int_{\mathbbm{R}^{d}}\sum_{k=1}^K.$ 

Now note that from \eqref{bivariate normal distribution generation 3} and \eqref{correlation}, $$\mathrm{Cov}(|Z_{1n}(x,\pi_k)|,|Z_{2n}(x+th,\pi_k)|)\to \mathrm{Cov}\left(\left|\sqrt{1-\rho(x, t, \pi_k)^2}\mathbbm{Z}_1 +\rho(x, t, \pi_k)\mathbbm{Z}_2\right|,\left|\mathbbm{Z}_2\right|\right)$$
Finally, as in the proof of (6.35) in \cite{gine2003bm}, using the bounded convergence theorem with respect to the Lebesgue measure, the first term of the variance converges to 
\begin{align*}
   \bar{s}^2_{21}=& \frac{1}{4} \int_{2'}  \frac{\sqrt{Var(\hat{\tau}(x,\pi_k)) Var(\hat{\tau}(x+th,\pi_k))}}{\sqrt{Var(\hat \Gamma(x,x',\pi_k)) Var(\hat \Gamma(x+th,x''',\pi_k)}}\mathrm{Cov}\left(\left|\sqrt{1-\rho(x, t, \pi_k)^2}\mathbbm{Z}_1 +\rho(x, t, \pi_k)\mathbbm{Z}_2\right|,\left|\mathbbm{Z}_2\right|\right)dxdx'dx'''dt \\
   =& \frac{1}{4} \int_{2'}  \frac{Var(\hat{\tau}(x,\pi_k)) }{\sqrt{Var(\hat \Gamma(x,x',\pi_k)) Var(\hat \Gamma(x,x''',\pi_k)}}\mathrm{Cov}\left(\left|\sqrt{1-\rho(x, t, \pi_k)^2}\mathbbm{Z}_1 +\rho(x, t, \pi_k)\mathbbm{Z}_2\right|,\left|\mathbbm{Z}_2\right|\right)dxdx'dx'''dt,
\end{align*}
Analogously, the second term  reduces to
\begin{align*}
   \bar{s}^2_{22}=& \frac{1}{4} \int_{2'}  \frac{\sqrt{Var(\hat{\tau}(x,\pi_k)) Var(\hat{\tau}(x+th,\pi_k))}}{\sqrt{Var(\hat \Gamma(x,x',\pi_k)) Var(\hat \Gamma(x'',x+th,\pi_k)}}\mathrm{Cov}\left(\left|\sqrt{1-\rho(x, t, \pi_k)^2}\mathbbm{Z}_1 +\rho(x, t, \pi_k)\mathbbm{Z}_2\right|,\left|\mathbbm{Z}_2\right|\right)dxdx'dx''dt, \\
   =& \frac{1}{4} \int_{2'}  \frac{Var(\hat{\tau}(x,\pi_k)) )}{\sqrt{Var(\hat \Gamma(x,x',\pi_k)) Var(\hat \Gamma(x'',x,\pi_k)}}\mathrm{Cov}\left(\left|\sqrt{1-\rho(x, t, \pi_k)^2}\mathbbm{Z}_1 +\rho(x, t, \pi_k)\mathbbm{Z}_2\right|,\left|\mathbbm{Z}_2\right|\right)dxdx'dx''dt, 
\end{align*}
the third term becomes
\begin{align*}
   \bar{s}^2_{23}=& \frac{1}{4} \int_{2'}  \frac{\sqrt{Var(\hat{\tau}(x',\pi_k)) Var(\hat{\tau}(x'+th,\pi_k))}}{\sqrt{Var(\hat \Gamma(x,x',\pi_k)) Var(\hat \Gamma(x'+th,x''',\pi_k)}}\mathrm{Cov}\left(\left|\sqrt{1-\rho(x', t, \pi_k)^2}\mathbbm{Z}_1 +\rho(x', t, \pi_k)\mathbbm{Z}_2\right|,\left|\mathbbm{Z}_2\right|\right)dxdx'dx'''dt\\
   =& \frac{1}{4} \int_{2'}  \frac{Var(\hat{\tau}(x',\pi_k))}{\sqrt{Var(\hat \Gamma(x,x',\pi_k)) Var(\hat \Gamma(x',x''',\pi_k)}}\mathrm{Cov}\left(\left|\sqrt{1-\rho(x', t, \pi_k)^2}\mathbbm{Z}_1 +\rho(x', t, \pi_k)\mathbbm{Z}_2\right|,\left|\mathbbm{Z}_2\right|\right)dxdx'dx'''dt, 
\end{align*}
and the fourth term reduces to
\begin{align*}
   \bar{s}^2_{24}=& \frac{1}{4} \int_{2'}  \frac{\sqrt{Var(\hat{\tau}(x',\pi_k)) Var(\hat{\tau}(x'+th,\pi_k))}}{\sqrt{Var(\hat \Gamma(x,x',\pi_k)) Var(\hat \Gamma(x'',x'+th,\pi_k)}}\mathrm{Cov}\left(\left|\sqrt{1-\rho(x', t, \pi_k)^2}\mathbbm{Z}_1 +\rho(x', t, \pi_k)\mathbbm{Z}_2\right|,\left|\mathbbm{Z}_2\right|\right)dxdx'dx''dt\\ 
   =& \frac{1}{4} \int_{2'}  \frac{Var(\hat{\tau}(x',\pi_k)) }{\sqrt{Var(\hat \Gamma(x,x',\pi_k)) Var(\hat \Gamma(x'',x',\pi_k)}}\mathrm{Cov}\left(\left|\sqrt{1-\rho(x', t, \pi_k)^2}\mathbbm{Z}_1 +\rho(x', t, \pi_k)\mathbbm{Z}_2\right|,\left|\mathbbm{Z}_2\right|\right)dxdx'dx''dt,
\end{align*}
Therefore, 
\begin{align*}
    \sigma_2^2=&\bar{s}^2_{21}+\bar{s}^2_{22}+\bar{s}^2_{23}+\bar{s}^2_{24}\\
    =&  \int_{2'}  \frac{Var(\hat{\tau}(x,\pi_k)) )}{\sqrt{Var(\hat \Gamma(x,x',\pi_k)) Var(\hat \Gamma(x'',x,\pi_k)}}\mathrm{Cov}\left(\left|\sqrt{1-\rho(x, t, \pi_k)^2}\mathbbm{Z}_1 +\rho(x, t, \pi_k)\mathbbm{Z}_2\right|,\left|\mathbbm{Z}_2\right|\right)dxdx'dx''dt, 
\end{align*}
$$.$$
Plugging the kernel-based estimates of the variances into the formula gives us the plug-in estimator $\hat \sigma_2^2.$ This variance formula can also be derived using the law of total variances.

\subsection{Proofs of Lemmas and Theorems}\label{Proofs of Lemmas and Theorems}
In this section, I discuss the proof of the asymptotic normality of $\hat{S}_1$ and $\hat{S}_2$ given in Proposition \ref{normality}.
\begin{lemma}[Asymptotics of $L_1-$norm of the  CATE function, \cite{chang2015nonparametric}]\label{normality in Chang}
Assume $\{(Y_i,D_i,X_i), i=1,\dots N\}$ are i.i.d copies of $(Y,D,X)$ Let  $\hat{\tau}(x)$ be the standard local constant CATE estimator, and let $\hat{w}(x)$ be the inverse scaled standard error estimators. Set $\hat{D}:=\int_{\mathcal{X}}\sqrt{N}|\hat{\tau}(x)|\hat{w}(x)dx.$ Then, under the regularity conditions in Assumption \ref{regularity conditions},
$$\frac{\hat{D}-\mathrm{Bias}(\hat{D})}{\sigma}\overset{d}{\to} N(0,1),$$ where $\mathrm{Bias}(\hat{D})$ is the asymptotic bias of $\hat{D}$ and 
$\sigma^2:=\lim_{N\to \infty} Var(\hat{D}),$ under the restriction that the true CATE function ${\tau}(x)=0 \,\,\forall x.$
\end{lemma}
\begin{proof}
    This lengthy proof  involves 
    \begin{enumerate}
        \item a uniform approximation $\hat{\tau}(x)$ to derive the asymptotic approximation of $\hat{D}$ by ${D}_N,$ 
        \item Poissonization of ${D}_N$ and limiting the integration interval to obtain  ${D}^P_N(B),$ where $B \subseteq \mathcal{X},$ and 
        \item de-Poissonizing ${D}^P_N(B)$ to derive its asymptotic normality and that of the studentized version of $\hat{D}$ by allowing  $B \to \mathcal{X},$
    \end{enumerate}
   I omit the detail and refer readers to Appendix B of \cite{chang2015nonparametric}.
\end{proof}

\subsubsection{Proof of the asymptotic normality of $\hat{S}_1$}
Recall that all within-cluster dependencies can be ignored under the working independence approach, and observations can be viewed as independent. Therefore, I can apply the results in Lemma \ref{normality in Chang}.
\subsubsection*{Uniform asymptotic approximation  }
Write $$\hat{\tau}(x,\pi)=\tau(x;\pi)+ (\tau_{N0}(x,\pi)-\mathbbm{E}(\tau_{N0}(x,\pi)))+(\mathbbm{E}(\tau_{N0}(x,\pi))-\tau(x;\pi))+ R_N(x,\pi),$$
where
    $$\tau_{N0}(x,\pi):=\frac{1}{Nh^d}\sum_{i=1}^{N}Y_i\mathbbm{1}(\Pi_i=\pi)\left[\frac{T_i}{P_1(x;\pi)}-\frac{(1-T_i)}{P_0(x;\pi)}\right]\cdot K\left(\frac{x-X_i}{h}\right),$$
    \begin{align*}
        R_N(x,\pi):=&\frac{1}{Nh^d}\sum_{i=1}^{N} Y_i\mathbbm{1}(\Pi_i=\pi)\left[\frac{T_i}{P_1(x;\pi)}-\frac{(1-T_i)}{P_0(x;\pi)}\right]\\
    &\times\left(T_i\frac{P_1(x,\pi)-\hat{P}_1(x,\pi)}{\hat{P}_1(x,\pi)}+(1-T_i)\frac{P_0(x,\pi)-\hat{P}_0(x,\pi)}{\hat{P}_0(x,\pi)}\right)\cdot K\left(\frac{x-X_i}{h}\right).
    \end{align*}
Therefore,
\begin{align*}
    \hat{\Gamma}(x;\pi_k,\pi_j)=&\tau(x;\pi_k)-\tau(x;\pi_j)+ (\tau_{N0}(x,\pi_k)-\tau_{N0}(x,\pi_j)-\mathbbm{E}(\tau_{N0}(x,\pi_k)-\tau_{N0}(x,\pi_j)))\\
    &+(\mathbbm{E}(\tau_{N0}(x,\pi_k)-\tau_{N0}(x,\pi_j))-\tau(x;\pi_k)-\tau(x;\pi_j)))+ R_N(x,\pi_k)-R_N(x,\pi_j).
\end{align*}

Now, define
    \begin{align*}
        \zeta_{_N}(x,\pi)=&\mathbbm{E}[Y|X=x,\Pi=\pi, T=1]-\mathbbm{E}[Y|X=x,\Pi=\pi, T=0]\\
        &-\mathbbm{E}[Y|X=x,\Pi=\pi, T=1]\frac{1}{Nh^dP_1(x,\pi)}\sum_{i=1}^{N}T_i\mathbbm{1}(\Pi=\pi)K\left(\frac{x-X_i}{h}\right)\\
        &+\mathbbm{E}[Y|X=x,\Pi=\pi, T=0]\frac{1}{Nh^dP_0(x,\pi)}\sum_{i=1}^{N}(1-T_i)\mathbbm{1}(\Pi=\pi)K\left(\frac{x-X_i}{h}\right).\\
    \end{align*}
The following lemma shows that $R_N(x,\pi)$ can be approximated by $\zeta_{_N}(x,\pi)$ uniformly over $x$ at a rate faster than $C^{-1/2}$.
 \begin{lemma}\label{three}
    Under the regularity conditions, we find that for $\pi_k, \pi_j \in \mathbf{\Pi},$
    $$\sup_{x \in \mathcal{X}}|(R_N(x,\pi_k)-R_N(x,\pi_j)) -(\zeta_{_N}(x,\pi_k)-\zeta_{_N}(x,\pi_j))|= o_p(N^{-1/2}).$$
    \end{lemma}
   \begin{proof}
  From, Lemma B.1 in \cite{chang2015nonparametric}, 
  for $\pi_k \in \Pi,$ we have 
  $$\sup_{x \in \mathcal{X}}|R_{N}(x,\pi_k)-\zeta_{_{N}}(x,\pi_k))|=o_p(N^{-1/2}).$$
  Now,
\footnotesize
  \begin{align*}
    \sup_{x \in \mathcal{X}}|(R_N(x,\pi_k)-R_N(x,\pi_j)) -(\zeta_{_N}(x,\pi_k)-\zeta_{_N}(x,\pi_j))|&\leq  \sup_{x \in \mathcal{X}}\{|R_N(x,\pi_k)-\zeta_{_N}(x,\pi_k)|+|R_N(x,\pi_j)-\zeta_{_N}(x,\pi_j)|\}\\
    &=\sup_{x \in \mathcal{X}}\{|R_N(x,\pi_k)-\zeta_{_N}(x,\pi_k)|\}+\sup_{x \in \mathcal{X}}\{|R_N(x,\pi_j)-\zeta_{_N}(x,\pi_j)|\}\\
   &= o_p(N^{-1/2})+o_p(N^{-1/2})= o_p(N^{-1/2}).
  \end{align*}
   \end{proof} 
 \normalsize   
\begin{lemma}\label{TandTstar}
    Under the regularity conditions, we have 
    $$\Tilde{T}_1(\pi_k,\pi_j)-T_{1N}^{*}(\pi_k,\pi_j)= o_p(1),$$
    where 
     $$
     \Tilde{T}_1(\pi_k,\pi_j) \coloneqq  \int_{\mathcal{X}}  \left\{\sqrt{N}|\hat{\tau}(x;\pi_k)-\hat{\tau}(x;\pi_j)|\right\} \hat{w}(x, \pi_k, \pi_j)dx,
$$
    \footnotesize{
    \begin{align} \label{TNstar}
         T_{1N}^{*}(\pi_k,\pi_j):=   \int_{\mathcal{X}}\left\{\sqrt{N}\left|\Gamma(x, \pi_k,\pi_j)+ [\tau_N(x,\pi_k)-\tau_N(x,\pi_j)]-\mathbbm{E}[\tau_N(x,\pi_k)-\tau_N(x,\pi_j)]\right|\right\}w(x,\pi_j,\pi_k) dx
    \end{align}
    }
     \normalsize
    and 
    $$\tau_N(x,\pi)=\tau_{N0}(x,\pi)+ \zeta_N(x, \pi).$$
    Hence, under the null hypothesis such that $\tau(x;\pi_k)=\tau(x;\pi_j)$ on $\mathcal{X}\times \Pi$, we have $$\Tilde{T}_1(\pi_k,\pi_j)-T_{1N}(\pi_k,\pi_j) =o_p(1),$$ where 
    $$T_{1N}(\pi_k,\pi_j):=   \int_{\mathcal{X}}\left\{\sqrt{N}\left|[\tau_N(x,\pi_k)-\tau_N(x,\pi_j)]-\mathbbm{E}[\tau_N(x,\pi_k)-\tau_N(x,\pi_j)]\right|\right\}w(x,\pi_j,\pi_k)\footnote{ Here, I use the true weight rather than the estimated weight but all results hold using the estimated weight since the estimator is uniformly consistent by assumption.} dx.$$
   and 
\begin{align*}
   \tau_N(x,\pi):=&\frac{1}{Nh^d}\sum_{i=1}^{N}( \{Y-\mathbbm{E}[Y|X=x,\Pi=\pi, T=1]\}\frac{T_i\cdot\mathbbm{1}(\Pi=\pi)}{P_1(x,\pi_k)}\\
   &-\{Y-\mathbbm{E}[Y|X=x,\Pi=\pi,T=0]\}\frac{(1-T_i)\cdot\mathbbm{1}(\Pi=\pi)}{P_0(x,\pi_k)} )\cdot K\left(\frac{x-X_i}{h}\right).
\end{align*}
    \end{lemma}

\begin{proof}
Using the triangle inequality and the proof of Lemma B.2 in \cite{chang2015nonparametric}, the proof of this Lemma is straight forward.
\end{proof}

\subsubsection*{Consistency of the estimators of asymptotic  variance}
\begin{lemma} \label{consis}
 Under the regularity conditions, and using the working under independence approach the following hold:
 \begin{enumerate}
     \item $\sup_{x \in \mathcal{X}}|\hat{\tau}(x, \pi)- \tau(x, \pi) |= O_{p}((Nh^d)^{-1/2}\log N +h^s)$ $\forall \pi \in \Pi,$
     \item $\sup_{x \in \mathcal{X}}|\hat{\rho}_2(x, \pi)- \rho_2(x, \pi) |= O_{p}((Nh^d)^{-1/2}\log N +h^s)$ $\forall \pi \in \Pi.$
 \end{enumerate}
\end{lemma}

\begin{proof}
This Lemma corresponds to Lemma B.2 in \cite{chang2015nonparametric}. Hence I omit the proof.
\end{proof}

\begin{theorem} \label{main theorem}
Under the regularity conditions, we have  that
$$\frac{T_{1N}(\pi_k,\pi_j)-\mathrm{Bias}(T_{1N}(\pi_k,\pi_j))}{{\sigma}_{1N}(\pi_k,\pi_j)} \overset{d}{\to} N(0,1),$$ where \,\,$\mathrm{Bias}(T_{1N}(\pi_k,\pi_j))$ is the asymptotic bias of \,\,$T_{1N}(\pi_k,\pi_j),$ and
${\sigma}_{1N}^2(\pi_k,\pi_j):=\lim_{N\to \infty} Var(T_{1N}(\pi_k,\pi_j))$ under the null.
\end{theorem}
\begin{proof}
    The proof of Theorem \ref{main theorem} is very lengthy so I defer it to Appendix \ref{appen c}.
\end{proof}

Now, from Lemma \ref{TandTstar}  and Theorem \ref{main theorem}, we have   \begin{align}
    \frac{\Tilde{T}_{1}(\pi_k,\pi_j)-\mathrm{Bias}(\Tilde{T}_{1}(\pi_k,\pi_j))}{\tilde{\sigma}_1(\pi_k,\pi_j)} \overset{d}{\to} N(0,1),
\end{align}where $\mathrm{Bias}(\Tilde{T}_{1}(\pi_k,\pi_j))=\mathrm{Bias}(T_{1N}(\pi_k,\pi_j))$ is the asymptotic bias of $\Tilde{T}_{1}(\pi_k,\pi_j)$ and 
$\tilde{\sigma}^2_1(\pi_k,\pi_j):=\lim_{N\to \infty} Var(\Tilde {T}_{1}(\pi_k,\pi_j))=\lim_{N\to \infty} Var(T_{1N}(\pi_k,\pi_j))$ under the null.

Next, observe that our  test statistic $\hat{T}_1$ can be written as 
$$\hat{T}_1=\sum_{k=1}^K\sum_{j>k}^K \Tilde{T}_{1}(\pi_k,\pi_j),$$ which asymptotically is the sum of dependent normal random variables with $\{\mathrm{Bias}(\Tilde{T}_{1}(\pi_k,\pi_j)),\pi_k,\pi_j\in \boldsymbol{\Pi}\}$ as the asymptotic biases and, $\{\tilde{\sigma}^2_1(\pi_k,\pi_j),\pi_k,\pi_j\in \boldsymbol{\Pi}\}$ as the asymptotic variances under the null hypothesis. Thus,  $$\frac{\hat{T}_{1}-\sum_{k=1}^K\sum_{j>k}^K\mathrm{Bias}(\Tilde{T}_{1}(\pi_k,\pi_j))}{\hat{\sigma}_1} \overset{d}{\to} N(0,1)$$
where  $\sum_{k=1}^K\sum_{j>k}^K\mathrm{Bias}(\Tilde{T}_{1}(\pi_k,\pi_j))=:a_1,$ $\hat{\sigma}^2_1:=\lim_{N\to \infty} \widehat{Var}(\hat{T}_{1}),$ with $a_1$ and  $\hat{\sigma}^2_1$ defined in Proposition \ref{normality}. So, note that  $\hat{\sigma}^2_1$ can also be viewed as the sum of the asymptotic variances of $\Tilde{T}_{1}(\pi_k,\pi_j)$ and their pairwise asymptotic covariances.

\subsubsection{Proof of Asymptotic normality of $\hat{S}_2$}
\subsubsection*{Uniform Asymptotic approximation } 
Recall that we can write 
 $$\hat{\tau}(x,\pi)=\tau(x;\pi)+ (\tau_{N0}(x,\pi)-\mathbbm{E}(\tau_{N0}(x,\pi)))+(\mathbbm{E}(\tau_{N0}(x,\pi))-\tau(x;\pi))+ R_N(x,\pi)$$
    Therefore,
\begin{align*}
    \hat{\Gamma}(x,x',\pi_k)=&\tau(x;\pi_k)-\tau(x',\pi_k)+ (\tau_{N0}(x,\pi_k)-\tau_{N0}(x',\pi_k)-\mathbbm{E}(\tau_{N0}(x,\pi_k)-\tau_{N0}(x',\pi_k)))\\
    &+(\mathbbm{E}(\tau_{N0}(x,\pi_k)-\tau_{N0}(x',\pi_k))-\tau(x;\pi_k)-\tau(x',\pi_k)))+ R_N(x,\pi_k)-R_N(x',\pi_k)
\end{align*}

Now, define
    \begin{align*}
        \zeta_{_N}(x,\pi)=&\mathbbm{E}[Y|X=x,\Pi=\pi, T=1]-\mathbbm{E}[Y|X=x,\Pi=\pi, T=0]\\
        &-\mathbbm{E}[Y|X=x,\Pi=\pi, T=1]\frac{1}{Nh^dP_1(x,\pi)}\sum_{i=1}^{N}T_i\mathbbm{1}(\Pi=\pi)K\left(\frac{x-X_i}{h}\right)\\
        &+\mathbbm{E}[Y|X=x,\Pi=\pi, T=0]\frac{1}{Nh^dP_0(x,\pi)}\sum_{i=1}^{N}(1-T_i)\mathbbm{1}(\Pi=\pi)K\left(\frac{x-X_i}{h}\right).
    \end{align*}
The following lemma shows that $R_N(x,\pi)$ can be approximated by $\zeta_{_N}(x,\pi)$ uniformly over $x$ at a rate faster than $N^{-1/2}$.

Under the regularity conditions, we find that for $\pi_k\in \mathbf{\Pi},$ $$\sup_{x \in \mathcal{X}}|(R_N(x,\pi_k)-R_N(x',\pi_k)) -(\zeta_{_N}(x,\pi_k)-\zeta_{_N}(x',\pi_k))|= o_p(N^{-1/2}).$$
  
   \begin{proof}
  Similar to the proof of Lemma \ref{three},
  for $\pi_k \in \Pi,$ we have 
  $$\sup_{x \in \mathcal{X}}|R_{N}(x,\pi_k)-\zeta_{_{N}}(x,\pi_k))|=o_p(N^{-1/2}).$$
  Now,
  \begin{align*}
    \sup_{(x, x') \in \mathcal{X}^2}&|(R_N(x,\pi_k)-R_N(x',\pi_k)) -(\zeta_{_N}(x,\pi_k)-\zeta_{_N}(x',\pi_k))|\\
    \leq & \sup_{(x, x') \in \mathcal{X}^2}\{|R_N(x,\pi_k)-\zeta_{_N}(x,\pi_k)|+|R_N(x',\pi_k)-\zeta_{_N}(x',\pi_k)|\}\\
     \leq & \sup_{x \in \mathcal{X}}\{|R_N(x,\pi_k)-\zeta_{_N}(x,\pi_k)|\}+\sup_{x' \in \mathcal{X}}\{|R_N(x',\pi_k)-\zeta_{_N}(x',\pi_k)|\}\\
   = & o_p(N^{-1/2})+o_p(N^{-1/2})= o_p(N^{-1/2}).
  \end{align*}
\end{proof} 

 Under the regularity conditions, we have 
    $$\tilde{T}_2(\pi_k)-T_{2N}^{*}(\pi_k)= o_p(1),$$
where 
 \begin{align*}
     \hat{T}_2(\pi_k) \coloneqq &  \int_{\mathcal{X}}\int_{\mathcal{X}}  \left\{\sqrt{N}|\hat{\tau}(x;\pi_k)-\hat{\tau}(x';\pi_k)|\right\} \frac{\hat{w}(x, x', \pi_k)}{2}dxdx' 
 \end{align*}
\begin{align*}
    T_{2N}^{*}(\pi_k):= &\int_{\mathcal{X}} \int_{\mathcal{X}} \Bigg\{\sqrt{N}\Bigg|\Gamma(x, x',\pi_k)+ [\tau_N(x,\pi_k)-\tau_N(x',\pi_k)]
    -\mathbbm{E}[\tau_N(x,\pi_k)-\tau_N(x',\pi_k)]\Bigg|\Bigg\}\frac{w(x,x',\pi_k)}{2} dxdx'
\end{align*}
and 
$$\tau_N(x,\pi)=\tau_{N0}(x,\pi)+ \zeta_N(x, \pi).$$
Hence, under the null hypothesis such that $\tau(x;\pi_k)=\tau(x',\pi_k)$ on $\mathcal{X}\times \Pi$, we have $$\tilde{T}_2(\pi_k)=T_{2N}(\pi_k) +o_p(1),$$ where 
    $$T_{2N}(\pi_k):=   \int_{\mathcal{X}}  \int_{\mathcal{X}}\left\{\sqrt{N}\left|[\tau_N(x,\pi_k)-\tau_N(x',\pi_k)]-\mathbbm{E}[\tau_N(x,\pi_k)-\tau_N(x',\pi_k)]\right|\right\}\frac{w(x,x',\pi_k)}{2} dxdx'$$
and 
\begin{align*}
   \tau_N(x,\pi):=&\frac{1}{Nh^d}\sum_{i=1}^{N}\Bigg( \{Y-\mathbbm{E}[Y|X=x,\Pi=\pi, T=1]\}\frac{T_i\cdot\mathbbm{1}(\Pi=\pi)}{P_1(x,\pi_k)}\\
   &-\{Y-\mathbbm{E}[Y|X=x,\Pi=\pi,T=0]\}\frac{(1-T_i)\cdot\mathbbm{1}(\Pi=\pi)}{P_0(x,\pi_k)} \Bigg)\cdot K\left(\frac{x-X_i}{h}\right).
\end{align*}

\begin{theorem} \label{main theorem 2}
Under the regularity conditions, we have  that
$$\frac{T_{2N}(\pi_k)-\mathrm{Bias}(T_{2N}(\pi_k))}{\sigma_{2N}(\pi_k)} \overset{d}{\to} N(0,1),$$ where \,\,$\mathrm{Bias}(T_{2N}(\pi_k))$ is the asymptotic bias of \,\,$T_{2N}(\pi_k),$ and
${\sigma}_{2N}^2(\pi_k):=\lim_{N\to \infty} Var(T_{2N}(\pi_k))$ under the null.
\end{theorem}
\begin{proof}
    The proof of Theorem \ref{main theorem 2} is also very lengthy but similar to Theorem \ref{main theorem}. I omit the proof  to save space.
\end{proof}

Now based on the fact that $\hat{T}_2=T_{2N} +o_p(1),$  and Theorem \ref{main theorem 2},
\begin{align}
    \frac{\Tilde{T}_{2}(\pi_k)-\mathrm{Bias}(\Tilde{T}_{2}(\pi_k))}{\tilde{\sigma}_2(\pi_k)} \overset{d}{\to} N(0,1),
\end{align} where $\mathrm{Bias}(\Tilde{T}_{2}(\pi_k))=\mathrm{Bias}({T}_{2N}(\pi_k))$ is the asymptotic bias of $\Tilde{T}_{2}(\pi_k)$  and
$\tilde{\sigma}^2_2(\pi_k):=\lim_{N\to \infty} Var(\Tilde {T}_{2}(\pi_k))=\lim_{N\to \infty} Var( {T}_{2N}(\pi_k)) $ under the null.

Finally, observe that our  test statistic $\hat{T}_2$ can be written as 
$$\hat{T}_2=\sum_{k=1}^K \Tilde{T}_{2}(\pi_k),$$ which asymptotically is the sum of independent normal variables with $\{\mathrm{Bias}(\Tilde{T}_{2}(\pi_k)),\pi_k,\in \boldsymbol{\Pi}\}$ as the asymptotic biases and, $\{\tilde{\sigma}^2_1(\pi_k,),\pi_k,\in \boldsymbol{\Pi}\}$ as the asymptotic variances under the null hypothesis. As a result,  
$$\frac{\hat{T}_{2}-\sum_{k=1}^K\mathrm{Bias}(\Tilde{T}_{2}(\pi_k))}{\hat{\sigma}_2} \overset{d}{\to} N(0,1)$$
where  $\sum_{k=1}^K\mathrm{Bias}(\hat{T}_{2})=:a_2,$ $\hat{\sigma}^2_2(\pi_k):=\lim_{N\to \infty} \widehat{Var}(\hat{T}_{2}),$ with $a_2$ and  $\hat{\sigma}^2_2$ defined in Proposition \ref{normality} (b). So, one can view  $\hat{\sigma}^2_1$ the sum of the asymptotic variances of $\Tilde{T}_{2}(\pi_k).$

\subsection{Proof of Theorem \ref{valid size}}

\begin{proof}
I prove  the first part  of Theorem \ref{valid size}. 
To save space, I omit the proof of the second part of Theorem \ref{valid size} because  it is similar to the first part.
\begin{align*}
    \begin{split}
        \Pr(\hat{S}_1>z_{1-\alpha})=&\Pr(\hat{T}_{1}>a_{1} +\hat{\sigma}_1z_{1-\alpha})\\
        =&1 - \Pr(\hat{T}_{1}\leq a_{1} +\hat{\sigma}_1z_{1-\alpha})\\
        \to& 1-(1-\alpha)\\
        =& \alpha.
    \end{split}
\end{align*}
 The convergence to $1- (1-\alpha)$ follows from Theorem  \ref{main theorem}. 
\end{proof}

\subsection{Proof of Theorem \ref{power theorem}}
\begin{proof}
 I prove the first part  of Theorem \ref{power theorem}. We can use similar arguments to prove the second part.
\begin{align*}
 \Pr(\hat{S}_1>z_{1-\alpha})  =& \Pr(\hat{T}_{1}>a_{1} +\hat{\sigma}z_{1-\alpha})\\
 =& \Pr\left(\frac{\hat{T}_{1}}{\sqrt{N}}>\frac{a_{_{1}} +\hat{\sigma}z_{1-\alpha}}{\sqrt{N}}\right)\\
 =& \Pr\left(\frac{\hat{T}_{1}}{\sqrt{N}}>0\right)-\Pr\left(0<\frac{\hat{T}_{1}}{\sqrt{N}}< \frac{a_{_{1N}} +\hat{\sigma}z_{1-\alpha}}{\sqrt{N}}\right)\\
 =& \Pr\left(\frac{\hat{T}_{1}}{\sqrt{N}}>0\right)-o(1)\\
 \to& 1,
\end{align*}
where the third equality holds because $(a_{_{1}} +\hat{\sigma}z_{1-\alpha})/\sqrt{N}= o(1)$ since $\sqrt{Nh^d}\to \infty$ as $C\to \infty.$ The last convergence to one follows from the definition of the alternative hypothesis, i.e.,  $H^{\Pi}_{1}: \int_{\mathcal{X}} \sum_{k =1}^K \sum_{j>k}^K \left\{|\tau(x;\pi_k)-\tau(x;\pi_j)|\right\} w(x, \pi_k, \pi_j)dx>0.$ 
\end{proof}

\subsection{Proof of Theorem \ref{local power theorem}}
\begin{proof} I prove the first part (i) of Theorem \ref{local power theorem}. We can use similar arguments to prove the second part (ii).
Under $ H^{\Pi}_{a}: \tau(x;\pi)- \tau(x;\pi')= C^{-1/2}h^{-d/4} N_0^{-1/2}\delta_1(x, \pi, \pi')$,  with a slight modification of the arguments in the proof of Theorem \ref{main theorem}, I can also show that  under $H_{1a},$  $T_{1N}^*(\pi_k,\pi_j) $ in \eqref{TNstar} is asymptotically normal, i.e., 
\begin{align}\label{normality under local alternatives}
    \frac{T_{1N}^*(\pi_k,\pi_j)-\mathbbm{E}\left[T_{1N}^{*}(\pi_k,\pi_j)\right]}{{\sigma}_{1N}(\pi_k,\pi_j)}\overset{d}{\to}N(0,1),
\end{align}
where 
$T_{1N}^{*}(\pi_k,\pi_j):=$  
$$\int_{\mathcal{X}}\left\{\left|h^{-d/4}\cdot\delta_1(x, \pi_k, \pi_j)+ \sqrt{N}\Big([\tau_N(x,\pi_k)-\tau_N(x,\pi_j)]-\mathbbm{E}[\tau_N(x,\pi_k)-\tau_N(x,\pi_j)]\Big)\right|\right\}w(x,\pi_j,\pi_k) dx,$$ 
$$\mathbbm{E}\left[T_{1N}^{*}(\pi_k,\pi_j)\right]= \int_{\mathcal{X}} \mathbbm{E}[|\mathbbm{Z}_1\cdot h^{-d/2}\sqrt{\rho_2(x, \pi_k, \pi_j)}+ h^{-d/4}\delta(x, \pi_k, \pi_j)|]\cdot w(x, \pi_k, \pi_j)dx,$$
 and,  ${\sigma}^2_{1N}(\pi_k,\pi_j):=Var(T_{1N}^{*}(\pi_k,\pi_j))$ which is not affected by the restrictions of  the alternative hypothesis. Thus,
 from Lemma \ref{TandTstar}  and \eqref{normality under local alternatives},  under $H^{\Pi}_{a},$ we have   
 \begin{align}
    \frac{\Tilde{T}_{1}(\pi_k,\pi_j)-\mathbbm{E}\left[\Tilde T_{1}(\pi_k,\pi_j)\right]}{\tilde{\sigma}_1(\pi_k,\pi_j)} \overset{d}{\to} N(0,1),
\end{align}where $\mathbbm{E}\left[\Tilde T_{1}(\pi_k,\pi_j)\right]=\mathbbm{E}\left[T_{1N}^{*}(\pi_k,\pi_j)\right]$ is the asymptotic bias of $\Tilde{T}_{1}(\pi_k,\pi_j)$ and 
$\tilde{\sigma}^2_1(\pi_k,\pi_j):=\lim_{N\to \infty} Var(\Tilde {T}_{1}(\pi_k,\pi_j))=\lim_{N\to \infty} Var(T_{1N}^{*}(\pi_k,\pi_j))$ under $H^{\Pi}_{a}.$

As a result, under $H^{\Pi}_{a},$  
$$\frac{\hat{T}_{1}-\sum_{k=1}^K\sum_{j>k}^K\mathbbm{E}\left[T_{1N}^{*}(\pi_k,\pi_j)\right]}{\hat{\sigma}_1} \overset{d}{\to} N(0,1)$$
where, $\hat{\sigma}^2_1:=\lim_{N\to \infty} \widehat{Var}(\hat{T}_{1}),$ as  defined in Proposition \ref{normality}.

Moreover, analogous to arguments in the proof of Theorem 4.3 in \cite{chang2015nonparametric}, we have 
$$\lim_{N\to \infty} \{\mathbbm{E}\left[T_{1N}^{*}(\pi_k,\pi_j)\right]- \mathrm{Bias}(T_{1N}(\pi_k,\pi_j))\}= \frac{1}{2\sqrt{\pi}}\int\delta^2(x, \pi_k, \pi_j)dx$$ which gives us 
$$\lim_{N\to \infty} \left\{\sum_{k=1}^K\sum_{j>k}^K\mathbbm{E}\left[T_{1N}^{*}(\pi_k,\pi_j)\right]- \sum_{k=1}^K\sum_{j>k}^K\mathrm{Bias}(T_{1N}(\pi_k,\pi_j))\right \}= \frac{1}{2\sqrt{\pi}}\int\sum_{k=1}^K\sum_{j>k}^K\delta^2(x, \pi_k, \pi_j)dx$$

Thus, under $H^{\Pi}_{a},$ 
\begin{align*}
    &\Pr(\hat{S}_1>z_{1-\alpha})\\
    &=\Pr(\hat{T}_{_{1}}> {a}_{_{1}} + \hat{\sigma}_1z_{1-\alpha})\\
    &= \Pr\Bigg(\frac{\hat{T}_{_{1}}- \sum_{k=1}^K\sum_{j>k}^K\mathbbm{E}\left[T_{1N}^{*}(\pi_k,\pi_j)\right]}{\hat\sigma_1}> \frac{-\sum_{k=1}^K\sum_{j>k}^K\mathbbm{E}\left[T_{1N}^{*}(\pi_k,\pi_j)\right]+a_{_{1}}}{\hat\sigma_1} + z_{1-\alpha} \Bigg)\\
    &= \Pr\Bigg(\frac{\hat{T}_{_{1}}- \sum_{k=1}^K\sum_{j>k}^K\mathbbm{E}\left[T_{1N}^{*}(\pi_k,\pi_j)\right]}{\hat\sigma_1}> \frac{-\sum_{k=1}^K\sum_{j>k}^K\mathbbm{E}\left[T_{1N}^{*}(\pi_k,\pi_j)\right]+\sum_{k=1}^K\sum_{j>k}^K\mathrm{Bias}(T_{1N}(\pi_k,\pi_j))}{\hat\sigma_1} + z_{1-\alpha} \Bigg)\\
    &\to \Pr\Bigg(\frac{\hat{T}_{_{1}}- \sum_{k=1}^K\sum_{j>k}^K\mathbbm{E}\left[T_{1N}^{*}(\pi_k,\pi_j)\right]}{\sigma_1}> \frac{-\int\sum_{k=1}^K\sum_{j>k}^K\delta^2(x, \pi_k, \pi_j)}{2\sqrt{\pi}\sigma_1} + z_{1-\alpha} \Bigg)\\
    &= 1- \Phi\Bigg(z_{1-\alpha}- \frac{1}{\sqrt{2\pi}\sigma_1}\int\sum_{k=1}^K\sum_{j=1}^K\delta^2(x, \pi_k, \pi_j)dx\Bigg)
\end{align*}
\end{proof}

\section{Proof of Theorem \ref{main theorem}} \label{appen c}
Recall that Theorem \ref{main theorem} says that
under the regularity conditions, we have  that
$$\frac{T_{1N}(\pi_k,\pi_j)-\mathrm{Bias}(T_{1N}(\pi_k,\pi_j))}{{\sigma}_{1N}(\pi_k,\pi_j)} \overset{d}{\to} N(0,1),$$ where
$T_{1N}(\pi_k,\pi_j):=   \int_{\mathcal{X}}\left\{\sqrt{N}\left|[\tau_N(x,\pi_k)-\tau_N(x,\pi_j)]-\mathbbm{E}[\tau_N(x,\pi_k)-\tau_N(x,\pi_j)]\right|\right\}w(x,\pi_j,\pi_k),$\\
$\mathrm{Bias}(T_{1N}(\pi_k,\pi_j))$ is the asymptotic bias of \,\,$T_{1N}(\pi_k,\pi_j),$ and
${\sigma}_{1N}^2(\pi_k,\pi_j):=\lim_{N\to \infty} Var(T_{1N}(\pi_k,\pi_j))$ under the null.

The proof relies on the proofs in   \cite{gine2003bm} and \cite{chang2015nonparametric}. I borrow most of the notation from  the proof of Lemma B.10 in \cite{chang2015nonparametric}.  To begin, I restate the following fact which  is  Fact 6.1 in \cite{gine2003bm} 

\begin{fact}
 Let $\{(W_i,V_i)': i=1,\dots, N\}$ be a sequence of iid random vectors in $\mathbbm{R}^2$ such that each component has mean zero and variance one and finite absolute moments of the third order. Also, let $(Z_1, Z_2)'$  be a bivariate normal with $\mathbbm{E}[Z_1]=\mathbbm{E}[Z_2]= 0$, $Var(Z_1)=Var(Z_2)= 1$ and $\mathrm{Cov}(Z_1, Z_2)=\mathrm{Cov}(W_i,V_i)=\rho$. Then there exist universal positive constants $A_1, A_2$ and $A_3$ such that
 \begin{align}
     \left|\mathbbm{E}\left| \frac{\sum_{i=1}^N W_i}{\sqrt{N}} \right|- \mathbbm{E}|Z_1|      \right| \leq \frac{A_1}{\sqrt{N}}\mathbbm{E}|W_i|^3
 \end{align}
  and, whenever $\rho^2< 1$
  \begin{align}
     \left|\mathbbm{E}\left| \frac{\sum_{i=1}^N W_i}{\sqrt{N}}\cdot  \frac{\sum_{i=1}^N V_i}{\sqrt{N}} \right|- \mathbbm{E}|Z_1Z_2|      \right| \leq \frac{A_2}{(1-\rho^2)^{3/2}\sqrt{N}}(\mathbbm{E}|W_i|^3 + \mathbbm{E}|V_i|^3)
 \end{align}
 and 
  \begin{align}\label{Berry essen 3}
     \left|\mathbbm{E} \left[\frac{\sum_{i=1}^N W_i}{\sqrt{N}}  \cdot\left| \frac{\sum_{i=1}^N V_i}{\sqrt{N}} \right|  \right]- \mathbbm{E}[Z_1|Z_2|]\right| \leq \frac{A_3}{(1-\rho^2)^{3/2}\sqrt{N}}(\mathbbm{E}|W_i|^3 + \mathbbm{E}|V_i|^3)
 \end{align}
 \end{fact}
 
 Now I use the "Poissonization"  technique of \cite{gine2003bm}. Let $\mathcal{N}$ denote a Poisson random variable with mean  $N$ defined on the same probability space as the sequence $\{W_i:i\geq 1\}:=\{(Y_i,X_i,T_i, \Pi_i ): i\geq 1\}$ and independent of this sequence.
 Define
 $$\chi_{_{t,k}}:=\frac{\mathbbm{E}[Y|X=x,\Pi=\pi_k, T=t]}{P_t(x,\pi_k)},$$
 $$\chi(\pi_k, x,T):=[\chi_{1,k}(\pi_k, x)\cdot T-\chi_{0,k}(\pi_k, x)\cdot(1-T)]\cdot\mathbbm{1}(\Pi=\pi_k),$$
 and 
 $$\psi(W_i,x,\pi_k):=[Y_i\cdot\mathbbm{1}(\Pi_i=\pi_k)\phi(x,\pi_k,T_i)-\chi(\pi_k, x,T_i)]\frac{1}{h^d}K\left(\frac{x-X_i}{h}\right)+\tau(x;\pi_k).$$
Then
    $$\tau_{_N}(x,\pi_k)=\tau_{_{N0}}(x,\pi_k)+\zeta_{_N}(x,\pi_k)=\frac{1}{N}\sum_{i=1}^{N}\psi(W_i,x,\pi_k).$$
Hence define, 
\begin{align*}
   \Gamma_N(x, \pi_k, \pi_j):=&\tau_{_N}(x,\pi_k)-\tau_{_N}(x,\pi_j)\\
   =&\frac{1}{N}\sum_{i=1}^{N}\psi(W_i,x,\pi_k) -\frac{1}{N}\sum_{i=1}^{N}\psi(W_i,x,\pi_j)\\
   =&\frac{1}{N}\sum_{i=1}^{N}\Theta(W_i,x,\pi_k, \pi_j), 
\end{align*}

where
\begin{align*}
    \Theta(W_i,x,\pi_k, \pi_j):=&\Bigg\{[Y_i\cdot[\mathbbm{1}(\Pi_i=\pi_k)\phi(x,T_i,\pi_k)-\mathbbm{1}(\Pi_i=\pi_j)\phi(x,T_i,\pi_j)]\\
    &-[\chi(\pi_k, x,T_i)-\chi(\pi_j, x,T_i)]\Bigg\}
    \cdot \frac{1}{h^d}K\left(\frac{x-X_i}{h}\right)
+\Gamma(x, \pi_k, \pi_j).
\end{align*}

Now we will Poissonize $\Gamma_N(x, \pi_k, \pi_j).$  To do so, again define
$$\Gamma_{\mathcal{N}}(x,\pi_k, \pi_j)=\frac{1}{N}\sum_{i=1}^{\mathcal{N}}\Theta(W_i,x,\pi_k, \pi_j)$$ where the empty sum is defined to be zero.
Note that by the law of iterated expectation and variance,
\begin{align}
    \mathbbm{E}\Gamma_{\mathcal{N}}(x,\pi_k, \pi_j)=\mathbbm{E}\Gamma_N(x, \pi_k, \pi_j)= \mathbbm{E}[\psi(W,x,\pi_k)] -\mathbbm{E}[\psi(W,x,\pi_j)],
\end{align}

\begin{align}
   \kappa_{\tau, N}(x,\pi_k,\pi_j):=NVar(\Gamma_{\mathcal{N}}(x,\pi_k, \pi_j))= \mathbbm{E}[\Theta^2(W,x,\pi_k, \pi_j)],
\end{align}

\begin{align}
   \kappa_{\tau, N}(x,\pi):=NVar(\tau_{\mathcal{N}}(x,\pi))
\end{align}

and 
\begin{align}
    NVar(\Gamma_N(x,\pi_k, \pi_j))= \mathbbm{E}[\Theta^2(W_i,x,\pi_k, \pi_j)]-\{\mathbbm{E}[\Theta(W_i,x,\pi_k, \pi_j)]\}^2.
\end{align}

Let $\epsilon \in \left(0,\int_{\mathcal{X}}f(x)dx\right)$
be an arbitrary constant. For constant $\{M_l>0: l=1,\dots, d\},$ let $\mathcal{B}(M)=\Pi_{l=1}^d[-M_l, M_l] \subset\mathcal{X}$ denote a Borel set in  $\mathbbm{R}^d$ with nonempty interior with finite Lebesgue measure $\lambda(\mathcal{B}(M)).$ For $v>0,$ define $\mathcal{B}(M,v)$ to be the $v$-contraction of $\mathcal{B}(M),$ i.e., $\mathcal{B}(M,v)=\{x \in \mathcal{B}(M):\inf_{y \in \mathbbm{R}^d \setminus B(M)}\{||x-y||\}\geq v\}, $  Choose $M,v>0$ and a Borel set $B$ such that
\begin{align}\label{M1}
    B \subset \mathcal{B}(M,v),
\end{align}

\begin{align}\label{M2}
 \int_{\mathbbm{R}^d\setminus\mathcal{B}(M)} f(x) dx:=\alpha >0,
\end{align}
\begin{align}\label{M3}
\int_{\mathcal{X}} f(x) dx -\int_{B} f(x)dx>\epsilon.
\end{align}
Such $M,v$ and $B$ exist by the absolute continuity of the  density $f$.
Lets define a  Poissonized version of  $T_{1N}(\pi_k,\pi_j)$ (restricted to B)  under the null hypothesis to be:
\begin{align*}
    T_{1N}^P(B):=& \int_{B}  \left\{\sqrt{N}\left|[(\Gamma_{\mathcal{N}}(x,\pi_k, \pi_j)]-\mathbbm{E}[(\Gamma_{\mathcal{N}}(x,\pi_k, \pi_j)]\right|\right\}w(x,\pi_j,\pi_k)dx\\
    &-\int_{B}  \left\{\sqrt{N}\mathbbm{E}\left|[(\Gamma_{\mathcal{N}}(x,\pi_k, \pi_j)]-\mathbbm{E}[(\Gamma_{\mathcal{N}}(x,\pi_k, \pi_j)]\right|\right\}w(x,\pi_j,\pi_k) dx.
\end{align*}

Let $$\sigma^2_{1N}(B)=\mathrm{Var}(T_{1N}^P(B)).$$
The following lemma shows that the variance of the approximated and poissonized version of $\hat{T}_1$ in $B$ converges to the variance of the "un-approximated" and  "un-poissonized" version  in $B$
\begin{lemma}\label{asymptotic variance}
 If the regularity conditions holds and $B$  satisfies (\ref{M1})-(\ref{M3}), then 
 \begin{align}
     \lim_{C\mapsto \infty}\sigma^2_{1N}(B)=\mathrm{P}(B)\cdot \mathrm{Var}(\hat{T}_1)= \sigma^2_{1,B},
 \end{align}
 where 
\begin{align}
  \sigma^2_{1,B}:=&\int_{-1}^1\int_{B_0}\mathrm{Cov}(|\sqrt{1-\rho^2(x,\pi_i, \pi_j,\pi_k, \pi_l,r)}\mathbbm{Z}_1 + \rho(x,\pi_i, \pi_j,\pi_k, \pi_l,r)\mathbbm{Z}_2 |, |\mathbbm{Z}_2|)\nonumber\\
  &\cdot \sqrt{\rho_2(x,\pi_i,\pi_j)\rho_2(x,\pi_k,\pi_l)}\cdot w(x,\pi_i, \pi_j) w(x',\pi_k, \pi_l)dxdr. 
\end{align}
\end{lemma}

\begin{proof}
Note that for each $(x,\pi_i,\pi_j),(x',\pi_k,\pi_l) \in \mathbbm{R}^d\times \Pi^2$ such that $||x-x'||>h,$ the random variables $\Gamma_{\mathcal{N}}(x,\pi_i, \pi_j)-\mathbbm{E}[\Gamma_{\mathcal{N}}(x,\pi_i, \pi_j)]$  and $\Gamma_{\mathcal{N}}(x',\pi_k, \pi_l)-\mathbbm{E}[\Gamma_{\mathcal{N}}(x',\pi_k, \pi_l)]$ are independent because they are functions of independent increments of a Poisson process and the kernel $K$ vanishes outside of the closed rectangle $[-1,1]^d$.
Therefore,
\begin{align*}
 \mathrm{Var}(T_{1N}^P(B))=& \int_{\mathcal{X}}\int_{\mathcal{X}} \mathrm{Cov}(\sqrt{N}|\Gamma_{\mathcal{N}}(x,\pi_i, \pi_j)-\mathbbm{E}[\Gamma_{\mathcal{N}}(x,\pi_i, \pi_j)]|,\sqrt{N}|\Gamma_{\mathcal{N}}(x',\pi_k, \pi_l)-\mathbbm{E}[\Gamma_{\mathcal{N}}(x',\pi_k, \pi_l)]|)\\
 &\cdot w(x,\pi_i, \pi_j) w(x',\pi_k, \pi_l)dxdx' \\ 
 =& \int_{\mathcal{X}}\int_{\mathcal{X}} \mathrm{Cov}(\sqrt{N}|\Gamma_{\mathcal{N}}(x,\pi_i, \pi_j)-\mathbbm{E}[\Gamma_{\mathcal{N}}(x,\pi_i, \pi_j)]|,\sqrt{N}|\Gamma_{\mathcal{N}}(x',\pi_k, \pi_l)-\mathbbm{E}[\Gamma_{\mathcal{N}}(x',\pi_k, \pi_l)]|)\\
 &\cdot \mathbbm{1}(h^{-1}(x-x')\in [-1,1]^d)\cdot w(x,\pi_i, \pi_j) w(x',\pi_k, \pi_l)dxdx',
\end{align*}

Set 
\begin{align}
    S_{\tau,\mathcal{N} }(x,\pi, \pi')=&\frac{\sqrt N\{\Gamma_{\mathcal{N}}(x,\pi, \pi')-\mathbbm{E}[\tau_{\mathcal{N}}(x,\pi, \pi')]\}}{\sqrt{\kappa_{\Gamma,N}(x,\pi,\pi')}}\nonumber\\
    =& \frac{\sqrt N\{\tau_{\mathcal{N}}(x,\pi )-\mathbbm{E}[\tau_{\mathcal{N}}(x,\pi )]\}}{\sqrt{\kappa_{\Gamma,N}(x,\pi,\pi')}} + \frac{\sqrt N\{\tau_{\mathcal{N}}(x,\pi' )-\mathbbm{E}[\tau_{\mathcal{N}}(x,\pi' )]\}}{\sqrt{\kappa_{\Gamma,N}(x,\pi,\pi')}}\nonumber\\
    \leq& \frac{\sqrt N\{\tau_{\mathcal{N}}(x,\pi )-\mathbbm{E}[\tau_{\mathcal{N}}(x,\pi )]\}}{\sqrt{\kappa_{\tau,N}(x,\pi)}} + \frac{\sqrt N\{\tau_{\mathcal{N}}(x,\pi' )-\mathbbm{E}[\tau_{\mathcal{N}}(x,\pi' )]\}}{\sqrt{\kappa_{\tau,N}(x,\pi')}} \nonumber\\
    =:& S_{\tau,\mathcal{N} }(x, \pi)+ S_{\tau,\mathcal{N} }(x,\pi').
\end{align}

Now, with  $\int_{B_0} dx< \infty,$ by the Lemma \ref{Gine}, we can show that
\begin{align}
    \sup_{x\in B_0 }\left| \sqrt{\kappa_{\Gamma,N}(x,\pi,\pi')} - h^{-d/2}\sqrt{\rho_2(x,\pi, \pi')}   \right|= O(h^{d/2}),
\end{align}
where 
\footnotesize
\begin{equation*}
 \rho_2(x,\pi, \pi')=\left\{ \sum_{t \in \{0,1\}}\sum_{
\pi_0\in \{\pi,\pi'\}}\frac{\mathbbm{E}[Y^2|X=x,T=t, \Pi=\pi_0]-(\mathbbm{E}[Y|X=x,T=t, \Pi=\pi_0] )^2}{P_t(x,\pi_0)}\right\}\cdot \int K^2(\xi)d\xi
\end{equation*}
\normalsize
We also have that
\begin{align}
     \int_{B_0}\int_{B_0}  \mathbbm{1}(h^{-1}(x-x')\in  [-1,1]^d)\cdot w(x,\pi_i, \pi_j) w(x',\pi_k, \pi_l)dxdx'=O(h^d). 
\end{align}
 And finally, by the  Cauchy Schwartz Inequality, we have 
 \begin{align}
    \sup_{(x,x') \in B_0 ^2 }\left|\mathrm{Cov}(| S_{\tau,\mathcal{N} }(x,\pi_i, \pi_j)|,| S_{\tau,\mathcal{N} }(x',\pi_k, \pi_l)|)    \right|= O(1), \,\, \forall\,\, (\pi_i, \pi_j, \pi_k, \pi_l) \in \boldsymbol{\Pi}^4. 
\end{align}
Therefore, we have 
$$\mathrm{Var}(T_{1N}^P(B))=\bar{\sigma}^2_{1N} + o(1),$$
where 
\begin{align*}
   \bar{\sigma}^2_{1N}:= \int_{\mathcal{X}}\int_{\mathcal{X}}  & \mathrm{Cov}(| S_{\tau,\mathcal{N} }(x,\pi_i, \pi_j)|,| S_{\tau,\mathcal{N} }(x',\pi_k, \pi_l)|)
  \cdot \sqrt{\rho_2(x,\pi_i,\pi_j)\rho_2(x',\pi_k,\pi_l)}\\
 &\times h^{-d}\cdot \mathbbm{1}(h^{-1}(x-x')\in  [-1,1]^d)\cdot w(x,\pi_i, \pi_j) w(x',\pi_k, \pi_l)dxdx'
\end{align*}

Now, for $b \in \{1,2\}$, let $Z_{bn}(x,\pi_i,\pi_j)=Z_{bn}(x,\pi_i)+Z_{bn}(x,\pi_j)$  where ($Z_{1n}(x,\pi_i)$, $Z_{2n}(x,\pi_i)$, $Z_{1n}(x,\pi_j)$, $Z_{2n}(x,\pi_j)$) are mean zero pairwise independent Gaussian processes.\\ Then  $(Z_{1n}(x,\pi_i,\pi_j), Z_{2n}(x',\pi_k,\pi_l))$ for $(x,\pi_i,\pi_j),(x',\pi_k,\pi_l) \in \mathbbm{R}^d\times\Pi^2,$ is  a mean zero bivariate Gaussian process. Now  for each $(x,\pi_i,\pi_j)\in \mathbbm{R}^d\times\Pi^2,(x',\pi_k,\pi_l) \in \mathbbm{R}^d\times\Pi^2,$  let \,\,$(Z_{1n}(x,\pi_i,\pi_j), Z_{2n}(x',\pi_k,\pi_l))$ and \,\,$(S_{\tau,\mathcal{N} }(x,\pi_i, \pi_j), S_{\tau,\mathcal{N} }(x',\pi_k, \pi_l))$ \,\,have the same covariance structure. That is,
$$(Z_{1n}(x,\pi_i,\pi_j), Z_{2n}(x,\pi_k,\pi_l))\overset{d}{=}\left(\sqrt{1-\rho^*_N(x,\pi_i, \pi_j,\pi_k, \pi_l)^2}\mathbbm{Z}_1 +\rho^*_N(x',\pi_i, \pi_j,\pi_k, \pi_l)\mathbbm{Z}_2, \mathbbm{Z}_2\right),$$
where $\mathbbm{Z}_1$ and $\mathbbm{Z}_2$ are independent standard normal random variables and 
 $$\rho^*_N(x,x',\pi_i, \pi_j,\pi_k, \pi_l):=\mathbbm{E}[S_{\tau,\mathcal{N} }(x,\pi_i, \pi_j) S_{\tau,\mathcal{N} }(x',\pi_k, \pi_l)].$$

Let 
\begin{align*}
   \bar{\tau}^2_{N,0}=&\int_{B_0}\int_{B_0}\mathrm{Cov}(| Z_{1n}(x,\pi_i,\pi_j)|,| Z_{2n}(x',\pi_k,\pi_l)|)
  \cdot \sqrt{\rho_2(x,\pi_i,\pi_j)\rho_2(x',\pi_k,\pi_l)}\\
 & \times h^{-d}\cdot\mathbbm{1}(h^{-1}(x-x')\in  [-1,1]^d)\cdot w(x,\pi_i, \pi_j) w(x',\pi_k, \pi_l)dxdx'.
\end{align*}
 By a change of variables $x'=x+th,$ we can write

\begin{align*}
   \bar{\tau}^2_{N,0}=&\int_{B_0}\int_{[-1,1]^d}\mathrm{Cov}(| Z_{1n}(x,\pi_i,\pi_j)|,| Z_{2n}(x+th,\pi_k,\pi_l)|)\\
 &  \cdot \sqrt{\rho_2(x,\pi_i,\pi_j)\rho_2(x+th,\pi_k,\pi_l)}\times\mathbbm{1}(x\in B_0)\mathbbm{1}(x+th\in B_0)\\
 & \cdot w(x,\pi_i, \pi_j) w(x+th,\pi_k, \pi_l)dxdt.
\end{align*}
Note that 
\begin{align*}
    \rho^*_N(x,x',\pi_i, \pi_j,\pi_k, \pi_l)=& N\mathbbm{E}\left[\frac{\{\Gamma_{\mathcal{N}}(x,\pi_i, \pi_j)-\mathbbm{E}[\Gamma_{\mathcal{N}}(x,\pi_i, \pi_j)]\}}{\sqrt{\kappa_{\Gamma,N}(x,\pi_i,\pi_j)}}\frac{\{\Gamma_{\mathcal{N}}(x',\pi_k, \pi_l)-\mathbbm{E}[\Gamma_{\mathcal{N}}(x',\pi_k, \pi_l)]\}}{\sqrt{\kappa_{\Gamma,N}(x',\pi_k,\pi_l)}}\right]\\
    =& \frac{\mathbbm{E}\left[Q(x,\pi_i,\pi_j)Q(x',\pi_i,\pi_j)Q(x,\pi_k,\pi_l)Q(x',\pi_k,\pi_l)\cdot \frac{1}{h^d}K\left(\frac{x-X}{h}\right)\frac{1}{h^d}K\left(\frac{x'-X}{h}\right)\right]}{\sqrt{\kappa_{\Gamma,N}(x,\pi_i,\pi_j)\kappa_{\Gamma,N}(x',\pi_k,\pi_l)}},
\end{align*}
where $$Q(x,\pi,\pi'):= \left\{[Y\cdot[\mathbbm{1}(\Pi=\pi)\hat{\phi}(x,T,\pi)-\mathbbm{1}(\Pi=\pi')\hat{\phi}(x,T,\pi')]-[\chi(\pi, x,T)-\chi(\pi', x,T)]\right\}.$$
By Lemma \ref{Gine} and a change of variable $x'=x+th$, we can show that
$$\rho^*_N(x,x',\pi_i, \pi_j,\pi_k, \pi_l)\to\frac{\int K(\xi)K(\xi+ t)d\xi}{\int K(\xi)d\xi}.$$

Therefore, by the bounded convergence theorem, we have that 
$$\lim_{n \to \infty}\bar{\tau}^2_{N,0}= \sigma^2_{1,B}.$$
Now, the desired result holds if 
\begin{equation}\label{lastr}
\bar{\tau}^2_{N,0}\to \bar{\sigma}^2_{N,0}\,\, \text{as}\,\, C\to \infty.    
\end{equation}

To prove (\ref{lastr}), let $\epsilon_0 \in (0, 1]$ and let $c(\epsilon_0) = (1 + \epsilon_0)^2 - 1$. Let $\mathcal{Q}_1$ and $\mathcal{Q}_2$ be two independent random
variables that are independent of $(\{Y_i, X_i\}_{i=1}^{\infty}, \mathcal{N})$, each having a two-point distribution that gives two points,
$\{\sqrt{c(\epsilon_0)}\}$ and $\{-\sqrt{c(\epsilon_0)}\}$, the equal mass of 1/2, so that $\mathbbm{E}[\mathcal{Q}_1]=\mathbbm{E}[\mathcal{Q}_2]=0$ and $\mathrm{Var}(\mathcal{Q}_1) = \mathrm{Var}(\mathcal{Q}_2) = c(\epsilon_0)$.
Let $S^\mathcal{Q}_{\tau, \mathcal{N}, 1}(x, \pi, \pi')= \frac{S_{\tau, \mathcal{N}}(x, \pi, \pi') + 2\mathcal{Q}_1}{1+ \epsilon_0}=\frac{S_{\tau, \mathcal{N}}(x, \pi) + \mathcal{Q}_1 +S_{\tau, \mathcal{N}}(x, \pi') + \mathcal{Q}_1}{1+ \epsilon_0}=:S^\mathcal{Q}_{\tau, \mathcal{N}, 1}(x, \pi)+S^\mathcal{Q}_{\tau, \mathcal{N}, 1}(x, \pi') $ and $S^\mathcal{Q}_{\tau, \mathcal{N}, 2}(x, \pi, \pi')= \frac{S_{\tau, \mathcal{N}}(x, \pi \pi') + 2\mathcal{Q}_2}{1+ \epsilon_0}=\frac{S_{\tau, \mathcal{N}}(x, \pi) + \mathcal{Q}_2 +S_{\tau, \mathcal{N}}(x, \pi') + \mathcal{Q}_2}{1+ \epsilon_0}=:S^\mathcal{Q}_{\tau, \mathcal{N}, 2}(x, \pi)+S^\mathcal{Q}_{\tau, \mathcal{N}, 2}(x, \pi').$
Define
\begin{align}
   \bar{\sigma}^2_{1N,\mathcal{Q}}:= &\int_{\mathcal{X}}\int_{\mathcal{X}} \mathrm{Cov}(| S^\mathcal{Q}_{\tau,\mathcal{N}, 1 }(x,\pi_i, \pi_j)|,| S^\mathcal{Q}_{\tau,\mathcal{N}, 2 }(x,\pi_k, \pi_l)|)
  \cdot \sqrt{\rho_2(x,\pi_i,\pi_j)\rho_2(x',\pi_k,\pi_l)} \nonumber\\
 & \times h^{-d}\cdot \mathbbm{1}(h^{-1}(x-x')\in  [-1,1]^d)\cdot w(x,\pi_i, \pi_j) w(x',\pi_k, \pi_l)dxdx',
\end{align}
and let \,$Z^{\mathcal{Q}}_{1, N}(x, \pi, \pi')= \frac{Z_{1, N}(x, \pi) + Z_{1, N}(x, \pi')+ 2\mathcal{Q}_1}{1+ \epsilon_0}=:Z^{\mathcal{Q}}_{1, N}(x, \pi)+Z^{\mathcal{Q}}_{1, N}(x, \pi')$ and $Z^{\mathcal{Q}}_{2, N}(x, \pi, \pi')= \frac{Z_{2, N}(x, \pi)+Z_{2, N}(x, \pi') + 2\mathcal{Q}_2}{1+ \epsilon_0}=:Z^{\mathcal{Q}}_{2, N}(x, \pi)+Z^{\mathcal{Q}}_{2, N}(x, \pi').$
Then $(Z^{\mathcal{Q}}_{1,N}(x, \pi_i, \pi_j),Z^{\mathcal{Q}}_{2,N}(x', \pi_k, \pi_l))$ is a mean zero multivariate Gaussian process such that, for each $(x, \pi_i, \pi_j) \in
\mathbbm{R}\times \Pi^2$ and $(x, \pi_k, \pi_l) \in \mathbbm{R}\times \Pi^2$, $(Z^{\mathcal{Q}}_{1,N}(x, \pi_i, \pi_j),Z^{\mathcal{Q}}_{2,N}(x', \pi_k, \pi_l))$ and $(S^\mathcal{Q}_{\tau, \mathcal{N}, 1}(x, \pi_i, \pi_j),S^\mathcal{Q}_{\tau, \mathcal{N}, 2}(x', \pi_k, \pi_l))$ have the same covariance
structure.

 Also, define
 \begin{align*}
   \bar{\tau}^2_{N,\mathcal{Q}}=&\int_{B_0}\int_{B_0} \mathrm{Cov}(| Z^{\mathcal{Q}}_{1,N}(x, \pi_i, \pi_j)|,| Z^{\mathcal{Q}}_{2,N}(x', \pi_k, \pi_l)|)\\
 & \cdot \sqrt{\rho_2(x,\pi_i,\pi_j)\rho_2(x',\pi_k,\pi_l)}\times h^{-d}\cdot\mathbbm{1}(h^{-1}(x-x')\in  [-1,1]^d)\\
 &\cdot w(x,\pi_i, \pi_j) w(x',\pi_k, \pi_l)dxdx'.
\end{align*}
Using the triangle inequality, we have
\footnotesize
\begin{align*}
 |\bar{\tau}^2_{N,\mathcal{Q}}-\bar{\sigma}^2_{N,\mathcal{Q}} |  =&\Bigg|\int_{B_0}\int_{B_0}\Bigg(\mathrm{Cov}(| Z^{\mathcal{Q}}_{1,N}(x, \pi_i, \pi_j)|,| Z^{\mathcal{Q}}_{2,N}(x', \pi_k, \pi_l)|) -\mathrm{Cov}(| S^\mathcal{Q}_{\tau,\mathcal{N}, 1 }(x,\pi_i, \pi_j)|,| S^\mathcal{Q}_{\tau,\mathcal{N}, 2 }(x,\pi_k, \pi_l)|)\Bigg)\\ &\cdot \sqrt{\rho_2(x,\pi_i,\pi_j)\rho_2(x',\pi_k,\pi_l)}
  \times h^{-d}\cdot\mathbbm{1}(h^{-1}(x-x')\in  [-1,1]^d)\cdot w(x,\pi_i, \pi_j) w(x',\pi_k, \pi_l)dxdx'\Bigg|\\
 \leq& \int_{B_0}\int_{B_0}\Bigg| \Bigg(\mathbbm{E}| Z^{\mathcal{Q}}_{1,N}(x, \pi_i, \pi_j)|\mathbbm{E}| Z^{\mathcal{Q}}_{2,N}(x', \pi_k, \pi_l)| 
 -\mathbbm{E}| S^\mathcal{Q}_{\tau,\mathcal{N}, 1 }(x,\pi_i, \pi_j)|\mathbbm{E}| S^\mathcal{Q}_{\tau,\mathcal{N}, 2 }(x,\pi_k, \pi_l)|\Bigg) \\ & \cdot \sqrt{\rho_2(x,\pi_i,\pi_j)\rho_2(x',\pi_k,\pi_l)}
  \times h^{-d}\cdot\mathbbm{1}(h^{-1}(x-x')\in  [-1,1]^d)\cdot w(x,\pi_i, \pi_j) w(x',\pi_k, \pi_l)dxdx'\Bigg|\\
 +& \int_{B_0}\int_{B_0}\Bigg| \Bigg(\mathbbm{E}| Z^{\mathcal{Q}}_{1,N}(x, \pi_i, \pi_j)|| Z^{\mathcal{Q}}_{2,N}(x', \pi_k, \pi_l)| -\mathbbm{E}| S^\mathcal{Q}_{\tau,\mathcal{N}, 1 }(x,\pi_i, \pi_j)|| S^\mathcal{Q}_{\tau,\mathcal{N}, 2 }(x,\pi_k, \pi_l)|\Bigg) \\
 & \cdot \sqrt{\rho_2(x,\pi_i,\pi_j)\rho_2(x',\pi_k,\pi_l)}
  \times h^{-d}\cdot\mathbbm{1}(h^{-1}(x-x')\in  [-1,1]^d)\cdot w(x,\pi_i, \pi_j) w(x',\pi_k, \pi_l)dxdx'\Bigg|\\
 \leq& \int_{B_0}\int_{B_0}\Bigg| \Bigg(\mathbbm{E}| Z^{\mathcal{Q}}_{1,N}(x, \pi_i)|\mathbbm{E}| Z^{\mathcal{Q}}_{2,N}(x', \pi_k)| -\mathbbm{E}| S_{\tau,\mathcal{N}, 1 }(x,\pi_i)|\mathbbm{E}| S_{\tau,\mathcal{N}, 2 }(x',\pi_k)| \\
 &+ \mathbbm{E}| Z^{\mathcal{Q}}_{1,N}(x, \pi_i)|\mathbbm{E}| Z^{\mathcal{Q}}_{2,N}(x', \pi_l)| -\mathbbm{E}| S_{\tau,\mathcal{N}, 1 }(x,\pi_i)|\mathbbm{E}| S_{\tau,\mathcal{N}, 2 }(x',\pi_l)|\\
 &+ \mathbbm{E}| Z^{\mathcal{Q}}_{1,N}(x, \pi_j)|\mathbbm{E}| Z^{\mathcal{Q}}_{2,N}(x', \pi_k)| -\mathbbm{E}| S_{\tau,\mathcal{N}, 1 }(x,\pi_j)|\mathbbm{E}| S_{\tau,\mathcal{N}, 2 }(x',\pi_k)|\\
  &+ \mathbbm{E}| Z^{\mathcal{Q}}_{1,N}(x, \pi_j)|\mathbbm{E}| Z^{\mathcal{Q}}_{2,N}(x', \pi_l)| -\mathbbm{E}| S_{\tau,\mathcal{N}, 1 }(x,\pi_j)|\mathbbm{E}| S_{\tau,\mathcal{N}, 2 }(x',\pi_l)|
 \Bigg) \\ 
 & \cdot \sqrt{\rho_2(x,\pi_i,\pi_j)\rho_2(x',\pi_k,\pi_l)}
 \times h^{-d}\cdot\mathbbm{1}(h^{-1}(x-x')\in  [-1,1]^d)\cdot w(x,\pi_i, \pi_j) w(x',\pi_k, \pi_l)dxdx'\Bigg|\\
 &+ \int_{B_0}\int_{B_0}\Bigg| \Bigg(\mathbbm{E}| Z^{\mathcal{Q}}_{1,N}(x, \pi_i)|| Z^{\mathcal{Q}}_{2,N}(x', \pi_k)| -\mathbbm{E}| S_{\tau,\mathcal{N}, 1 }(x,\pi_i)|| S_{\tau,\mathcal{N}, 2 }(x',\pi_k)| \\
 &+ \mathbbm{E}| Z^{\mathcal{Q}}_{1,N}(x, \pi_i)|| Z^{\mathcal{Q}}_{2,N}(x', \pi_l)| -\mathbbm{E}| S_{\tau,\mathcal{N}, 1 }(x,\pi_i)|| S_{\tau,\mathcal{N}, 2 }(x',\pi_l)|\\
 &+ \mathbbm{E}| Z^{\mathcal{Q}}_{1,N}(x, \pi_j)|| Z^{\mathcal{Q}}_{2,N}(x', \pi_k)| -\mathbbm{E}| S_{\tau,\mathcal{N}, 1 }(x,\pi_j)|| S_{\tau,\mathcal{N}, 2 }(x',\pi_k)|\\
  &+ \mathbbm{E}| Z^{\mathcal{Q}}_{1,N}(x, \pi_j)|| Z^{\mathcal{Q}}_{2,N}(x', \pi_l)| -\mathbbm{E}| S_{\tau,\mathcal{N}, 1 }(x,\pi_j)|| S_{\tau,\mathcal{N}, 2 }(x',\pi_l)|\Bigg)\\ & \cdot \sqrt{\rho_2(x,\pi_i,\pi_j)\rho_2(x',\pi_k,\pi_l)}
 \times h^{-d}\cdot\mathbbm{1}(h^{-1}(x-x')\in  [-1,1]^d)\cdot w(x,\pi_i, \pi_j) w(x',\pi_k, \pi_l)dxdx'\Bigg|\\
 =:&\Delta_{1N,\mathcal{Q}}(\pi_i,\pi_k)+\Delta_{1N,\mathcal{Q}}(\pi_i,\pi_l)+\Delta_{1N,\mathcal{Q}}(\pi_j,\pi_k)+\Delta_{1N,\mathcal{Q}}(\pi_j,\pi_l) \\
 &+\Delta_{2N,\mathcal{Q}}(\pi_i,\pi_k)+\Delta_{2N,\mathcal{Q}}(\pi_i,\pi_l)+\Delta_{2N,\mathcal{Q}}(\pi_j,\pi_k)+\Delta_{2N,\mathcal{Q}}(\pi_j,\pi_l).
\end{align*}
\normalsize
Following the proof of Lemma B.10 in \cite{chang2015nonparametric}, note that for all $\{\pi,\pi'\} \in \Pi,$
$\Delta_{1N\mathcal{Q}}(\pi,\pi')= o(1)$,\, and $\Delta_{2N\mathcal{Q}}(\pi,\pi')= o(1).$
 It is straightforward to show that
$|\bar{\sigma}^2_{N,\mathcal{Q}}-\bar{\sigma}^2_{N,0} |\to 0$ and 
$|\bar{\tau}^2_{N,\mathcal{Q}}-\bar{\tau}^2_{N,0} |\to 0$ as $\epsilon_0 \to 0.$ Hence the triangle inequality establishes (\ref{lastr}), and thus the lemma has been proved.
\end{proof}

Let $M$ be defined as in (\ref{M1})-(\ref{M3}) and let 

$$U_{_N}:= \frac{1}{\sqrt{N}}\left\{\sum_{i=1}^\mathcal{N}\mathbbm{1}[X_i \in \mathcal{B}(M)]-N \Pr(X \in \mathcal{B}(M))\right\}$$ and 
$$V_{_N}:= \frac{1}{\sqrt{N}}\left\{\sum_{i=1}^\mathcal{N}\mathbbm{1}[X_i \in \mathbbm{R}^{d}\setminus \mathcal{B}(M)]-N \Pr(X \in \mathbbm{R}^{d}\setminus\mathcal{B}(M))\right\}$$
Also define 
$$S_N:=\frac{T_{_{1N}}^P(B)}{\sigma_{1N}(B)}$$
\begin{lemma}
 Under the regularity conditions, we have that 
 $$(S_N, U_N)\overset{d}{\to}(\mathbbm{Z}_1, \sqrt{1-\alpha}\mathbbm{Z}_2)$$ where $\mathbbm{Z}_1$ and $\mathbbm{Z}_2$ are independent $N(0,1)$ random variables and $\alpha$ is defined as in (\ref{M2}).
\end{lemma}
\begin{proof}
Let
\begin{align*}
    \Delta_N(\pi_k,\pi_j, x)=&\sqrt{N}\Bigg\{\left| \Gamma_{\mathcal{N}}(x,\pi_k,\pi_j)-\mathbbm{E}[\Gamma_{\mathcal{N}}(x,\pi_k,\pi_j)] \right|\\
    &- \mathbbm{E}\left| \Gamma_{\mathcal{N}}(x,\pi_k,\pi_j)-\mathbbm{E}[\Gamma_{\mathcal{N}}(x,\pi_k,\pi_j)] \right| \Bigg\}\cdot w(x,\pi_k, \pi_j)
\end{align*}

Construct a partition of $\mathcal{B}(M).$ Consider  a regular grid
$G_{\textbf{i}}=(i_1h, (i_1 +1)h]\times\dots\times(i_dh, (i_d +1)h]$ where $\textbf{i}=(i_1, \dots, i_{d}),$ \,\,$i_1, \dots, i_{d}$ are integers.
Define $R_{\textbf{i}}=G_{\textbf{i}}\cap\mathcal{B}(M),$ $\mathcal{I}_i=\{\textbf{i}\in \mathbbm{Z}^d :(G_{\textbf{i}}\cap\mathcal{B}(M))\neq \emptyset \}.$ Then, we see that $\{R_{\textbf{i}}:\textbf{i}\in\mathcal{I}_i \subset  \mathbbm{Z}^d\}$ is a partition of $\mathcal{B}(M)$ with $\lambda(R_{\textbf{i}})\leq A_1h^d,$ $m_N:=\#(\mathcal{I}_i)\leq A_2h^{-d}$  for some positive constants $A_1$ and $A_2$.

In addition,set $$\alpha_{\textbf{i},N}=\frac{\int_{R_{\textbf{i}}}\mathbbm{1}(x \in B)\cdot\sum_{k=1}^K\sum_{j=1}^K\Delta_N(\pi_k,\pi_j, x) dx}{\sigma_N(B)}$$ and 
$$u_{\textbf{i},N}=\frac{1}{\sqrt{N}}\left\{ \sum_{j=1}^{\mathcal{N}}\mathbbm{1}(X_j \in R_{\textbf{i}} )- N\Pr(X \in R_{\textbf{i}}) \right\}.$$
Then, we have\,\,
$S_N= \sum_{\textbf{i}\in \mathcal{I}_i}\alpha_{\textbf{i},N}$ and $U_N= \sum_{\textbf{i}\in \mathcal{I}_i}u_{\textbf{i},N}.$ 
We can verify  that 
$\mathrm{Var}(S_N)= 1$ and $\mathrm{Var}(U_N)=1-\alpha.$
For arbitrary $\lambda_1$ and $\lambda_2 \in \mathbbm{R},$ let
$$y_{\textbf{i},N}=\lambda_1\alpha_{\textbf{i},N} +\lambda_2u_{\textbf{i},N}.$$ Notice that $\{y_{\textbf{i},N}: \textbf{i} \in \mathcal{I}_i \}$ is an array of mean zero one-dependent random fields. 

We need to show that
\begin{align}\label{cond 1}
\mathrm{Var}\left(\sum_{\textbf{i} \in \mathcal{I}_i}y_{_{\textbf{i},N}}\right)=& \mathrm{Var}(\lambda_1 S_N +\lambda_2 U_N) \to  \lambda_1^2 +\lambda_2^2 (1-\alpha)
\end{align}
and 
\begin{align}\label{cond 2}
\sum_{\textbf{i} \in \mathcal{I}_i}\mathbbm{E}|y_{_{\textbf{i},N}}|^r =o(1)\, \text{ for some } 2<r<3.
\end{align}
The  results of the Lemma follows from the central limit theorem of   \cite{shergin1993central} and the Cram\'er-Wold device.
 To show, \eqref{cond 1}, which holds if we have 
 \begin{align}
     &\mathrm{Cov}(S_N, U_N)=O\left(\frac{1}{\sqrt{Nh^{2d}}}\right),
      \end{align}
 which implies that
\begin{align}\label{cond 3}
     &\mathrm{Cov}\left(\int_{B}  \left\{\sqrt{N}\left|[(\Gamma_{\mathcal{N}}(x,\pi_k, \pi_j)]-\mathbbm{E}[(\Gamma_{\mathcal{N}}(x,\pi_k, \pi_j)]\right|\right\}w(x,\pi_j,\pi_k) dx, U_N\right)=O\left(\frac{1}{\sqrt{Nh^{2d}}}\right). 
 \end{align}

For any $(x, \pi_k, \pi_j) \in B \times \Pi^2$ we have 
\begin{align}\label{dist eq}
    \left(S_{_{\tau,\mathcal{N} }}(x,\pi_k, \pi_j), \frac{U_N}{\sqrt{\Pr(X \in B(M))}}\right)\overset{d}{=}\left(\frac{1}{\sqrt{N}}\sum_{i=1}^NQ^{(i)}_{_{\tau,N}}(x,\pi_k, \pi_j), \frac{1}{\sqrt{N}}\sum_{i=1}^NU^{(i)}\right)
\end{align}
where \,$(Q^{(i)}_{_{\tau,N}}(x,\pi_k, \pi_j),U^{(i)})$\, for $i=1,\dots, N$ are i.i.d. copies of $(Q_{_{\tau,N}}(x,\pi_k, \pi_j), U)$ with $Q_{_{\tau,N}}(x,\pi_k, \pi_j)$ is defined as 
\begin{align*}
    Q_{_{\tau,N}}(x,\pi_k, \pi_j)= \left[\sum_{i=1}^\eta \Theta(W_i,x,\pi_k, \pi_j)- \mathbbm{E}[\Theta(W,x,\pi_k, \pi_j)]\right]/ \sqrt{\mathbbm{E}[\Theta^2(W_i,x,\pi_k, \pi_j)]}
\end{align*}
and
\begin{align*}
   U= \left[\sum_{i=1}^\eta \mathbbm{1}[X_i \in \mathcal{B}(M)]-\Pr(X \in \mathcal{B}(M))\right]/ \sqrt{\Pr(X \in \mathcal{B}(M))}.
\end{align*}
where $\eta$ denote an independent Poisson random variable with mean 1 that is independent of $\{W_i:n\geq 1\}.$ 
Let $(Z_{1N}(x,\pi_k, \pi_j),Z_{2N})$\, for \,$(x,\pi_k, \pi_j) \in \mathbbm{R}\times\Pi^2$ be a mean zero Gaussian process such that, for each $(x,\pi_k, \pi_j)\in \mathbbm{R}\times\Pi^2,$\, $(Z_{1N}(x,\pi_k, \pi_j),Z_{2N})$ and the left-hand side of (\ref{dist eq}) has the same covariance structure. That is, 

$$(Z_{1N}(x,\pi_k, \pi_j),Z_{2N})\overset{d}{=}(\sqrt{1- (\gamma^*_i(x,\pi_k, \pi_j))^2}\mathbbm{Z}_1 + \gamma^*_i(x,\pi_k, \pi_j))\mathbbm{Z}_2,  \mathbbm{Z}_2),$$
where $\mathbbm{Z}_1$ and $\mathbbm{Z}_2$ are independent standard normal random variables and 
$$\gamma^*_i(x,\pi_k, \pi_j))=\mathbbm{E}\left[S_{\tau, \mathcal{N}}(x,\pi_k, \pi_j).\,\frac{U_N}{\Pr(X \in \mathcal{B}(M))}\right].$$

We can show that
$$\sup_{ B}\left|\mathbbm{E}\left[\frac{\tau_{\mathcal{N}}(x, \pi_k)-\mathbbm{E}[\tau_{\mathcal{N}}(x, \pi_k)]}{\sqrt{\kappa_{\Gamma,N}(x,\pi_k,\pi_j)}}\cdot \frac{U_N}{\Pr(X \in \mathcal{B}(M))}\right]\right|= O(h^{d/2}), \,\, \forall \pi\in \boldsymbol{\Pi}.$$

Using the triangle inequality, notice that we have 
\begin{align} \label{Bounded in Probability eq}
   \sup_{B }|\gamma^*_i(x,\pi_k, \pi_j))| =O(h^{d/2}),  \,\, \forall \pi\in \boldsymbol{\Pi},
\end{align}
which in turn is less than or equal to $\epsilon$ for all sufficiently large $N$ and any $0<\epsilon<1/2.$
This result and (\ref{Berry essen 3})  imply that 
\begin{align}
    \sup_{B }\left| \mathrm{Cov}\left(|S_{\tau, \mathcal{N}}(x,\pi_k, \pi_j)|, \frac{U_N}{\Pr(X \in \mathcal{B}(M))} \right) - \mathbbm{E}[|Z_{1N}(x,\pi_k, \pi_j)|Z_{2N}]\right|\leq O\left(\frac{1}{\sqrt{N h^{2d}}}\right)
\end{align}
Again, using the  triangle inequality, this implies that
\begin{align}\label{bounded in prob 1}
    &\sup_{B}\left| \mathrm{Cov}\left(\left|S_{\tau, \mathcal{N}}(x,\pi_k, \pi_j)\right|, \frac{U_N}{\Pr(X \in \mathcal{B}(M))} \right) - \mathbbm{E}\left[\left|\sum_{k=1}^K \sum_{j=1}^KZ_{1N}(x,\pi_k, \pi_j)\right|Z_{2N}\right]\right| \nonumber\\
    &\leq O\left(\frac{1}{\sqrt{N\cdot h^{2d}}}\right).
\end{align}

On the other hand, for all $\pi_k, \pi_j \in \boldsymbol{\Pi}$, 
\begin{align}\label{bounded in prob 2}
\sup_{B}|\mathbbm{E}[Z_{1N}(x,\pi_k, \pi_j)|Z_{2N}|=&\sup_{B }|\gamma^*_i(x,\pi_k, \pi_j)\mathbbm{E}[|Z_{1N}(x,\pi_k, \pi_j)|Z_{1N}(x,\pi_k, \pi_j)]\\
\leq& \sup_{B }|\gamma^*_i(x,\pi_k, \pi_j)|\mathbbm{E}[Z^2_{1N}(x,\pi_k, \pi_j)]\\
=& \sup_{B }|\gamma^*_i(x,\pi_k, \pi_j)|= O(h^{d/2}).
\end{align}
Using the law of iterated expectations and (\ref{Bounded in Probability eq}).

Therefore, (\ref{bounded in prob 1}) and (\ref{bounded in prob 2}) imply that 

$$\sup_{B}|\mathrm{Cov}(\sqrt{N}|[(\Gamma_{\mathcal{N}}(x,\pi_k, \pi_j)]-\mathbbm{E}[(\Gamma_{\mathcal{N}}(x,\pi_k, \pi_j)]|, U_N)|\leq O\left(\frac{1}{Nh^{2d}} + h^{d/2} \right)$$ 

which when combined with $\lambda(B)<\infty$ yields \eqref{cond 3} and  hence \eqref{cond 1} as desired.

Next we establish (\ref{cond 2}). \cite{chang2015nonparametric} shows that for any $\pi_k \in \Pi,$
\begin{align}\label{equation chang}
&\sum_{\textbf{i} \in \mathcal{I}_i}\mathbbm{E}\left|\frac{\int_{R_{\textbf{i}}}\mathbbm{1}(x \in B) \sqrt{N_k}\{|\tau_{\mathcal{N}}(x,\pi_k)-\mathbbm{E}\tau_{\mathcal{N}}(x,\pi_k) |- \mathbbm{E}[|\tau_{\mathcal{N}}(x,\pi_k)-\mathbbm{E}\tau_{\mathcal{N}}(x,\pi_k) |]\}w(x, \pi_k)}{\sigma_{_{N_{_k}}}(B)}\right|^{r}\nonumber\\
&\leq O(N_k\cdot h^{rd/2})=o(1)
\end{align}
Notice that,
\begin{align}\label{conc 1}
\begin{split}
\sum_{\textbf{i} \in \mathcal{I}_i}\mathbbm{E}|\alpha_{\textbf{i},N}|=&\sum_{\textbf{i} \in \mathcal{I}_i}\mathbbm{E}\left|\int_{R_{\textbf{i}}} \mathbbm{1}(x \in B)\cdot  \Delta_i(\pi_k, \pi_j, x)dx\right|^r\\
=& \sum_{\textbf{i} \in \mathcal{I}_i}\mathbbm{E}\left|\int_{R_{\textbf{i}}} \mathbbm{1}(x \in B)\cdot \Delta_i(\pi_k, \pi_j, x)dx\right|^r\\
=&o(1)
\end{split}
\end{align}
using (\ref{equation chang}). 

Also, from existing results  we can verify that
\begin{align}\label{conc 2}
\sum_{\textbf{i} \in \mathcal{I}_i}\mathbbm{E}|u_{_{\textbf{i},N}}|^r \to 0
\end{align}
Therefore, combining  \eqref{cond 1} and \eqref{cond 2}, we have (\ref{cond 2}). This now completes the proof of the Lemma
\end{proof}

We are now ready to prove asymptotic normality 
\begin{lemma}\label{equality }
 Under the regularity conditions, the following holds
 \begin{align*}
     &\lim_{N\to \infty}  \int_{B_0} \Bigg\{\sqrt{N}\mathbbm{E}\left[\left|[(\Gamma_{\mathcal{N}}(x,\pi_k, \pi_j)]-\mathbbm{E}[(\Gamma_{\mathcal{N}}(x,\pi_k, \pi_j)]\right|\right]- \mathbbm{E}|\mathbbm{Z}| \kappa^{1/2}_{\tau, N}(x,\pi_k,\pi_j)\Bigg\}w(x,\pi_j,\pi_k) dx=0
 \end{align*}

 and 
  \begin{align*}
     & \lim_{N\to \infty}  \int_{B_0}\Bigg\{\sqrt{N}\mathbbm{E}\left[\left|[(\Gamma_{\mathcal{N}}(x,\pi_k, \pi_j)]-\mathbbm{E}[(\Gamma_{\mathcal{N}}(x,\pi_k, \pi_j)]\right|\right]-\mathbbm{E}|\mathbbm{Z}_1|\cdot \kappa^{1/2}_{\tau, N}(x,\pi_k,\pi_j)\Bigg\}w(x,\pi_j,\pi_k) dx=0,
  \end{align*}
where $\mathbbm{Z}$  is a standard normal random variable.
\end{lemma} 
\begin{proof}
Using Lemma \ref{Gine}, and similar arguments in the proof of Lemma 6.3 of \cite*{gine2003bm}, the results are established.
\end{proof}

Define 
\begin{align*}
    L_N(B)=& \frac{\sqrt{N}}{\sigma_N(B)}  \int_{B_0} \Bigg\{\left|[(\Gamma_{N}(x,\pi_k, \pi_j)]-\mathbbm{E}[(\Gamma_{N}(x,\pi_k, \pi_j)]\right| \\
    &- \mathbbm{E}\left|[(\Gamma_{N}(x,\pi_k, \pi_j)]-\mathbbm{E}[(\Gamma_{N}(x,\pi_k, \pi_j)]\right|\Bigg\}w(x,\pi_j,\pi_k) dx.
\end{align*}

\begin{lemma}
 Under the regularity conditions, we have 
  $$ L_N(B) \overset{d}{\to} \mathbbm{Z}$$ as $C \to \infty,$
where $\mathbbm{Z}$ is a standard normal random variable.
\end{lemma}

\begin{proof}
Notice that
\begin{align*}
 S_N=& \frac{\sqrt{N}}{\sigma_N(B)} \int_{B}\left\{\sqrt{N}\left|[(\Gamma_{\mathcal{N}}(x,\pi_k, \pi_j)]-\mathbbm{E}[(\Gamma_{\mathcal{N}}(x,\pi_k, \pi_j)]\right|\right\}w(x,\pi_j,\pi_k)\\
 &-\frac{\sqrt{N}}{\sigma_N(B)} \int_{B} \left\{\sqrt{N}\mathbbm{E}\left|[(\Gamma_{\mathcal{N}}(x,\pi_k, \pi_j)]-\mathbbm{E}[(\Gamma_{\mathcal{N}}(x,\pi_k, \pi_j)]\right|\right\}w(x,\pi_j,\pi_k) dx   
\end{align*}

Adopting the de-poissonization arguments of \cite{gine2003bm} , we have 

\begin{align*}
 (S_N|\mathcal{N}=N)\overset{d}{=}&\frac{\sqrt{N}}{\sigma_N(B)} \int_{B} \left\{\sqrt{N}\left|[(\Gamma_{N}(x,\pi_k, \pi_j)]-\mathbbm{E}[(\Gamma_{N}(x,\pi_k, \pi_j)]\right|\right\}w(x,\pi_j,\pi_k)\\
 &-\frac{\sqrt{N}}{\sigma_N(B)} \int_{B}  \left\{\sqrt{N}\mathbbm{E}\left|[(\Gamma_{N}(x,\pi_k, \pi_j)]-\mathbbm{E}[(\Gamma_{N}(x,\pi_k, \pi_j)]\right|\right\}w(x,\pi_j,\pi_k) dx\\
 \to& \mathbbm{Z}  
\end{align*}

Now from  Lemma \ref{equality }, we know that 
\begin{align*}
&\lim_{C\to \infty}  \int_{B}\left\{\sqrt{N}\mathbbm{E}\left|[(\Gamma_{\mathcal{N}}(x,\pi_k, \pi_j)]-\mathbbm{E}[(\Gamma_{\mathcal{N}}(x,\pi_k, \pi_j)]\right|\right\}w(x,\pi_j,\pi_k) dx \\
-& \int_{B}  \left\{\sqrt{N}\mathbbm{E}\left|[(\Gamma_{\mathcal{N}}(x,\pi_k, \pi_j)]-\mathbbm{E}[(\Gamma_{\mathcal{N}}(x,\pi_k, \pi_j)]\right|\right\}w(x,\pi_j,\pi_k) dx = 0.
\end{align*}
Hence, we have 
$$ L_N(B) \overset{d}{\to} \mathbbm{Z}$$ as $C \to \infty,$ as required
\end{proof}

We are now ready to prove Theorem \ref{main theorem}.
We can show that for any $\pi_k \in \Pi,$
\begin{align} \label{Chang 2}
\begin{split}
    &\lim_{C\to \infty}\sup \mathbbm{E}\left(\sqrt{N_k}\int_{B_l^c}\{|\tau_{N_k}(x, \pi_k)-\mathbbm{E}\tau_{N_k}(x, \pi_k)|- \mathbbm{E}[|\tau_{N_k}(x, \pi_k)-\mathbbm{E}\tau_{N_k}(x, \pi_k)|]\}w(x,\pi_k)dx\right)^2\\
    &\leq C_2\lambda(\mathcal{X})\left( \sup_{x \in\mathcal{X} }\mathbbm{E}[|Y^2\hat{\phi}(\pi_k, x, T)^2||X=x]+E[|\chi(\pi_k, x,T)^2|]\right)\int_{B_l^c}f(x)dx,
\end{split}
\end{align}
where $C_2$ is a positive constant and $\{B_l:l>1\}$ is a sequence of Borel sets in $\mathbbm{R}^d$ that has a finite Lebesgue measure $\lambda(B_l)$ and satisfies \eqref{M1}-\eqref{M3} with $B_l=B$ for each $l$ and let 
\begin{align} \label{prere 1}
    \lim_{l\to \infty}\int_{B^c_l}f(x) dx =0.
\end{align}
Also,  for each $l \geq 1$, by Lemma \ref{asymptotic variance}, we have
\begin{align} \label{prere 2}
    \lim_{l\to \infty}\sigma^2_{1, B_l} = \sigma^2_1.
\end{align}

Following the proof of Theorem B.2 in \cite{chang2015nonparametric}, we can show that 
\footnotesize
\begin{align*}
    &\lim_{C\to \infty}\sup \mathbbm{E}\Bigg(\sqrt{N}\int_{B_l^c} \Bigg\{\left|[(\Gamma_{\mathcal{N}}(x,\pi_k, \pi_j)]-\mathbbm{E}[(\Gamma_{\mathcal{N}}(x,\pi_k, \pi_j)]\right|-\mathbbm{E}\left|[(\Gamma_{\mathcal{N}}(x,\pi_k, \pi_j)]-\mathbbm{E}[(\Gamma_{\mathcal{N}}(x,\pi_k, \pi_j)]\right|\Bigg\}w(x,\pi_j,\pi_k) dx\Bigg)^2\\ 
    &\leq  C_2\lambda(\mathcal{X})\cdot\Bigg\{\left( \sup_{x \in\mathcal{X} }\mathbbm{E}[|Y^2\hat{\phi}(\pi_k, x, T)^2||X=x]+E[|\chi(\pi_k, x,T)^2|]\right)\Bigg\}\cdot\int_{B_l^c}f(x)dx.\\   
\end{align*}
\normalsize

Using this result, with the results in \eqref{prere 1}--\eqref{prere 2} and Theorem 4.2 in \cite{billingsley1968convergence}, we conclude that 
$$
 \int_{\mathcal{X}}  \left\{\sqrt{N}\left|[\tau_N(x,\pi_k)-\tau_N(x,\pi_j)]-\mathbbm{E}[\tau_N(x,\pi_k)-\tau_N(x,\pi_j)]\right|\right\}w(x,\pi_j,\pi_k) dx \overset{d}{\to} \sigma_0 \mathbbm{Z}.
$$
The proof is complete because we can use Lemma (\ref{equality })
to show that
$$\lim_{C\to \infty}  \int_{B_0}\left\{\sqrt{N}\mathbbm{E}\left[\left|[(\Gamma_{\mathcal{N}}(x,\pi_k, \pi_j)]w(x,\pi_j,\pi_k)-\mathbbm{E}[(\Gamma_{\mathcal{N}}(x,\pi_k, \pi_j)]\right|\right]- a_i\right\} dx=0.$$

\end{alphasection}

\end{document}